\documentclass[11pt]{book}

\usepackage{calc}
\usepackage[latin1]{inputenc}
\usepackage[frenchb,english]{babel} 
\usepackage{graphicx}
\usepackage{amsmath,amsfonts,amssymb,amsthm}
\usepackage{fancyhdr}
\usepackage{enumerate}
\usepackage{makeidx}
\usepackage{calligra}

\newtheorem{thm}{Theorem}[chapter]
\newtheorem{cor}[thm]{Corollary}
\newtheorem{lemma}[thm]{Lemma}
\newtheorem{prop}[thm]{Proposition}
\theoremstyle{definition}
\newtheorem{defn}[thm]{Definition}
\theoremstyle{remark}
\newtheorem{remark}[thm]{Remark}

\newcommand{\norm}[1]{\left\Vert {#1} \right\Vert}
\newcommand{\abs}[1]{\left\vert {#1} \right\vert}
\newcommand{\set}[1]{\left\{ {#1} \right\}}
\newcommand{\prt}[1]{\left( {#1} \right)}
\newcommand{\parenthese}[1]{\left( {#1} \right)}
\newcommand{\scal}[1]{\left< {#1} \right>}

\newcommand{\dv}[2]{\frac{\mathrm d{#1}}{\mathrm d{#2}}}
\newcommand{\pd}[2]{\frac{\partial{#1}}{\partial{#2}}}
\newcommand{\eq}{\ =\ }
\newcommand{\setN}{{\mathbb N}}              
\newcommand{\setZ}{{\mathbb Z}}
\newcommand{\setR}{{\mathbb R}}
\newcommand{\setC}{{\mathbb C}}

\newcommand{\sR}{{\mathbb R}}
\newcommand{\sC}{{\mathbb C}}
\newcommand{\implie}{\Rightarrow}
\newcommand{\limplie}{\ \ \Longrightarrow\ \ }

\newcommand{\equi}{\Leftrightarrow}
\newcommand{\lequi}{\ \ \Longleftrightarrow\ \ }
\newcommand{\Lequi}{\Longleftrightarrow}
\newcommand{\precc}{\prec\!\!\prec}

\newcommand{\A}{\mathcal{A}}
\newcommand{\B}{\mathcal{B}}
\newcommand{\D}{\mathcal{D}}
\newcommand{\F}{\mathcal{F}}
\newcommand{\G}{\mathcal{G}}
\renewcommand{\H}{\mathcal{H}}
\newcommand{\I}{\mathcal{I}}
\newcommand{\J}{\mathcal{J}}
\renewcommand{\L}{\mathcal{L}}
\newcommand{\M}{\mathcal{M}}
\newcommand{\N}{\mathcal{N}}
\renewcommand{\P}{\mathcal{P}}
\newcommand{\K}{\mathcal{K}}
\newcommand{\T}{\mathcal{T}}
\newcommand{\C}{\mathcal{C}}
\newcommand{\R}{\mathcal{R}}
\newcommand{\X}{\mathcal{X}}
\newcommand{\U}{\mathcal{U}}

\newcommand{\Cl}{\mathbb{C}\text{l}}
\newcommand{\bA}{{\bf A}}

\DeclareMathOperator{\tr}{tr}

\usepackage{caption}
\fancypagestyle{plain}{
\fancyhead{}
}

\makeatletter \def\cleardoublepage{\clearpage\if@twoside \ifodd\c@page\else
\hbox{} \vspace*{\fill}  \thispagestyle{empty}
\newpage
\if@twocolumn\hbox{}\newpage\fi\fi\fi} \makeatother

\makeindex
\begin{document}
\frontmatter

\begin{titlepage}

\noindent
\makeatletter
\includegraphics[height=.2\textheight]{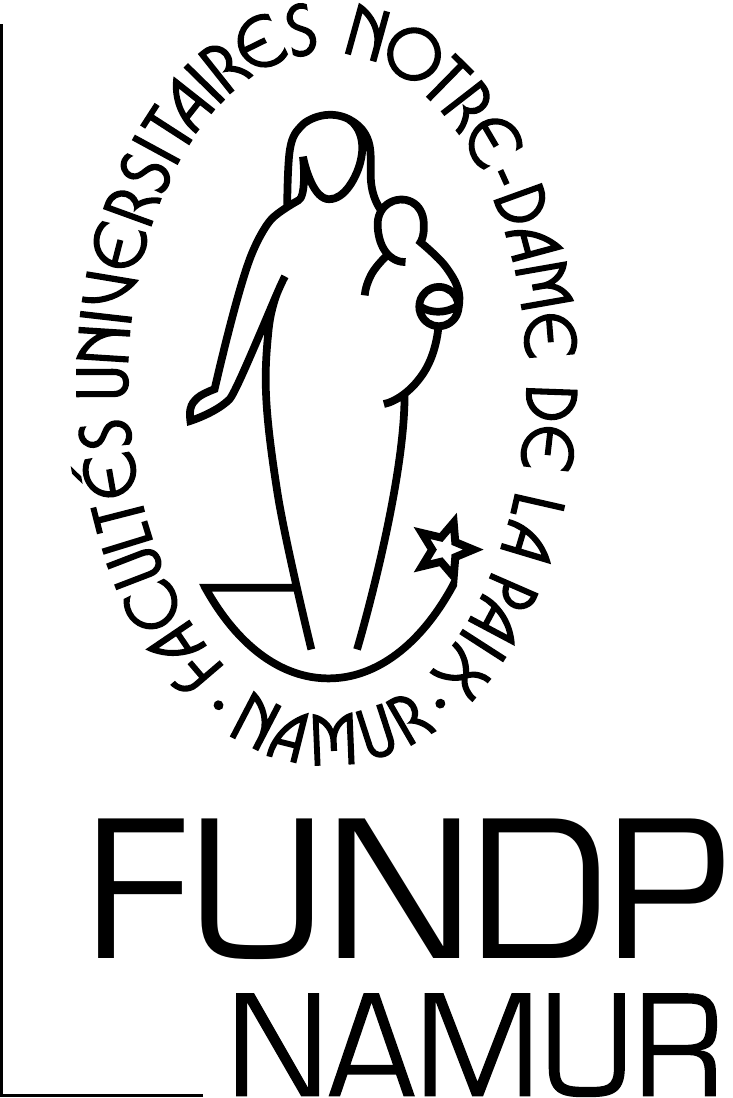}

\makeatother

\vspace{\stretch{1}}

\begin{center}
\large Facult\'{e}s Universitaires Notre-Dame de la Paix\\
\large Facult\'{e} des Sciences -- D\'{e}partement de Math\'{e}matique
\end{center}

\vspace{\stretch{2}}

\begin{center}
\usefont{T1}{ppl}{b}{n} \Huge Lorentzian approach to noncommutative geometry
\end{center}

\vspace{\stretch{3}}

\begin{flushright}
  Thèse pr\'{e}sent\'{e}e par\\
  \textbf{
Nicolas Franco
  }\\
  en vue de l'obtention du grade\\
  de Docteur en Sciences
\end{flushright}

\vspace{\stretch{3}}

\noindent%
\textbf{Composition du jury} :

\smallskip

\noindent%
Timoteo {Carletti} (pr\'esident)\\
Andr\'{e} {F\"uzfa}\\
Marc {Lachièze-Rey}\\
Dominique {Lambert} (promoteur)\\
Anne {Lema\^{i}tre} (promoteur)\\
Pierre {Martinetti}\\

\vspace{\stretch{1}}

\begin{center}
Août 2011
\end{center}

\vspace{\stretch{0}}

\end{titlepage}


\clearpage

\thispagestyle{empty}

\selectlanguage{frenchb}

\vfill
\vspace*{\stretch{5}}

\begin{center}
Facult\'{e}s Universitaires Notre-Dame de la Paix\\
Facult\'{e} des Sciences -- D\'epartement de Math\'ematique\\
Rue de Bruxelles 61, B-5000 Namur, Belgium
\end{center}

\selectlanguage{english}



\cleardoublepage

\selectlanguage{english}
\addcontentsline{toc}{chapter}{Abstract}
\thispagestyle{empty} 

\vspace*{-2cm}

\begin{center}
  { \usefont{T1}{ppl}{b}{n} \Large  Lorentzian approach to noncommutative geometry}\\  \large Nicolas Franco
\end{center}

\paragraph*{Abstract}
This thesis concerns the research on a Lorentzian generalization of Alain Connes' noncommutative geometry. In the first chapter, we present an introduction to noncommutative geometry within the context of unification theories. The second chapter is dedicated to the basic elements of noncommutative geometry as the noncommutative integral, the Riemannian distance function and spectral triples. In the last chapter, we investigate the problem of the generalization to Lorentzian manifolds. We present a first step of generalization of the distance function with the use of a global timelike eikonal condition. Then we set the first axioms of a temporal Lorentzian spectral triple as a generalization of a pseudo-Riemannian spectral triple together with a notion of global time in noncommutative geometry.

\bigskip

\noindent%
Ph.D. thesis in Mathematics
\bigskip

\selectlanguage{frenchb}

\begin{center}
  {\usefont{T1}{ppl}{b}{n} \Large Approche Lorentzienne en g\'eom\'etrie noncommutative}\\  \large  Nicolas Franco
\end{center}

\paragraph*{R\'{e}sum\'{e}}
Le sujet de cette thèse est la recherche d'une généralisation Lorentzienne de la géométrie noncommutative d'Alain Connes. Dans le premier chapitre, nous présentons une introduction à la géométrie noncommutative dans le contexte des théories d'unification. Le second chapitre est dédié aux éléments de base de la géométrie noncommutative, comme l'intégrale noncommutative, la fonction de distance Riemannienne et les triplets spectraux. Dans le dernier chapitre, nous explorons le problème de la généralisation aux variétés Lorentziennes. Nous présentons une première étape de généralisation de la fonction de distance basée sur une condition eikonale globale de type temps. Ensuite, nous fixons les premiers axiomes d'un triplet spectral Lorentzien temporel, représentant une généralisation d'un triplet spectral pseudo-Riemannien muni d'une notion de temps global en géométrie noncommutative.


\bigskip

\noindent%
Dissertation doctorale en Sciences (orientation math\'ematique)

\bigskip

\noindent
Date: 31 août 2011\\
Promoteurs (advisors): Pr D. Lambert et Pr A. Lemaître

\selectlanguage{english}

\vfill


\cleardoublepage
\vspace*{3cm}
\section*{Remerciements\markboth{Remerciements}{Remerciements}}
\addcontentsline{toc}{chapter}{Remerciements}
\selectlanguage{frenchb}
\thispagestyle{empty} 

Je remercie chaleureusement mes promoteurs, Dominique Lambert et Anne Lemaître, qui m'ont fortement soutenu pendant toute la durée de ma thèse. Tout particulièrement, je leur suis reconnaissant de m'avoir laissé une très grande liberté dans mes pérégrinations mathématiques m'aillant mené jusqu'à cette magnifique théorie qu'est la géométrie noncommutative, ainsi que de m'avoir laissé la possibilité de mener une carrière scientifique simultanément à une carrière musicale.\\

Je remercie mes différents collègues des facultés, et en particulier André Füzfa pour les nombreuses discussions très enrichissantes et pour m'avoir bien souvent ouvert les portes du monde des physiciens.\\

Je n'oublie pas ma famille, et surtout mon épouse qui a dû supporter mes longues absences durant les mois de rédaction, de même que mes deux petits bouts qui ont vu le jour durant cette expérience.\\

Ce travail a pu être effectué grâce à un mandat du F.R.S.-FNRS ainsi que d'un financement complémentaire des FUNDP.

\selectlanguage{english}


\tableofcontents
\addcontentsline{toc}{chapter}{Contents}


\chapter*{Introduction\markboth{Introduction}{Introduction}}
\addcontentsline{toc}{chapter}{Introduction}

Geometry and algebra are often considered as two distinct branches of mathematics. If one needs a proof, it is just sufficient to read the titles of general books or mathematical lessons on these topics. However, this common interpretation is incorrect, since those domains are strongly related to each other. One of the first fusion between geometry and algebra came from Ren\'e Descartes, who can be considered as the father of the algebraization of geometry. From that time onwards, many correspondences and influences between them have been developed, with a quite complete duality between geometrical spaces and commutative algebras among the outcomes. Noncommutative geometry is the extension of this duality to the noncommutative world.\\

We can recover a similar distinction in physics of fundamental interactions. The oldest known fundamental interaction is gravitation, while the three others -- electromagnetism, strong and weak interactions -- are far more recent. Each of those physical interactions has a mathematical background in which it can be expressed. Gravitation is clearly based on geometrical elements while the others need the introduction of algebraic theories. Nevertheless, despite the strong relations between geometry and algebra, the discovery of a common mathematical background to all those interactions is a puzzle, and this constitutes one of the great challenges of the current research in physics.\\

Alain Connes has developed from many years now a theory of noncommutative geometry combining geometrical and algebraic aspects which could be a good candidate for such common background. From this point of view, noncommutative geometry can be seen as an outsider to string theories, but which is unfortunately far less widespread than the last ones. Moreover, this theory is still at an early stage of development, with many unanswered questions and remaining problems.\\

This theory provides a mathematical structure supporting at the same time Euclidean gravity and a classical standard model. This result is very interesting on its own and it deserves to be developed and studied in details. However, this model concerns only at this time Euclidean gravity, which means gravity with a positive signature, based on Riemannian geometry. Since gravitation is entirely based on Lorentzian gravity, with a signature of type \mbox{$(-,+++)$}, such model does not correspond to any physical reality. So the theory of noncommutative geometry should be considered mainly at a mathematical level, at least for its gravitational part, and for which further important developments are still needed in order to make it a physical one.\\

The question of the lack of a complete Lorentzian formulation of the theory is too often laid on the table, mainly for the reason of favoring the development of the still complicated Riemannian case, and this is why we have dedicated our research to this problem \cite{F1,F2,F3}. This dissertation consists of a general introduction to noncommutative geometry together with a presentation of the current process of generalization to the Lorentzian case. It is divided into three important chapters with the following structure.\\

The Chapter \ref{chapterintro} is a general introduction to some mathematical frameworks developed in the purpose of unifying physics theories. We talk about noncommutative geometry in general and also quantum gravity. The main idea of this chapter is to present different ways to merge geometry and algebra in the context of unification theories. Instead of a technical prolegomenon, we propose a walk among mathematical theories, scattered with useful definitions and some highlighted developments. In particular, we take the time to present the complete proof of Gel'fand theorem, which is the cornerstone of noncommutative geometry. This introduction constitutes actually our personal approach to the fundamental question of finding a suitable mathematical background for unification theories.\\

The Chapter \ref{chaprie} is the presentation of Alain Connes' theory of noncommutative geometry. We develop the construction of the differential structure of noncommutative geometry in the case of a compact Riemannian manifold, with a special attention to the distance function. The main axioms of spectral triples, the basic structures of noncommutative geometry, are given. We conclude this chapter with a review of the construction of the standard model of particle physics in the framework of noncommutative geometry.\\

The Chapter \ref{chaplor} is the longest chapter of this dissertation, where we consider the question of the generalization of the theory to Lorentzian manifolds. An almost complete review of the existing literature on the subject is presented and discussed, while new problems are analyzed. Then a detailed presentation of our contributions is given. In a first time, we present a first step of generalization of the distance function to the Lorentzian case. This construction was presented in \cite{F3}, also with a conceptual presentation in \cite{F2}. In a second time, we present some unpublished works about the research of a causal counterpart to spectral triples. In particular, the first axioms of temporal Lorentzian spectral triples are given and discussed, with the introduction of a notion of global time in noncommutative geometry.\\

Although we paid attention to be self-consistent while giving definitions and properties, it was not possible to provide an introduction on every mathematical or physical concepts. In particular, usual concepts about topology and functional analysis, especially those concerning Hilbert spaces, are assumed to be known. A good introduction to these topics can be found in \cite{DM}.\\

\cleardoublepage
 \vspace*{\fill}
 \selectlanguage{frenchb}
\begin{center}
\noindent { \Large\calligra "Mais je ne m'areste point a expliquer cecy plus en detail, a cause que je vous osterois le plaisir de l'apprendre de vous mesme, \& l'utilit\'e de cultiver vostre esprit en vous y exerceant, qui est a mon avis la principale, qu'on puisse tirer de cete science.}\\
\end{center}
\begin{center}
\noindent {\Large\calligra Aussy que je n'y remarque rien de si difficile, que ceux qui seront un peu vers\'es en la Geometrie commune, \& en l'Algebre, \& qui prendront garde a tout ce qui est en ce trait\'e, ne puissent trouver."}\\ 
\end{center}
  \vspace{2cm}
\noindent {R. Descartes}, {\it La Géométrie}
   \vspace{2cm}
   \selectlanguage{english}
   \thispagestyle{empty}  \vspace*{\fill}


\cleardoublepage
\hbox{} \vspace*{\fill} \thispagestyle{empty}
\mainmatter
\chapter[Algebraization of geometry]{Unification theories as algebraization of geometry}\label{chapterintro}

We begin our dissertation by introducing noncommutative geometry from a large point of view, and positioning the theory in the framework of unification theories. We will start by considering the problem of unifying current physics theories as a motivation to the development of new mathematical tools. The main point of this first chapter will be the correlation between geometrical theories and algebraic ones.\\


\section[Combining geometry and algebra]{Combining geometry and algebra\vspace{0.1cm}\\ (a quick review of current physical theories and mathematics behind)}

The main goal of any physicist is to discover and test abstract theories that can describe and predict natural phenomena. Those theories are based on mathematical formalisms, so mathematicians and physicists meet each other in the sense that the first ones must produce mathematical tools and theories which can be useful for the second ones. Of course this intersection is not the only possibility of research, since there are so many fields of research in mathematics which will probably never find any application in physics, and in the same way there exist some research fields in physics for which no suitable mathematical tools are available at least at the present time. So we can see that researches in physics and mathematics are strongly dependent on each other, and it is not a waste of time to think about which mathematical fields could be developed further in order to meet the research interests of physicists.\\

We have mentioned that physicists are interested in the description of natural phenomena, but the physicists are more ambitious than that. Their dream is not to set many theories describing all phenomena in the universe but to set {\it one} theory which could explain those phenomena. So when they have two different theories working in two different fields, the next step is nothing but to find a new theory that combines both.\\

This final dream has a name, the {\it theory of everything}, a unique theory which could involve all known fundamental interactions. And so far these interactions are in number of four:
\begin{itemize}
\item Gravitation, the weakest of all interactions which is only attractive and depends on massive elements
\item Electromagnetism, acting between charged particles
\item The strong interaction, which insures atomic and nuclear cohesion
\item The weak interaction, a weaker nuclear interaction between neutrinos, leptons and quarks\\
\end{itemize}

Three of these interactions -- electromagnetism, strong and week interaction -- can at this time be described in a single unified way, thanks to gauge theories, but gravitation remains the worst student.\\

We will devote this section to a quick overview of current physics theories about fundamental interactions, and we will underline mathematical tools that are used behind. This section will give us the opportunity to introduce a good number of mathematical notions which will be very useful for the remaining of our dissertation.\\

\subsection{General relativity}

We first begin with the gravitational force. Up this day, the best physical theory describing gravitation is Einstein's general relativity, which is mainly based mathematically on pseudo-Riemannian geometry. We need to introduce some preliminary notions.

\begin{defn}
A ($n$-dimensional) {\bf\index{manifold} manifold} $\M$ is a second countable Hausdorff space which is locally homeomorphic to $\setR^n$, i.e.~for each point $p\in\M$ there exists an open $U_p\subset\M$ and a homeomorphism \mbox{$\varphi_p : U_p \rightarrow \setR^n$} called chart. A collection $\set{(U_\alpha,\varphi_\alpha)}$ of charts covering $\M$ is an atlas and the functions between euclidian spaces \mbox{$\varphi_{\alpha\beta} = \varphi_\alpha \circ \varphi_\beta^{-1}$} are the transition maps. The manifold is said {\bf\index{manifold!differentiable} differentiable} if the transition maps are $C^k$, and by this way allow to define $C^k$ functions on $\M$.
\end{defn}

For more simplicity, we will only use $C^\infty$ differential manifolds.

\begin{defn}
A {\bf\index{tangent vector} tangent vector} at the point $p \in \M$ is a map \mbox{$v_p : C^\infty(\M) \rightarrow \sR$} such that:
\begin{itemize}
\item $v_p(f+g) =v_p(f) + v_p(g)$
\item $v_p(\alpha f) =\alpha\,v_p(f) \qquad \qquad \qquad \qquad \forall\alpha\in\setR,\; \forall f,g\in C^\infty(\M)$
\item $v_p(fg) =v_p(f)\,g(p) + f(p)\,v_p(g)$
\end{itemize}
The vector space of all tangent vectors at one point is the tangent space $T_p(\M)$, and the assignments of a tangent vector to each point of $\M$ are {\bf\index{vector field} vector fields}, so smooth vector fields are actually operators \mbox{$C^\infty(\M) \rightarrow C^\infty(\M)$}. The products of vector fields (in sense of composition) are not vector fields, but the commutator \mbox{$[u,v] = uv - vu$} is a vector field.\end{defn}

Tangent vectors can be seen as directional derivatives, so if $\set{x^\mu}$, $\mu = 1,...,n$, are local coordinates on $\M$, a vector field can be written as \mbox{$X = X^\mu \pd{}{x^\mu}$} with $X^\mu \in C^\infty(\M)$ and with $\set{\pd{}{x^\mu}}$ the natural basis usually noted $\set{\partial_{\mu}}$, where we use Einstein's summation convention on pairs of indices.

\begin{defn}\label{defkform}
A {\bf\index{differential form} differential k-form} is a field giving at each point $p\in\M$ an antisymmetric multilinear map:
$$ \omega_p : \bigotimes_k T_p(\M) \rightarrow \setR.$$
The space of differential k-forms is noted $\Omega^k(\M)$ and is a vector space. The space of all differential forms \mbox{$\Omega(\M) = \bigoplus_{k=0}^\infty \Omega^k(\M)$} is endowed with an anticommutative exterior product \mbox{$\wedge : \Omega^p(\M) \times \Omega^q(\M) \rightarrow \Omega^{p+q}(\M)$} and an exterior derivative \mbox{$d : \Omega^p(\M) \rightarrow \Omega^{p+1}(\M)$}.
\end{defn}

The elements of $\Omega^1(\M)$ are fields of linear forms defined on tangent spaces $T_p(\M)$, and are called {\bf\index{vector field!covariant} covariant vector fields}. The vector space of covariant vectors on one point $p$ is the cotangent space $T^*_p(\M)$, which is dual to the tangent space $T_p(\M)$. Its natural base is composed  of the differentials $\set{d{x^\mu}}$, so a covariant vector field is of the form \mbox{$w = X_\mu dx^\mu$, $X_\mu \in C^\infty(\M)$}. By use of the exterior product we can construct a base for $\Omega^k(\M)$ with elements \mbox{$dx^\mu \wedge dx^\nu \wedge ... \wedge dx^\lambda$}.

\begin{defn}
A {\bf\index{tensor} tensor} at the point $p \in \M$ is a multilinear map
$$T : \prt{\bigotimes_r T_p(\M)} \otimes \prt{\bigotimes_s T^*_p(\M)} \rightarrow \setR.$$
Tensor fields can be represented in terms of coordinates:
$$T = T_{\nu_1...\nu_s}^{\mu_1...\mu_r}\,\partial_{x^{\mu_1}}\otimes...\otimes\partial_{x^{\mu_r}}\otimes dx^{\nu_1}\otimes...\otimes dx^{\nu_s}$$
where $\otimes$ is the tensor product and $T_{\nu_1...\nu_s}^{\mu_1...\mu_r} \in C^\infty(\M)$, and the application of a tensor fields to any vector and covariant vector fields can be expressed in terms of the coordinates only:
$$T(v_1,...,v_s,w_1,...,w_r) = T_{\nu_1...\nu_s}^{\mu_1...\mu_r}\, {v_1}^{\nu_1}...{v_s}^{\nu_s}{w_1}_{\mu_1}...{w_r}_{\mu_r}.$$
\end{defn}

\begin{defn}
A {\bf\index{metric} metric} is a specific symmetric tensor field $g$ of rank two  with components $g_{\mu\nu}$ giving at each point $p\in\M$ a scalar product between two tangent vectors $g_p : T_p(\M) \times T_p(\M) \rightarrow \setR$.  The metric must be non-degenerate, in the sense that all proper values must not be zero. The dual metric is an inner product between covariant vectors given by $g^{\mu\nu}g_{\nu\lambda} = \delta^\mu_\lambda$. The {\bf\index{metric!signature} signature} of the metric is the number $m-q$ (or the couple $(m,q)$) where $m$ is the number of positive proper values and $q$ the number of negative proper values. If the signature is equal to the dimension $n$, the metric is said to be {\bf\index{metric!Riemannian}\index{manifold!Riemannian} Riemannian}, and otherwise {\bf\index{metric!pseudo-Riemannian}\index{manifold!pseudo-Riemannian} pseudo-Riemannian}. In the particular case where the signature is $n-2$ (only one negative proper value), the metric is said to be {\bf\index{metric!Lorentzian}\index{manifold!Lorentzian} Lorentzian}.
\end{defn}

The mathematical framework of general relativity is a Lorentzian 4-dimensional manifold, so a smooth 4-dimensional manifold with a Lorentzian metric of signature $(3,1)$. This kind of manifold is called {\bf\index{spacetime} spacetime} and the points on it are called {\bf\index{event} events}, and represent a unique spatial position at a unique local time.\\

The metric of the spacetime fixes the dynamics of the system, as particles have their inertial motion given by geodesics. The {\bf\index{geodesic} geodesics} are curves on $\M$ with extremum length, so they are functions \mbox{$x : I \subset\setR \rightarrow \M$} that are extrema of the action:
$$l(x) = \int_I \sqrt{\abs{g_{\mu\nu} \dot x^\mu \dot x^\nu}}\; dt$$
and can be determined by the Euler--Lagrange equation:
$$\ddot x^\lambda + \Gamma^\lambda_{\mu\nu} \,\dot x^\mu \,\dot x^\nu = 0$$
with
$$\Gamma^\lambda_{\mu\nu}  = \frac 12 \,g^{\lambda\kappa} \prt{\partial_{x^\mu} g_{\nu\kappa} + \partial_{x^\nu} g_{\mu\kappa} - \partial_{x^\kappa} g_{\mu\nu}}$$
called the Levi--Civita connection.\\

Of course the choice of the metric cannot be free, as there are some constraints on the curvature of the spacetime.

\begin{defn}
The {\bf\index{tensor!Riemann} Riemann tensor} is the tensor of rank four given by:
$$R^\kappa_{\lambda\mu\nu} = \partial_{x^\mu} \Gamma^{\kappa}_{\nu\lambda} - \partial_{x^\nu} \Gamma^{\kappa}_{\mu\lambda} + \Gamma^{\eta}_{\nu\lambda} \Gamma^{\kappa}_{\mu\eta} - \Gamma^{\eta}_{\mu\lambda} \Gamma^{\kappa}_{\nu\eta}.$$
The successive contractions of the Riemann tensor give the {\bf\index{tensor!Ricci} Ricci tensor}:
$$\text{Ric}_{\mu\nu} = R^\lambda_{\mu\lambda\nu}$$
and the {\bf\index{curvature!scalar} scalar curvature}:
$$R = g^{\mu\nu}\text{Ric}_{\mu\nu}.$$
\end{defn}

\begin{defn}
The {\bf\index{tensor!Einstein} Einstein tensor} is the symmetric and divergenceless tensor of rank two:
$$G_{\mu\nu} = \text{Ric}_{\mu\nu} - \frac 12 g_{\mu\nu} R.$$
\end{defn}

The equations of general relativity giving constraints on the metric are, in case of absence of matter:
\begin{equation}\label{einstvide}
G_{\mu\nu} = 0
\end{equation}
and with presence of matter:
\begin{equation}\label{einstnonvide}
G_{\mu\nu} + \Lambda g_{\mu\nu}= \frac{8\pi G}{c^4} T_{\mu\nu}
\end{equation}
where $T_{\mu\nu}$ is the energy-momentum tensor describing the density and flux of energy and momentum, and where $\Lambda$ is the cosmological constant.\\

The equations (\ref{einstvide}) can be expressed in a Lagrangian formalism as the extremum of the Einstein--Hilbert action:
$$S = \frac{c^4}{16\pi G} \int  R \sqrt{\abs{\det g}} \;  d^4x$$
and similarly for (\ref{einstnonvide}) by adding a matter action.\\

Einstein's equations (\ref{einstnonvide}) show us that the gravitational force is strongly related to geometry. The meaning of these equations is that the mass, or more generally the energy, influences the geometry of spacetime, while in the same time the geometry of spacetime dictates the inertial movement of massive and even massless particles. The deflection of light is one of the best examples of the geometrical aspect of gravitation.\\

\subsection{Quantum theories}

Now let us switch to the other current physical theories. While gravitation is a really good theory to describe interactions at large scales, physics at small scales is the domain of quantum theories.\\

Quantum theories arise from the quantization of classical theories as classical mechanics or fields theories. The usual variables (coordinates, fields) are replaced by operators acting on a particular Hilbert space, whose elements are the states of the system. The easiest way to introduce a procedure of quantization is by the method of canonical quantization, whose we give a quick review here.

\begin{defn}
An {\bf\index{algebra} associative algebra $\A$} is a vector space with an associative and distributive product \mbox{$\A \times \A \rightarrow \A : (a,b) \leadsto ab$} (a ring structure). If the product is commutative, the algebra is said to be commutative or abelian. \end{defn}

\begin{defn}
A {\bf\index{algebra!Lie} Lie algebra $\A$} is a vector space with a binary operation \mbox{$\A \times \A \rightarrow \A : (a,b) \leadsto [a,b]$} which is bilinear and antisymmetric and respects the Jacobi identity \mbox{$[x,[y,z]] + [y,[z,x]] + [z,[x,y]] = 0$} $\forall x, y, z \in\A$.
\end{defn}

Any (noncommutative) algebra defines trivially a Lie algebra by use of the commutator \mbox{$[a,b] = ab - ba$}. As we have seen, the commutator of vector fields also defines a Lie algebra structure. Likewise, the classical mechanics in its Hamiltonian form can be expressed in term of a Lie algebra.\\

Hamiltonian mechanics uses a phase space $\Gamma = (q^i,p_i)$ where $q^i$ are (generalized) coordinates and $p_i = \pd{L}{\dot q^i}$ are (generalized) momenta derived from the Lagrangian function $L$. The Hamiltonian function $H = \dot q^i p_i - L$ introduces the relations $\dot q = \pd{H}{p}$ and $\dot p = - \pd{H}{q}$.

\begin{defn}
We define the {\bf\index{Poisson bracket} Poisson bracket} between two functions $f$ and $g$ by:
$$\set{f,g}=\sum \pd{f}{q^i}\pd{g}{p_i} - \pd{f}{p_i}\pd{g}{q^i}.\vspace{0.3cm}$$
\end{defn}

The Poisson bracket gives a structure of Lie algebra with an additional Leibniz rule $\set{a,bc} = \set{a,b}c + b\set{a,c}$, which is called a \mbox{\bf\index{algebra!Poisson}Poisson algebra}. The whole dynamics of the system can be expressed in terms of the Poisson bracket. The canonical relations between the variables are given by:
$$
\set{q^i,q^j} \eq 0, \quad \set{p_i,p_j} \eq 0, \quad \set{q^i,p_j} \eq \delta^i_j
$$
and the evolution of any function $f$ not explicitly time-dependent is given by:
$$\dv{}{t} f \eq \set{f,H}.$$

The main idea of canonical quantization is to replace the phase space $\Gamma$ with the Poisson algebra structure by a Hilbert space $\H$ with a Lie algebra structure given by the commutator of operators on $\H$. Observables are represented by Hermitian operators on $\H$ (the old variables $q^i$ and $p_i$ become some position and momentum operators $\hat x^i$ and $\hat p_i$). The possible values of an observable are given by eigenvalues, with eigenvectors being the possible states of the system.\\

To summary, we have the following correspondences:
$$
\begin{array}{ccc}
\Gamma &\leadsto &\H\\[4pt]
H &\leadsto &\hat H\\[4pt]
q^i &\leadsto &\hat{x}^i\\[4pt]
p_i &\leadsto &\hat{p}_i\\[4pt]
\set{\ ,\ } &\leadsto &-\frac{i}{\hbar}[\ \,,\ ]
\end{array}
$$
with $\hbar$ being the reduced Planck constant, and the canonical relations:
$$
[\hat{x}^i,\hat{x}^j] \eq 0, \quad [\hat{p}_i,\hat{p}_j] \eq 0, \quad [\hat{x}^i,\hat{p}_j] \eq i \hbar\; \delta^i_j.
$$
The evolution of any operator $A$ not explicitly time dependent is given by the Heisenberg equation:
$$
\dv{}{t} A \eq -\frac{i}{\hbar} [A, \hat H].\vspace{0.3cm}
$$

We can notice that the noncommutativity of the algebra is mandatory in order to guarantee nontrivial evolution equations.\\

The same procedure can be applied to quantify fields theory (which is sometimes called second quantization) in order to lead to quantum field theory, and especially to quantum electrodynamics, which gives a quantum theory of the electromagnetic interaction. Quantum electrodynamics is a particular case of gauge theories, with a commutative gauge group. We will now spend some time to introduce notions and concepts of gauge theories.\\

\subsection{Gauge theories}

Quantum mechanics is a theory which is invariant under a global symmetry, a phase change $e^{i\theta}$ on states. Gauge theories extend this concept to local invariances (given by gauge transformations), and to more complex symmetries given by different Lie groups.

\begin{defn}
A {\bf\index{Lie group} Lie group $G$} is a group which is also a smooth manifold and such that multiplication and inversion are smooth maps.
\end{defn}

Typical Lie groups are:
\begin{itemize}
\item $O(n)$ the group of $n\times n$ real orthogonal matrices
\item $SO(n)$ the subgroup of  $O(n)$ of matrices with determinant $1$
\item $U(n)$ the group of $n\times n$ complex unitary matrices
\item $SU(n)$ the subgroup of  $U(n)$ of matrices with determinant $1$
\end{itemize}

Since Lie groups are manifolds, we can consider the tangent space at the identity element, with a natural structure of Lie algebra. This is the Lie algebra associated to the Lie group.

\begin{defn}
A {\bf\index{fiber bundle} fiber bundle} is a quadruple \mbox{$(E,\M,\pi,F)$} where $E$ is a manifold called the total space, $\M$ a manifold called the base, $\pi : E \rightarrow \M$ a projection which is surjective and continuous and $F$ a fiber such that $\forall p\in \M$ $\pi^{-1}(p) \cong F$ (where $\cong$ denotes an isomorphism).
\end{defn}

A {\bf\index{section} section} on a fiber bundle is a continuous $C^\infty$ function $s : \M \rightarrow E $ such that $\forall p\in \M$, $s(p)\in \pi^{-1}(p)$. The space of sections on the fiber bundle $E$  is often denoted by $\Gamma(E)$, and a fiber localized at a particular point is denoted by $F_p = \pi^{-1}(p)$.\\

There are particular cases of fiber bundles:
\begin{itemize}
\item Whenever the fiber $F$ is the tangent space at each point of $\M$, the fiber bundle is called the {\bf\index{fiber bundle!tangent} tangent bundle} $T\M$. Sections of a tangent bundle are just vector fields. In the same way we have the cotangent bundle $T^*\M$. More generally, a {\bf\index{fiber bundle!vector} vector bundle} is a fiber bundle whose fiber is a vector space.
\item Whenever the fiber is a Lie group $G$ with $G$ acting transitively on each fiber, we have a {\bf\index{fiber bundle!principal} principal bundle} or {\bf $G$-bundle}. A gauge is the choice of a particular section of a principal bundle, and a gauge transformation is a transformation between two sections. 
\end{itemize}

\begin{defn}
 Given a $G$-bundle $(E,\M,\pi,G)$, a {\bf\index{connection} $G$-connection} $A$ is a $1$-form with value in the Lie algebra of $G$, and gives rise to a covariant derivative defined by $\nabla_\mu X = \partial_\mu X + A_\mu X$.\footnote{The operator $\nabla$ is  also called itself a connection.} The {\bf\index{curvature} curvature} is the $2$-form given by \mbox{$F = dA + A \wedge A$}.
\end{defn}

From the $G$-connection, we can introduce the Yang-Mills Lagrangian
$$\L_{YM} = \frac 1{4 g^2} F_{\mu\nu} F^{\mu\nu}$$
where $g$ is a coupling constant, and whose action gives a restriction on gauge fields $A_\mu$. So we have here a field theory which can be quantized, and quantized fields give rise to gauge bosons, particles carrying the fundamental forces. The number of gauge bosons depends directly on the dimension of the Lie algebra of $G$.\\

The first gauge theory to be discovered was quantum electrodynamics with the abelian gauge group $U(1)$ and one gauge boson (photon). The theory was extended to the non-abelian groups $SU(2)$ and $SU(3)$, respectively describing the weak interaction with three weak bosons ($W^+$, $W^-$, $Z$) and the strong interaction with height bosons (gluons).\\

\subsection{Standard model}

The standard model of particle physics is currently the most complete and experimented model which includes description of matter and all the fundamental interactions except gravitation. It is a quantized non-abelian gauge theory with gauge group $U(1) \times SU(2) \times SU(3)$.\\

The standard model contains two types of particles, fermions (particles representing matter) and bosons (particles carrying interactions). There are $12$ fermions divided in two groups (leptons and quarks) and in three generations. The Bosons are the $12$ gauges bosons we have described before, plus an additional hypothetical massive particle called the Higgs boson -- never observed at the current time -- whose role is to explain the mass of some other particles.\\

The standard model is a quite successful unification theory, in the sense that it includes in a quantized single model three of the four interactions. Its mathematical background is clearly algebraic since it is a gauge theory based on the noncommutative group \mbox{$U(1) \times SU(2) \times SU(3)$}. However, the standard model cannot explain any of the gravitational aspects of physics, neither the usual Einstein's theory of gravitation nor the more recent cosmological elements as dark matter and dark energy.\\

So the current {\it dream} about physics of fundamental interactions is to find a way to combine these two problems: gravitation and quantum theory of elementary particles. From a mathematical point of view, it can be seen as finding a way to combine in a single formalism geometrical aspects from gravitation and algebraic aspects from quantized gauge theories. So an important key could be the construction of new mathematical tools which could give a complete and well defined framework to support such a theory.


\newpage

\section{Einstein algebras}\label{einsteinalg}

We have seen in the last section that the construction of any kind of unification theory is  related to the construction of a way to deal with both geometrical structures (for gravitation) and algebraic structures (for quantization and gauge theories). The first chapter of our dissertation will be completely devoted on this idea: finding a way to combine those different aspects of mathematics. We will go through diverse possibilities -- passing from the well-known and unavoidable Gel'fand theorem to the theories of quantum gravity and of course to noncommutative geometry -- but we start with a naive idea that came from R. Geroch in 1972 \cite{GerEA}. This approach itself has never been really developed, except for some extensions by M. Heller \cite{HellerEA}, but the initial idea seems for us very important and we will use it as a kind of Ariane's thread.\\

We will first start by resetting the basic notions of pseudo-Riemannian geometry from a more abstract point of view.

\begin{defn}
A {\bf\index{module} left (right) $\A$-module} over an algebra $\A$ (or a ring) is an abelian group $(M, +)$ with a left (right) operation \mbox{$\A × M \rightarrow M : (a,x) \leadsto ax$} (resp. \mbox{$M × \A \rightarrow M : (x,a) \leadsto xa$}) such that for all $a,b \in\A$, $x,y\in M$ we have:
\begin{itemize}
\item $a(x + y) = ax + ay$
\item $(a + b)x = ax + bx$
\item $(ab)x = a(bx)$
\item $1.x = x $\quad if $\A$ has a unity
\end{itemize}
and similarly for a right module. A module which is both left and right with compatible multiplication is called a {\bf\index{module!bimodule} bimodule}.\end{defn}

Now let us take a manifold $\M$. We  consider any manifold to be smooth, i.e.~with $C^\infty$ functions and transition maps. Let $\A$ be the collection of real-valued functions $C^\infty(\M,\setR)$. Then $\A$ is a vector space with pointwise sum and product by scalars, and an algebra with pointwise multiplication. From $\A$ we can extract the subalgebra $\R$ of constant functions, which is isomorphic to $\setR$ if seen as a ring.

\begin{defn}
A {\bf\index{derivation} derivation} on $\A$ with the subring $\R \cong \setR$ is a mapping \mbox{$\xi : \A \rightarrow \A$} with the properties:
\begin{itemize}
\item $\xi(\alpha\, a+b) = \alpha\, \xi(a) + \xi(b)$
\item $\xi(ab) = \xi(a)\, b + a\, \xi(b)$
\item $a\in\R \implies \xi(a) = 0$
\end{itemize}
for all $a,b \in\A$, $\alpha\in\setR$.
\end{defn}

Let us note the collection of all derivations by $\D$. Then $\D$ is a left $\A$-module since $a\xi$ is still a derivation for $a\in\A$ and $\xi\in\D$. Moreover, the commutator \mbox{$[\xi,\eta] = \xi \eta - \eta \xi$} is also a derivation, so we have a Lie algebra structure. Since $\A$ is the set of continuous real functions on the manifold $\M$, we can identify the set $\D$ with the collection of all smooth vector fields on $\M$. Then the dual module $\D^*$, the set of all linear functionals on $\D$, is nothing but the space of covariant vector fields (fields of differential 1-forms). In the end, every tensor field is just a multilinear mapping \mbox{$\D \times \cdots \times \D \times \D^* \times \cdots \times \D^* \rightarrow \A$}.\\

A {\bf\index{metric} metric} can be defined as a symmetric isomorphism \mbox{$g : \D \rightarrow \D^*$}, symmetric in the sense of \mbox{$g(\xi,\eta) = g(\eta,\xi)$} if we define it as a tensor field \mbox{$g : \D \times \D \rightarrow \A : (\xi,\eta) \leadsto g(\xi) \eta$}. $g$ will be a Riemannian metric if there exists a basis \mbox{$(\xi_0, \xi_1, ... , \xi_{n-1})$} of $\D$ such that $g(\xi_i, \xi_j) = \delta_{ij}$, or a Lorentzian metric if  $g(\xi_0,\xi_0) = -1$ with the other relations being unchanged.\\

By this manner, we have reset the basis of pseudo-Riemannian geometry with continuous functions and derivations, but what is interesting in that approach is that we have  never used the manifold $\M$ itself, except for fixing the algebra of continuous functions $\A$. In other words, it is possible to construct the same elements starting with an arbitrary algebra $\A$, with the respect of the condition that $\A$ is a commutative algebra with a unity and with a subring $\R$ containing the unity and isomorphic to $\setR$. So from now we will remove the manifold $\M$ and we will only work  with our sets $\A$, $\D$ and a metric $g$ in order to construct the other fundamental elements of pseudo-Riemannian geometry.\\

A {\bf\index{covariant derivative} covariant derivative} in the $\A$-module $\D$ is any mapping 
$$\nabla : \D \times \D \rightarrow \D : (\xi,\eta) \leadsto \nabla_\xi (\eta)$$
 with linearity for the first argument, additivity for the second, and
$$\nabla_\xi (a \eta) = \xi(a) \eta + a \nabla_\xi \eta \qquad \forall a\in\A,\ \forall \xi, \eta \in\D.$$
Then the covariant derivative can be extended to any tensor field (for the second argument) by imposing the Leibniz rule on tensor product. It is well known that there exists a unique symmetric covariant derivative (the Levi--Civita connection) compatible with a given metric $g$, i.e.~such that $\nabla_\xi\, g = 0\ \forall\xi\in\D$.\\

\vspace{0.2cm}                                                     

The {\bf\index{tensor!Riemann} Riemann tensor} is the mapping 
$$R : \D^* \times \D \times \D \times \D \rightarrow \A$$
 defined by \mbox{$R(\varphi,\xi,\eta,\gamma) = \varphi(R_{\xi,\eta} \gamma)$} where \mbox{$R_{\xi,\eta} = \left[ \nabla_\xi, \nabla_\eta \right] - \nabla_{[\xi,\eta]}$}.\\

At the end, we need a way to construct the Ricci tensor, and especially a way to take a trace of tensor fields. Since the extension to any tensor is obvious, we can restrict the definition to tensors of rank two.\\

\begin{defn}
A {\bf\index{contraction} contraction} on the set of tensor fields \mbox{$\D^* \times \D \rightarrow \A$} is an operation "$\tr$" such that $\tr(\alpha)$ is an element of $\A$ for every rank two tensor $\alpha$, and with the following properties:
\begin{itemize}
\item $\tr(\alpha + a \beta) = \tr(\alpha) + a \tr(\beta)$ for $a\in\A$
\item $\tr(\xi\otimes \eta) =\xi(\eta)$ for $\xi \in \D^*$, $\eta \in \D$ and $\xi\otimes\eta$ being the tensor product.\\
\end{itemize}
\end{defn}

These properties imply the existence and unicity of the contraction operation. So we can define the {\bf\index{tensor!Ricci} Ricci tensor} to be \mbox{$\text{Ric} : \D \times \D \rightarrow \A : (\xi,\gamma) \leadsto \tr R(\,\cdot\,, \xi, \,\cdot\,, \gamma)$}, with the trace taken over the dot arguments.\\

\vspace{0.2cm}                                                     

So what have we now? We have a way, or more precisely an idea to construct the fundamental elements of Einstein general relativity based on a purely algebraic framework, with no more reference to manifold, so no reference on particular events. We have of course to put some constraints on this algebraic system in order to correspond to a geometrical one. Here are the constraints derived from Geroch's ones~\cite{GerEA}:\\
\begin{defn}
An {\bf\index{algebra!Einstein} Einstein algebra} is an algebra $\A$ over $\setR$ with derivations $\D$ and a symmetric isomorphism $g : \D \rightarrow \D^*$ such that:
\begin{itemize}
\item $\A$ is commutative with a unity
\item There exists a subspace $\R\subset\A$ with ring structure and containing the unity which is isomorphic to $\setR$
\item $\D$ vanishes on $\R$
\item $g$ is a Lorentzian metric
\item The Ricci tensor vanishes\\
\end{itemize}
\end{defn}

The last condition must be relaxed to obtain Einstein algebras with sources, by instead imposing that the Einstein tensor is equal to a suitable energy-momentum tensor. It is then obvious that every spacetime which is a solution of Einstein's equation of general relativity has a corresponding Einstein algebra, but the reverse is not guaranteed. Actually, general relativity can be seen as a special case of Einstein algebras.\\

\vspace{0.2cm}                                                     

What we can hold from this is the fact that, in general relativity, reference to a geometrical manifold can be replaced by an algebraic structure based on continuous functions, and that this algebraization can conduct to more general spaces. The questions now are the following ones:
\begin{itemize}
\item Can all the information be recovered at a geometrical level from the algebraic one (so is there any loss of information)?
\item Can all concepts of pseudo-Riemannian geometry be translated in an algebraic framework?
\item Can this kind of algebraization be used to create more general spaces which could include quantization and/or the other fundamental interactions?
\end{itemize}

The answer of the first question will be mainly presented in the next section, with the establishment of the Gel'fand transform and the Gel'fand--Naimark theorem. The other ones will be discussed in the sequel.


\newpage

\section{Gel'fand's theory}\label{gelfsect}

This section is completely devoted to the important theorem proved by I. Gel'fand and M. Naimark, but we need first to introduce the background of $C^*$-algebras. Most of these notions can be found in the books of J. M. Gracia--Bondia, C. Varilly and H. Figueroa \cite{Var} or of J. Dixmier \cite{DixmierALG}.\\

A {\bf\index{algebra!normed} normed algebra} is an algebra equipped with a norm $\norm{\,\cdot\,}$ such that \mbox{$\norm{ab} \leq \norm{a} \norm{b}$} for all elements $a,b$ of the algebra. If this norm is complete, then we have a {\bf\index{algebra!Banach} Banach algebra}.\\

A vector subspace $\B\subset\A$ is a subalgebra if it is closed under multiplication. A left (right) {\bf\index{ideal} ideal} $\I$ is a subalgebra of $\A$ where $ab\in\I$ ($ba\in\I$) $\forall a\in\A, \forall b\in\I$ and is called two-sided if both left and right. An ideal is maximal if there is no other ideal containing it and distinct from $\A$.\\

A {\bf\index{algebra!*-algebra} *-algebra} is an algebra with an involution map $a \leadsto a^*$ satisfying $a^{**} = a$, \mbox{$(\lambda a + b)^* = \bar\lambda a^* + b^*$} and \mbox{$(ab)^* = b^*a^*$}.  A *-homomorphism between two *-algebras is a group homomorphism $f$ such that \mbox{$f(ab) = f(a)f(b)$} and $f(a^*)={f(a)}^*$. In this section, we will assume all algebras to be over the field $\setC$ of complex numbers.\\

If a Banach algebra contains a unity, i.e.~an element $1\in\A$ such that \mbox{$1 a = a 1 = a,\ \forall a \in \A$} and $\norm{1}=1$, the algebra is said {\bf\index{algebra!unital} unital}. If it does not, it is always possible to construct an {\bf\index{algebra!unitization} unitized Banach algebra} \mbox{$\A^+= \A \oplus  \setC$} with the trivial sum, the product 
$$(a,\lambda)(b,\mu) = (ab+\lambda b + \mu a, \lambda \mu),$$
 the extended norm 
 $$\norm{(a,\lambda)} = \sup{\set{\norm{a b + \lambda b} : \norm{b} \leq 1}}$$
 and the unity \mbox{$1 = (0,1)$}.

\begin{defn}
A {\bf\index{algebra!$C^*$-algebra} $C^*$-algebra $\A$} is a Banach *-algebra, unital or not, that satisfies the equality:
\begin{equation}\label{Cstarsprop} \norm{a^*a} = \norm{a}^2 \quad\text{or equivalently}\quad \norm{aa^*} = \norm{a}^2,\quad\forall a\in\A.\end{equation}
\end{defn}

A best example of $C^*$-algebra is the algebra $C_0(X)$ of complex continuous functions {\bf\index{vanishing at infinity} vanishing at infinity} on a locally compact Hausdorff space $X$ with the supremum norm \mbox{$\norm{f}_\infty = \sup_{x\in X}\abs{f(x)}$} and the pointwise product. The space $C_0(X)$ can be mathematically defined as the space of continuous functions $f\in C(X)$ such that $\forall \epsilon >0$, $\abs{f} < \epsilon$ outside a compact set. If the space $X$ is compact, then $C_0(X) = C(X)$ and is a unital $C^*$-algebra. Another example is the $C^*$-algebra $\B(\H)$ of bounded linear operators on a Hilbert space $\H$ with composition of operators, where the closed ideal $\K(\H)$ of compact operators  is also a $C^*$-algebra. We can note that if $\A$ is a non-unital $C^*$-algebra, its unitization $\A^+$ is automatically a $C^*$-algebra.

\begin{defn}
Let $\A$ be a unital Banach algebra. The {\bf\index{spectrum} spectrum $\sigma(a)$} of $a\in\A$ is the complement of the {\bf\index{resolvent!set} resolvent set} of $a$, i.e.~the set \mbox{$\setC\setminus\rho(a)$} where \mbox{$\rho(a)=\set{\lambda\in\setC : \prt{a-\lambda 1}^{-1}\in\A}$}. If $\A$ is non-unital, the spectrum is taken in $A^+$. Moreover, the {\bf\index{spectral radius} spectral radius} is the real number \mbox{$r(a) = \sup\set{\abs{\lambda} : \lambda\in\sigma(a)}$}.
\end{defn}

\begin{defn}\label{defspectrum1}
Let $\A$ be a commutative Banach algebra. The {\bf\index{spectrum!of a commutative algebra} spectrum $\Delta(\A)$} of $\A$ is the set of all non-zero *-homomorphisms \mbox{$\chi :\A \rightarrow \setC : a \leadsto \chi(a)$}, where each *-homomorphism is called a {\bf\index{character} character}. Any character can be extended to $\A^+$ by setting \mbox{$\chi(0,1)=1$}. We can note that every character is automatically continuous (because they are norm decreasing).
\end{defn}

These two definitions of "spectrum" are  related, as we will show latter. We go now directly into the main theorem of this section -- the Gel'fand--Naimark theorem based on the Gel'fand transform -- and we will devote the rest of this section to the proof and consequences of this result.

\begin{defn}\label{geltransf}
The {\bf\index{Gel'fand transform} Gel'fand transform} is defined by
$$ \bigvee : \A \rightarrow C_0\prt{\Delta(\A)} : a \leadsto \hat a \ \text{ where }\ \hat a(\chi) = \chi(a).$$
\end{defn}

\begin{thm}[Gel'fand--Naimark]\label{GelNai}
Let $\A$ be a commutative $C^*$-algebra. Then the Gel'fand transform is an isometric *-isomorphism between $\A$ and $C_0\prt{\Delta(\A)}$.
\end{thm}

In order to prove this theorem, we need some technical lemmas. For the following we will assume $\A$ to be unital and commutative. For the non-unital case, we can perform the proof in $A^+$ and use similar arguments.

\begin{lemma}\label{rayspect}
Let $\A$ be a Banach algebra, we have the following identity for the spectral radius:
$$r(a) = \lim_{n\rightarrow\infty} \norm{a^n}^{\frac 1 n}.$$
Moreover, if $\A$ is a $C^*$-algebra and $a$ is normal ($aa^*=a^*a$) or Hermitian (a=a*), or if $\A$ is a commutative $C^*$-algebra we have:
$$r(a) =  \norm{a}.$$
\end{lemma}

\begin{proof}
First we will prove that the spectrum is not empty.  We define the function $f = \phi \circ R$ where $\phi$ is a continuous linear form on $\A$ and $R(\lambda) = \prt{a-\lambda 1}^{-1} $ is the {\bf\index{resolvent} resolvent}. The resolvent respects the following formula: 
$$R(w)-R(z) = (w-z)R(w)R(z)$$
 which gives 
 $$f(w)-f(z) = (w-z)\,\phi\prt{R(w)R(z)}.$$
 So $f$ is holomorphic on the resolvent set and \mbox{$f'(z) = \phi\prt{R^2(z)}$}. If we suppose the spectrum to be empty, then $f$ must be holomorphic on the whole complex plan.

Now let us choose $z\in\setC$ such that \mbox{$\abs{z}>\norm{a}$}, then 
\begin{equation}\label{seriesR}
R(z) = - \frac 1 z \sum_0^\infty \frac{a^n}{z^n}
\end{equation}
 with the series being convergent, and 
 \begin{equation}\label{seriesf}
f(z) = - \frac 1 z \sum_0^\infty \frac{\phi(a^n)}{z^n} 
\end{equation}
 is of order $\frac 1 z$ and so bounded. By the Liouville theorem\footnote{Liouville theorem: Every holomorphic function $f : \setC \rightarrow \setC$ for which there exists a positive number $M$ such that \mbox{$\abs{f(z)} \leq M \ \forall z\in\setC$} is constant.}, if $f$ is holomorphic, $f$ is constant, which is absurd because $\phi$ is an arbitrary function. So the spectrum must not be empty.

The expression of the resolvent in (\ref{seriesR}) gives us the fact that $r(a)\leq \norm{a}$. Moreover, this series converges for \mbox{$\abs{z} > \lim_{n\rightarrow\infty} \norm{a^n}^{\frac 1n}$}, so \mbox{$z\in\sigma(a) \implies \abs{z} \leq \lim_{n\rightarrow\infty} \norm{a^n}^{\frac 1n}$.} By taking the supremum over $z\in\sigma(a)$ we have \mbox{$r(a) \leq \lim_{n\rightarrow\infty} \norm{a^n}^{\frac 1n}$}.

Then let us suppose $\abs{z}>r(a)$. We can remember that the function (\ref{seriesf}) is well defined in this case, so
$$\lim_{n\rightarrow\infty}\phi(\frac{a^n}{z^n})) = 0.$$
 $\phi$ is an arbitrary element of the dual space $\A'$ which is a Banach space for the norm \mbox{$\norm{\phi} = \sup_{a\in\A} \abs{\phi(a)}$}. If we consider the family 
$$\F = \set{\frac{a^n}{z^n} : n\in \setN},$$
 then we have that for each $\phi \in\A'$ the set \mbox{$\set{\abs{\phi(b)} : b\in\F}$} is bounded. For each $b$ we can build a map \mbox{$\hat b : \A' \rightarrow \setC : \phi \leadsto \phi(b)$} whose the norm is:
$$\Vert\hat b\Vert' = \sup_{{\phi\in\A'}\atop{\phi\neq 0}} \frac{\abs{\phi(b)}}{\norm{\phi}} = \norm{b} = \frac{\norm{a^n}}{\abs{z}^n}.$$

By the uniform boundedness principle\footnote{Uniform boundedness principle: If $F$ is a collection of continuous linear operators from the Banach space $X$ to the normed vector space $Y$ and if for all $x\in X$ we have \mbox{$\sup_{T\in F} \norm{T(x)} < \infty$}, then \mbox{$\sup_{T\in F} \norm{T} < \infty$}.} the set \mbox{$\set{\Vert \hat b \Vert' : b\in\F}$} is bounded, so the set \mbox{$\set{\frac{\norm{a^n}}{\abs{z}^n} : n\in\setN}$} is bounded by a constant $C$. From \mbox{$\norm{a^n} \leq \abs{z}^n C$}, we find  
$$\lim \sup_{n\rightarrow\infty} \norm{a^n}^{\frac 1n} \leq \abs z$$
 which implies 
 $$\lim \sup_{n\rightarrow\infty} \norm{a^n}^{\frac 1n} \leq r(a)$$
  as $\abs{z}$ can be arbitrary close to $r(a)$. Thus
  $$r(a) = \lim_{n\rightarrow\infty} \norm{a^n}^{\frac 1n}.$$

If $\A$ is a $C^*$-algebra and $a$ is normal, we can use the fact that $aa^*$ is Hermitian and the $C^*$-algebra property (\ref{Cstarsprop}) to get:
\begin{eqnarray*}
\norm{a^{2^{n}}}^2 &=  &\norm{a^{2^{n}} \prt{a^{2^{n}}}^*} \qquad\text{($C^*$-algebra property)}\\
&= & \norm{\prt{aa^*}^{2^{n}}}\qquad\text{($normal$)}\\
&= &\norm{  \prt{aa^*}^{2^{n-1}} \prt{\prt{aa^*}^{2^{n-1}}}^*  }  \qquad\text{(Hermitian )}\\
&= & \norm{  \prt{aa^*}^{2^{n-1}} }^2 \qquad\text{($C^*$-algebra property)}\\
&= & \norm{aa^*}^{2^n}  \qquad\text{(iteration from the second line)}\\
&= &\norm{a}^{2^{n+1}} \qquad\text{($C^*$-algebra property)}
\end{eqnarray*}

The same result occurs if $a$ is Hermitian or if the algebra $\A$ is commutative, since in this case the normal relation $aa^*=a^*a$ is automatically verified.

The obtained equality \mbox{$\norm{a} = \norm{a^{2^{n}}}^{\frac{1}{2^n}}$} gives the unique possible value of the limit \mbox{$r(a) = \lim_{n\rightarrow\infty} \norm{a^n}^{\frac 1n}$}, and so \mbox{$r(a) = \norm{a}$}.
\end{proof}

\vspace{0.8em}                                                      

\begin{lemma}[Gel'fand--Mazur theorem]\label{mazur}
A Banach algebra in which every non-zero element is invertible is isometrically isomorphic to $\setC$.
\end{lemma}
\begin{proof}
As shown in the proof of the Lemma \ref{rayspect}, the spectrum of a non-zero element $a\in\A$ is non empty, so for every $a\in\A$, $a\neq0$, there exists $\lambda\in\sigma(a) \subset\setC$ such that $a-\lambda 1 = 0$, because $0$ is the only non-invertible element. $a\leadsto \lambda$ gives the isomorphism.
\end{proof}

\vspace{0.8em}                                                      

The following lemma gives the relation between the notion of spectrum of a particular element of $\A$ and the notion of spectrum of the algebra $\A$ itself.

\vspace{0.8em}                                                      

\begin{lemma}\label{equispect}
Let $\A$ be a unital, commutative Banach algebra and \mbox{$a\in\A$}. Then $\lambda\in\sigma(a)$ if and only if there exists $\chi\in\Delta(\A)$ such that $\chi(a) = \lambda$. Therefore, \mbox{$r(a) = \sup\set{\abs{\chi(a)} : \chi\in\Delta(\A)}$}.
\end{lemma}

\begin{proof}
$\ker(\chi)$ is a two-sided ideal in $\A$ since \mbox{$\chi(ab) = \chi(a)\chi(b) = 0$}, \mbox{$\forall a\in\A$}, \mbox{$\forall b\in\ker(\chi)$}. Since $\chi$ is a linear functional, $\ker(\chi)$ is in particular a vector subspace of $\A$ of codimension one, so $\ker(\chi)$ is a maximal ideal in $\A$. Therefore, $\ker(\chi)$ cannot contain any invertible element. Now let us take $\chi\in\Delta(\A)$ with $\chi(a) = \lambda$. Then \mbox{$\chi(a-\lambda 1)=0$}, so $(a-\lambda 1)\in\ker(\chi)$. Because there is no invertible element in $\ker(\chi)$, $(a-\lambda 1)^{-1}$ does not exist, and $\lambda\in\sigma(a)$.

For the reverse, let us take $\lambda\in\sigma(a)$. Then by Zorn's lemma,\footnote{By Zorn's lemma, every unital, commutative algebra contains a maximal ideal.} \mbox{$\A(a-\lambda 1)$} is a proper ideal contained in a maximal ideal $\I$. The quotient space $\A / \I$ is a Banach algebra with the norm \mbox{$\norm{[b]} = \inf_{c\in[b]} \norm{c}$} and all non-zero elements are invertible, so $\A / \I$ is isomorphic to $\setC$ by the Gel'fand--Mazur theorem (Lemma \ref{mazur}). So there exists a *-homomorphism \mbox{$\chi : \A/\I \rightarrow \setC$} which can be extended to $\A$ by setting \mbox{$\chi(b) = \chi([b])$}. By construction, $\I=\ker(\chi)$, so \mbox{$a\in\I\implie \chi(a-\lambda 1) = 0\implie \chi(a) = \lambda$}.
\end{proof}

\vspace{0.5em}                                                      

\begin{lemma}\label{topo}
The spectrum $\Delta(\A)$ of a unital, commutative $C^*$-algebra $\A$ can be endowed with the weak-* topology -- called the Gel'fand topology -- in which $\Delta(\A)$ is a compact Hausdorff space.
\end{lemma}

\begin{proof}
The characters of $\A$ form a subset of the unit sphere in $\A'$, the topological dual of $\A$. Indeed, by Lemma \ref{rayspect} and \ref{equispect},
$$ \norm{\chi} = \sup_{{a\in\A}\atop{a\neq 0}}\frac{\abs{\chi(a)}}{\norm{a}}  \leq \sup_{{a\in\A}\atop{a\neq 0}}\frac{  \sup{ \abs{ \tilde\chi(a) }: \tilde \chi \in\Delta(\A) }   }{\norm{a}}  = \sup_{{a\in\A}\atop{a\neq 0}}\frac{ r(a)  }{\norm{a}}  = 1$$
and because $\chi(1) = 1$ we must have $\norm{\chi} = 1$.

Now we can show that $\Delta(\A)$ is closed. Let us take a sequence of characters $\prt{\chi^\alpha} \rightarrow \chi$, then 
$$\chi(a^*b) = \lim_\alpha \chi^\alpha(a^*b) = \lim_\alpha \overline{\chi^\alpha(a)} \chi^\alpha(b)  = \overline{\chi(a)} \chi(b) $$
 and similarly for the addition, so $\chi$ is a character.

If we endow $\A'$ with the weak-* topology -- the weakest topology such that all functions \mbox{$x : \A' \rightarrow \setC : \phi \leadsto \phi(x)$} are continuous -- then the unit sphere of $\A'$ is compact (Banac--Alaoglu theorem), and since a closed subspace of a compact space is compact, $\Delta(\A)$ is compact for this topology.

Now we fix the set \mbox{$C = \set{\hat a : a \in \A}$}, the range of the Gel'fand transform. Each element in $C$ is a continuous function, by definition of the weak-* topology. Moreover, let us consider any $\chi_1,\chi_2\in\Delta(\A)$ with $\chi_1 \neq \chi_2$, there exists $a\in\A$ such that \mbox{$\chi_1(a) \neq \chi_2(a) \equi \hat a(\chi_1) \neq \hat a(\chi_2)$}, so $C$ separates the points of $\Delta(\A)$. Because $C$ is a set of continuous functions which separate the points of the spectrum, the topology is Hausdorff.
\end{proof}

\vspace{0.5em}                                                      

Now we can come into the proof of our Theorem \ref{GelNai}:

\begin{proof}[Proof of Gel'fand--Naimark's theorem]
The *-homomorphism property is obvious:
$$\widehat{a^*b}(\chi) = \chi(a^*b) = \overline{\chi(a)}\,\chi(b) = \overline{\hat a(\chi)} \,\hat b(\chi). $$

By definition of the supremum norm and using the Lemma \ref{equispect}, we have:
$$\norm{\hat a} = \sup_{\chi\in\Delta(\A)}\abs{\hat a(\chi)} = \sup_{\chi\in\Delta(\A)}\abs{\chi(a)} = r(a).$$

Then a simple application of Lemma \ref{rayspect} gives the isometry:
$$ r(a) = \norm{a} \limplie \norm{\hat a} = \norm{a}.$$

The Gel'fand transform is an injection since, if we have \mbox{$\hat a_1 = \hat a_2$}, by isometry \mbox{$\norm{\hat a_1 - \hat a_2} = \norm{a_1 - a_2}$}, hence $a_1 = a_2$.

To prove the surjection, let us remind that, by Lemma \ref{topo}, $\Delta(\A)$ is a compact Hausdorff space, and \mbox{$C = \set{\hat a : a \in \A}$} is a subset of $C_0\prt{\Delta(\A)}$ which separates the point of $\Delta(\A)$. Because $\A$ is unital, $\hat 1(\chi) = \chi(1) = 1$ is a non-zero constant  function. Then we can apply the Stone--Weierstrass theorem\footnote{We will use this version of the Stone--Weierstrass theorem: If $X$ is a locally compact Hausdorff space and $S$ is a subalgebra of $C_0(X)$, then $S$ is dense in $C_0(X)$ if and only if it separates points and vanishes nowhere, i.e.~for every $x\in X$, there is some $a\in S$ such that $a(x)\neq0$.} which leads to the fact that $C$ is dense in $C_0\prt{\Delta(\A)}$. We complete the proof by showing that $C$ is closed, so $C = C_0\prt{\Delta(\A)}$. Indeed, if $\prt{\hat a^\alpha}$ is a Cauchy sequence, by isometry \mbox{$\norm{\hat a^\alpha - \hat a^\beta} = \norm{a^\alpha - a^\beta}$}, then $\prt{a^\alpha}$ is a Cauchy sequence in the Banach space $\A$ which converges to $a\in\A$, and so $\prt{\hat a^\alpha}$ converges to $\hat a \in C$.
\end{proof}

The Gel'fand transform -- and its associated theorem -- gives us a new way to deal with geometrical aspects. We already know that, on every manifold, the space of continuous functions (vanishing at infinity if the manifold is not compact) is a $C^*$-algebra and that many informations about the geometrical system can be translated in the algebraic framework of continuous functions (cf. Section \ref{einsteinalg}), but the preservation of all the information was not guaranteed.\\

This is not the case anymore, since a topological manifold being a Hausdorff space, the Theorem \ref{GelNai} can be applied, so it is possible to recover the complete manifold directly from the algebra of continuous functions, just by considering the space of characters, as illustrated by the Figure \ref{Gelf_fig}.\\

\vspace{0.5em}                                                      

\begin{figure}[!ht]
\begin{center}
\includegraphics[width=10cm]{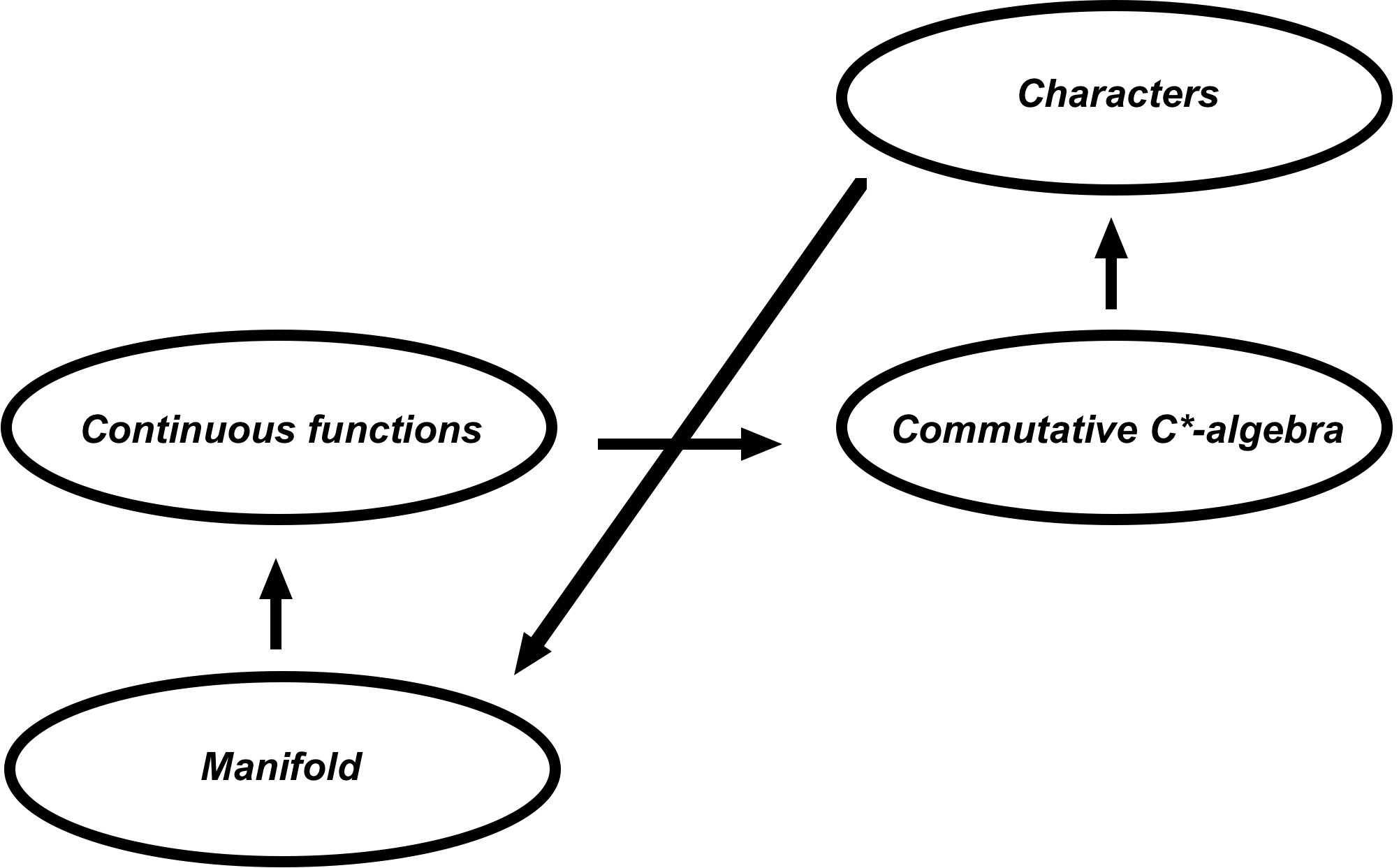}
\end{center}
\caption{How to recover a manifold with the Gel'fand transform}\label{Gelf_fig}
\end{figure}

\vspace{0.5em}                                                      

So it is possible to completely trade geometrical spaces for algebras, but in order to formalize this we need to introduce the language of categories.\\

\vspace{0.5em}                                                      

\begin{defn}
A {\bf\index{category} category} is a class of objects together with a class of morphisms between those objects with an associative composition between the morphisms and the existence of identity morphisms $1$ such that $f \circ 1 = f$ and $1 \circ f = f$ for every suitable morphism $f$.
\end{defn}

\begin{defn}
A covariant (contravariant) {\bf\index{functor} functor $\F$} between two categories $C$ and $D$ is a mapping that associates to each object $X\in C$ an object $\F(X)\in D$, and associates to each morphism \mbox{$f : X \rightarrow Y \in C$} a morphism \mbox{$\F(f) : \F(X) \rightarrow \F(Y) \in D$} (resp. \mbox{$\F(f) : \F(Y) \rightarrow \F(X) \in D$}) such that:
\begin{itemize}
\item $\F(\text{id}_X) = \text{id}_{\F(X)}$
\item $\F(g\circ f) = \F(g) \circ \F(f) $ (resp. $\F(g\circ f) = \F(f) \circ \F(g) $)
\end{itemize}
\end{defn}

\begin{defn}
Two categories $C$ and $D$ are {\bf\index{category!equivalent}  equivalent} if there exist two covariant functors $\F : C \rightarrow D$ and $\G : D \rightarrow C$ and two natural isomorphisms \mbox{$\epsilon : \F\circ\G \rightarrow \mathbb I_D$} and \mbox{$\eta : \mathbb I_C \rightarrow \G \circ\F$}, called natural transformations, where $\mathbb I_C$ and $\mathbb I_D$ are the identity functors. If the functors are contravariant, the categories are said to be {\bf\index{category!dually equivalent} dually equivalent}.
\end{defn}

\begin{lemma}\label{homoepsi}
Let us take the evaluation map \mbox{$\epsilon_x : C(X) \rightarrow X : f \leadsto f(x)$} on a compact Hausdorff space $X$. Then \mbox{$\epsilon_X : X \rightarrow \Delta(C(X)) : x \leadsto \epsilon_x$} is a homeomorphism between the space $X$ and the space $ \Delta(C(X) )$ with its Gel'fand topology.
\end{lemma}

\begin{proof}
$\epsilon_X$ is trivially continuous by the Gel'fand topology. It is also injective since by Urysohn's lemma,\footnote{Urysohn's lemma: A topological space is normal if and only if any two disjoint closed subsets can be separated by a function. This lemma can be applied since every compact Hausdorff space is normal.} there exists a function $f\in C(X)$ with values $f(x_0)=0$ and $f(x_1)=1$ on two distinct points  $x_0, x_1 \in X$, so \mbox{$\epsilon_{x_0} (f) \neq \epsilon_{x_1} (f)$}.

To show the surjection, let us take $\chi \in \Delta(C(X))$. Then $\ker(\chi)$ is a maximal ideal of $C(X)$ which separates the points of $X$. Because $\ker(\chi)$ cannot be dense in $C(X)$, we can deduce by Stone--Weierstrass theorem that $\ker(\chi)$ vanishes somewhere, i.e.~there exists a point $x\in X$ such that \mbox{$f(x) = \epsilon_x(f) = 0 \ \forall f\in\ker(\chi)$}. So the spaces $\ker(\chi)$ and $\ker(\epsilon_x)$ are identical since they are both maximal ideals, and for all $f \in C(X)$, \mbox{$\chi\prt{ f - \chi(f) 1} = 0 \implies \epsilon_x\prt{ f - \chi(f) 1} = 0 \implies \epsilon_x(f) = \chi(f)$}.
\end{proof}

\begin{prop}\label{equivgelf}
The category of compact Hausdorff spaces (with continuous maps) and the category of unital, commutative $C^*$-algebras (with unital *-homomorphisms) are dually equivalent.
\end{prop}

\begin{proof}
If \mbox{$f : X \rightarrow Y$} is a continuous mapping between two compact Hausdorff spaces, let us define the mapping $\C$ by \mbox{$\C f : C(Y) \rightarrow C(X) : h \leadsto h \circ f$}. Then $\C$ is a contravariant functor from the category of compact Hausdorff spaces and continuous maps to the category of unital, commutative $C^*$-algebras and unital *-homomorphisms.

Moreover, if $\phi$ is a unital *-homomorphism between two commutative $C^*$-algebras $\A$ and $\B$, the mapping \mbox{$\Sigma$} defined by \mbox{$\Sigma \phi : \Delta(\B) \rightarrow \Delta(\A) : \chi \leadsto \chi \circ \phi$} is also a contravariant functor.

From Lemma \ref{homoepsi}, \mbox{$\epsilon_X : X \rightarrow \Delta(C(X)) : x \leadsto \epsilon_x$} is a homeomorphism between the space $X$ and the space $ \Delta(C(X) )$. So $\epsilon$ is a natural transformation between the identity functor on the category of compact Hausdorff spaces and  the functor $\Sigma \circ \C$.

The natural transformation between the identity functor on the category of unital, commutative $C^*$-algebras and the functor $\C\circ \Sigma$ is just given by the Gel'fand transform $\bigvee$, with \mbox{$\bigvee_\A : \A \rightarrow C\prt{\Delta(\A)} : a \leadsto \hat a$}.
\end{proof}

\vspace{0.5em}                                                      

This result cannot be directly generalized to the non-unital case, since the mapping \mbox{$h \leadsto h \circ f$} does not necessarily send functions vanishing at infinity to functions vanishing at infinity. We need to introduce further notions.

\vspace{0.5em}                                                      

\begin{defn}
{A \bf\index{proper map} proper map} between two locally compact Hausdorff spaces is a map such that inverse images of compact subsets are compact.
\end{defn}

Proper maps send functions vanishing at infinity to functions vanishing at infinity. Indeed, if $h\in C_0(Y)$ and $f : X \rightarrow Y$ is a proper map, for $\epsilon > 0$ and by definition of $C_0(Y)$, there exists a compact set $K \subset Y$ such that $\abs{h(y)} < \epsilon$ if $y\notin K$, and so $K' = f^{-1}(K) \subset X$ is a compact such that $\abs{(h\circ f)(x)} < \epsilon$ if $x\notin K'$.

\vspace{0.5em}                                                      

\begin{defn}
The {\bf\index{compactification!Alexandroff} Alexandroff compactification} of a locally compact Hausdorff space $X$ is the compact Hausdorff space \mbox{$X^+ = X\ \cup\ \set{\infty}$} with the suitable extended topology.
\end{defn}

\vspace{0.5em}                                                      

The space of continuous functions $C(X^+)$ on the compactified space $X^+$ is no more than the unitization $\prt{C_0(X)}^+$ of the algebra of functions vanishing at infinity on the initial non-compact space $X$ \cite{Wegge}.

\vspace{0.5em}                                                      

\begin{defn}
A {\bf\index{pointed compact space} pointed compact space} is a pair \mbox{$(X,\star)$} where $X$ is a compact Hausdorff space and $\star\in X$ a particular element called the basepoint. A morphism $f$ between two pointed compact spaces \mbox{$(X,\star_X)$}  and \mbox{$(Y,\star_Y)$} is such that  $f(\star_X) = \star_Y$. The space $C(X,\star)$ is defined as the set of functions $f\in C(X)$ where $f(\star) = 0$.
\end{defn}

\vspace{0.5em}                                                      

\begin{prop}\label{essensurj}
The Alexandroff compactification defines a functor between locally compact Hausdorff spaces (with proper maps) to pointed compact spaces (with basepoint morphisms) which is surjective for the objects (so the functor is called essentially surjective). 
\end{prop}

\vspace{0.5em}                                                      

\begin{proof}
We just have to choose $\infty$ as the basepoint, and extend any continuous proper map $f = X \rightarrow Y$ to a morphism \mbox{$f^+ : X^+ \rightarrow Y^+$} by setting $f^+(\infty) = \infty$. Conversely, if \mbox{$(X,\star)$} is a pointed compact space, then \mbox{$X \setminus \set{\star}$} is a locally compact Hausdorff space, and the restriction of any morphism to \mbox{$X \setminus \set{\star}$} is proper since every map between compact spaces are proper, so the functor is essentially surjective.
\end{proof}

\vspace{0.5em}                                                      

\begin{remark}\label{remessensurj}
Let us notice that not every morphisms between pointed compact spaces can be recovered from a map between locally compact Hausdorff spaces. So the two categories are not equivalent.
\end{remark}

\vspace{0.5em}                                                      

\begin{prop}\label{equivgelfnu}
The category of pointed compact spaces (with morphisms) and the category of commutative $C^*$-algebras (with *-homomorphisms) are dually equivalent.
\end{prop}

\vspace{0.5em}                                                      

\begin{proof}
In fact there are two equivalences.

First let us define the {\it slice category of unital, commutative \mbox{$C^*$-algebras}} as the category whose objects are the pairs $(\A,f)$ with $\A$ a unital, commutative $C^*$-algebra, $f : \A \rightarrow \setC$ a homomorphism, and morphisms \mbox{$g : (\A,f) \rightarrow (\A',f')$} such that \mbox{$g : \A \rightarrow \A'$} is a morphism with $g \circ f = f'$. Then the unitization \mbox{$^+ : \A \rightarrow (\A^+, \pi_\A)$} with the unitized algebra  $\A^+$ and the homomorphism \mbox{$\pi_\A : \A^+ \rightarrow \setC : (a,z) \leadsto z$} is a functor from the category of commutative $C^*$-algebras to the slice category of unital, commutative $C^*$-algebras. The reverse functor is just given by the kernel application $\A = \ker(\pi_\A)$. So these two categories are equivalent.

To find the second equivalence, let us remark that every pointed compact space $(X,\star)$ can be seen as a couple \mbox{$(X,f)$} with $X$ a compact Hausdorff space and a morphism \mbox{$f : X \rightarrow \star$} with value to a point $\star \in X$. If we use the  functors $\C$ and $\Sigma$ defined in the proof of the Proposition \ref{equivgelf} we have that $\C(\star) = \setC$, so these functors define a dually equivalence between the category of pointed compact spaces and the slice category of unital, commutative $C^*$-algebras.

We conclude by the fact that the composition of two equivalences of categories is an equivalence.
\end{proof}

\vspace{0.5em}                                                      

The Proposition \ref{equivgelfnu} generalizes the Proposition \ref{equivgelf} to non-unital algebras. The complete equivalence is between the pointed compact spaces and the commutative $C^*$-algebras. Since every pointed compact space arises as the compactification of a locally compact Hausdorff space (Proposition \ref{essensurj}), we would like to extend the equivalence between locally compact Hausdorff space and commutative $C^*$-algebras, but the Remark \ref{remessensurj} tells us that it is not possible, because non-unital $C^*$-algebras introduce possible morphisms with no correspondence as morphisms between locally compact Hausdorff spaces. However, the Gel'fand theorem is still valid in this case to give an isomorphism between points and characters.\\

\vspace{0.5em}                                                      

In this section we have seen that the Gel'fand transform is a very useful mathematical tool creating some correspondences between geometrical spaces and algebraic spaces with $C^*$-algebra structure, especially because the correspondence created is an isomorphism and can be used to translate information from geometry to algebra as well as from algebra to geometry without loss of information. In the remaining of the chapter, we will illustrate two recent mathematical tools using this correspondence between geometry and algebra in order to fix mathematical frameworks which can be useful for unifying physical theories. The first one is the theory of quantum gravity developed to solve the problem of quantization of gravitation, and the second will be the general framework of noncommutative geometry.


\newpage

\section{Quantum gravity}\label{QG}

This section will be an illustration of our current main concern: can new mathematical tools combining geometry and algebra be developed to create more general spaces representing gravitation with additional physical aspects as quantization or the description of other fundamental interactions? One of the answers to this problem is given by a recent theory called {\it quantum gravity}, often called {\it loop quantum gravity}  from its initial formulation in terms of loops.\\

This theory arose late 1980s from A.~Ashtekar, C.~Rovelli and L.~Smolin, and provides a way to quantify Einstein's theory of general relativity. The mathematical formalism we present here -- known as the differential formalism of the theory -- was set in 1992-1994 and is better defined mathematically than the initial one.  As we present this theory as an illustration, we will not go into proofs and details but only give the main construction of the mathematical formalism behind. For more development we refer the reader to Rovelli's book \cite{Rov} for the physical aspects and to Thiemann's book \cite{Thiemann} for the mathematical side.\\

Let $\sigma$ be a locally compact 3-dimensional manifold with a Riemannian metric $g_{ab}$. We consider the structure of a $SU(2)$-bundle over $\sigma$, and define the triad field $e^i_a$  by $g_{ab} = \delta_{ij} e^i_a e^j_b$ (with $i=1,2,3$ relative to the Lie algebra of $SU(2)$).

\begin{defn}
The {\bf\index{Ashtekar variables} Ashtekar variables}\footnote{A pedagogical introduction of these variables can be found in \cite{Baez}.} are the two sets of variables $A^i_a$ and $E^a_i$ where:
\begin{itemize}
\item $A^i_a$ is a $SU(2)$ gauge connection on $\sigma$ (a 1-form with value in the Lie algebra of $SU(2)$)
\item $E^a_i$ is the {\it electrical field} given by
$$E^a_i = \frac 12 \epsilon_{ijk} \epsilon^{abc} e^j_b e^k_c$$
where $\epsilon_{ijk} $ is the Levi--Civita antisymmetric symbol.
\footnote{
$\epsilon_{ijk} =
\begin{cases}
  1 &\text{ for even permutations $(i,j,k)$ of } (1,2,3)\\
 -1 &\text{ for odd permutations}\\
 0 &\text{ otherwise}
\end{cases}$
}
\end{itemize}
\end{defn}

$A^i_a$ and $E^a_i$ respect the canonical structure given by the Poisson algebra:
\begin{equation*}
\begin{split}
&\{E^a_i(x), E^b_j(y)\} = 0,\quad \{A^i_a(x),A^j_b(y)\} = 0,\\
&\{E^a_i(x),A^j_b(y)\} = \frac{8\pi G}{c^4} \gamma\,\delta^a_b \delta^j_i \delta(x,y).
\end{split}
\end{equation*}
where $\gamma$ is a free (usually real) parameter of the theory, called the Immirzi parameter.\\

The classical equations of this theory can be derived from general relativity using a 3+1 decomposition of the spacetime ($\sigma\times\setR$) and the Hamiltonian formalism of general relativity (ADM formalism \cite{ADM}). These equations can be written in 3 groups:
\begin{itemize}
\item $G_j = D_a E^a_j = 0$
\item $H_a = F^j_{ab} E^b_j = 0$
\item $H =\prt{F^j_{ab} +(\gamma^2+1)\epsilon_{jmn}(A^m_a-\Gamma^m_a)(A^n_b-\Gamma^n_b)}\epsilon_{jkl} E^a_k E^b_l = 0$
\end{itemize}
with 
\begin{itemize}
\item $D_a E^a_j= \partial_a E^a_j+ \epsilon_{jkl} A^k_a E^a_l\ $ the covariant derivative relative to $A^j_a$
\item $F^j_{ab} = \partial_aA^j_b - \partial_bA^j_a + \epsilon_{jkl}A^k_aA^l_b\ $ the curvature relative to $A^j_a$
\item $\Gamma^i_{a} = \frac 12 \epsilon_{ijk} e^b_k \prt{\partial_{b}\,e^j_a - \partial_{a}\,e^j_b + e^{cj} e_{al} \partial_{b}\, e^l_c}\ $ the connection compatible with $e^j_a$ (spin connection)\\
\end{itemize}

The first two equations $G_j = 0$ and $H_a=0$ describe the geometry of the 3-dimensional manifold $\sigma$, by imposing respectively invariance under $SU(2)$ local transformation and diffeomorphism. The last equation \mbox{$H=0$} is the Hamiltonian constraint and can be seen as a fitting condition of slices of 3-dimensional manifolds $\sigma$ under a fourth one with additional variables given by the Lagrange multipliers of those constraints.\\

So the classical variables of quantum gravity are $SU(2)$ connections with electrical fields as momentum variables. We will note the space of all smooth $SU(2)$ connections by $\A$. The main idea of the theory is to perform a canonical quantization on this Hamiltonian system, so to replace connections $A^i_a$ by states, i.e.~functionals $\psi$ defined on $\A$ and belonging to a Hilbert space $\H$. Moreover the constraints must have a representation in terms of Hermitian operators on $\H$, especially the Hamiltonian constraint giving the Wheeler--DeWitt equation $\hat H \psi = 0$. Then those constraints would define a subspace $\H_{phys} \subset \H$ of physical states.\\

One could choose the space $C(\A)$ of functions on the space of connections with a kind of smoothness condition and with a scalar product defined on it as the space of states, but constraints of diffeomorphism and $SU(2)$ invariance are hardly represented in this case. Instead, the Hilbert space is chosen to be $\H = L^2(\bar\A,d\mu)$ with $\bar\A$ a more larger space of distributional connections, and $\mu$ a suitable measure that should be defined.

\begin{defn}
Let $c : [0,1] \rightarrow \sigma$ be a curve in $\sigma$, $A\in\A$ a particular connection and the function $H_{c,A} : [0,1] \rightarrow SU(2)$ the unique solution of the differential equation 
$$\dot H_{c,A} (t) = H_{c,A}(t) \; A^j_a(c(t))  \; \dot c^a(t) \; \tau_j$$
 with the initial condition \mbox{$H_{c,A}(0) = 1_{SU(2)}$} and $\tau_j$ the generators of the Lie algebra of $SU(2)$. Then \mbox{$A(c) = H_{c,A} (0) \in SU(2)$} is the {\bf\index{holonomy} holonomy} operator (or parallel transport).
\end{defn}

The holonomy has a good behaviour under gauge and spatial transformations, so it is a good candidate for the construction of the Hilbert space. The definition of the holonomy can be extended to any piecewise smooth curve by taking the product of the holonomy on each piece, and is independent of any reparametrization or retracing if the orientation is conserved. So we can consider the holonomy to be defined on paths (equivalence class under reparametrization or retracing with a fixed orientation), with $\P$ the set of all paths on $\sigma$. Moreover, for any $A\in\A$ and any $p,p'\in\P$, if we consider composition and inversion of paths, we have that
\begin{equation}\label{groupoid}
A(p\circ p \prime) = A(p)\,A(p\prime) \ \text{ and }\ A(p^{-1}) = \prt{A(p)}^{-1}.
\end{equation}

Now we can remark that the 3-dimensional manifold $\sigma$ can be interpreted as a category, the points of $\sigma$ being the objects and $\P$ being the set of morphisms. In fact, each of those morphisms being an isomorphism, the category is called a {\bf\index{groupoid} groupoid}. We will denote this groupoid also by $\P$. By (\ref{groupoid}), any  $A\in\A$ defines a groupoid morphism, but not every groupoid morphism can be expressed in terms of a connection $A$. So the space
$$\bar \A = Hom(\P,SU(2))$$
of all groupoid morphisms from the set of paths in $\sigma$ onto the gauge group $SU(2)$ is a distributional extension of $\A$.\\

Now we have to set a measure on this space. For that we need to introduce some graph notions. An oriented {\bf\index{graph} graph} $\gamma$ in $\sigma$ is a set of vertices $V(\gamma)$ that are points on $\sigma$ together with a set of edges $E(\gamma)$ that are paths between elements of $V(\gamma)$. Each graph is in fact a subgroupoid of $\P$, so we can define the elements:
$$\X_\gamma = Hom(\gamma,SU(2)).$$

The notion of subgraph defines a partially ordered relation $\gamma\prec\gamma'$ which induces  a notion of projection on $\set{\X_\gamma}$ by
$$p_{\gamma'\gamma} : \X_{\gamma'} \rightarrow \X_{\gamma} : x_{\gamma'} \leadsto \prt{x_{\gamma'}}_{\vert\gamma} \qquad \forall \gamma\prec\gamma'$$
so we can define the projective limit:
$$\overline\X = \set{(x_\gamma)_{\gamma} : p_{\gamma'\gamma}(x_{\gamma'}) = x_\gamma \ \forall\,\gamma\prec \gamma'} \ \subset\ \bigotimes_{\gamma}\X_\gamma.$$
This space is just isomorphic to $\bar\A$ by the map
\begin{equation}\label{isoAX}
\bar\A\rightarrow \bar\X :  \mathfrak H \leadsto \prt{\mathfrak H_{\arrowvert_\gamma}}_{\gamma}
\end{equation}
so we can define our measure in the space $\bar\X$ instead.\\

Let us take $C(\X_\gamma)$ the space of functions on $\X_\gamma$ and their union \mbox{$\cup_{\gamma}C(\X_\gamma)$}. If we take two functions \mbox{$f,f'\in\cup_{\gamma}C(\X_\gamma)$}, then $\exists \,\gamma,\gamma'$ such that \mbox{$f\in C(\X_\gamma)$} and \mbox{$f'\in C(\X_{\gamma'})$} and we can define an equivalence relation by
$$f\sim f' \lequi p^*_{\gamma''\gamma} f = p^*_{\gamma''\gamma'} f'\quad \forall\, \gamma, \gamma' \prec \gamma''$$
where $p^* $ is the pullback map.\footnote{i.e.~such that \mbox{$(p^*_{\gamma''\gamma} f )(x_{\gamma''}) = f(p_{\gamma''\gamma}(x_{\gamma''}))$}. }\\

The quotient space
$$Cyl(\overline\X) =  \mbox{$\cup_{\gamma}C(\X_ \gamma)$}/\sim$$
is called the space of {\bf\index{function!cylindrical} cylindrical functions} and its closure $\overline{Cyl(\overline\X)}$ is a $C^*$-algebra. We have then two isomorphisms, one given by the Gel'fand transform
$$\bigvee : \overline{Cyl(\overline\X)} \rightarrow C(\Delta(\overline{Cyl(\overline\X)}))$$
and the other by the map
$$\overline\X \rightarrow \Delta(\overline{Cyl(\overline\X)}) : x = (x_ \gamma)_{\gamma} \leadsto \chi(x)$$
 $$\text{such that }\ [\chi(x)](f) = f(x_\gamma)\ \text{ if }\ f\in C(\X_\gamma)$$
 so at the end we have the isomorphism $\overline{Cyl(\overline\X)} \cong C(\overline\X)$.\\

Now let us look at the set of spaces $\X_\gamma$. We can notice that each \mbox{$x_\gamma \in \X_\gamma$} is completely determined by its values on the edges \mbox{$x_\gamma(e) \in SU(2)$}, $e\in E(\gamma)$. So there is a bijection
 $$\X_l \rightarrow \text{SU}(2)^{\sharp E(\gamma)} : x_\gamma \leadsto \prt{x_\gamma(e)}_{e\in E(\gamma)}.$$
 Now we can use the facts that $SU(2)$ is a compact Hausdorff space and that there is a finite number of edges on every graph to say that \mbox{$\X_l \cong \text{SU}(2)^n$} can be endowed with a compact Hausdorff topology and with a finite measure $\mu_\gamma$ given by the Haar measure. The projective limit $\bar\X$ becomes compact in the product topology. With this family of measures $\mu_\gamma$ we can define a linear functionnal on the space of cylindrical functions by
$$\Lambda : Cyl(\overline\X) \rightarrow \sC : f \in C(\X_\gamma) \leadsto \int_{\X_{\gamma}} f(x_{\gamma})\ d\mu_\gamma(x_{\gamma})$$
which can be extended continuously to the closure $\overline{Cyl(\overline\X)} \cong C(\overline\X)$.\\

The conclusion comes from the Riesz representation theorem\footnote{Riesz representation theorem: Let $X$ be a locally compact Hausdorff space, then for any positive linear functional $\psi$ on $C_0(X)$, there is a unique Borel regular measure $\mu$ on $X$ such that $\psi(f) = \int_X f(x)\, d\mu(x)$.} which guarantees the existence of a measure $\mu$ on $\overline\X$, and so on $\bar\A$ by the isomorphism  (\ref{isoAX}). This measure is called the Ashtekar--Lewandowski measure. The Hilbert space of quantum gravity is complete and given by
$$\H = L^2\prt{\bar\A,d\mu}.$$
The remarkable result is that an orthonormal basis on this Hilbert space can be given by spin networks (which are graphs such that  each edge is associated to an irreducible representation of $SU(2)$). All the states of quantum gravity can by this way be written in terms of spin networks.\\

So now we can remember our concern: can algebraization of geometry be used to create more general spaces which could include quantization or the other fundamental interactions? Quantum gravity gives a positive answer about the first possibility. The switch from classical variables to connections and the creation of the algebraic structure of cylindrical functions on the space of distributional connections allow us to create a well defined mathematical background that includes gravitational aspects and quantization. However, the theory of quantum gravity does not include any possibility to describe other interactions as electromagnetism or nuclear interactions, just because those possibilities are not in the scope of the theory.\\

Whereas quantum gravity seems to be a promising theory, we can wonder about the construction that leads to this algebraic formalism. First, the equations come from the ADM formalism, which is not the simplest we can find about gravitation. After the transformation to the Hamiltonian formalism, equations are transformed to a completely new set of variables. Then only the whole system is translated into an algebra of functions. The newly obtained algebraic formalism seems far away from the starting intuitive geometrical space. In other words, the newly obtained objects cannot be easily interpreted as usual geometric ones. This leads to the fact that, except a few informations obtained by some geometrical operator as the area operator, quantum gravity has a difficult interpretation in terms of geometry. There is no real way to go back from the algebraic objects to the geometrical counterparts.\\

We can now add the question: is it possible to construct a complete algebraic formulation of geometry in a more trivial way, in the sense that the obtained algebraic objects conserve a geometrical interpretation? Wa can remember the Gel'fand transform: spaces of functions on any manifold can be interpreted as abstract commutative $C^*$-algebras, and in the other way characters on any commutative $C^*$-algebra can be interpreted as points on some manifold. Typically we would like a complete correspondence between the geometrical formalism and the algebraic formalism, with each object on the one side having a (possibly abstract) interpretation on the other side. So we want a theory which includes an {\it algebraization of geometry} as well as a {\it geometrization of algebra}, and in such way that we can introduce gravitation via the geometrical interpretation and the other fundamental interactions via the algebraic interpretation. Nice program, but this description is no more than the fundamental idea supporting the link between  noncommutative geometry and physics.


\newpage

\section{Noncommutative geometry}

Now we come to the foundations of the theory of noncommutative geometry. To introduce that, we can remember what we have learned from the Section \ref{gelfsect}.\\

From any manifold $\M$ (or even from any locally compact Hausdorff space) we can consider the commutative $C^*$-algebra $\A=C_0(\M)$ of continuous functions and transpose some geometrical notions to it (as derivations for example). Then by considering the spectrum $\Delta(\A)$ we can recover the manifold itself with the isomorphism given by the Gel'fand transform (Definition \ref{geltransf}). Moreover, any commutative $C^*$-algebra will give rise to a Hausdorff space by considering the spectrum thanks to the equivalence of categories.\\

Of course any construction of this type requires to use commutative $C^*$-algebras, since they represent algebras of functions with product given by the necessarily commutative pointwise product, but we can wonder what could happen if we try to consider noncommutative $C^*$-algebras instead. Is it possible in this case to take the spectrum and to consider it as a geometrical space?\\

For that, we need to update our definition of  spectrum. Indeed, by the Definition \ref{defspectrum1}, the spectrum is the space of characters which are non-zero *-homomorphisms from $\A$ to $\setC$, but this definition does not make sense for noncommutative algebras, since in this case \mbox{$\chi(a)\chi(b) - \chi(b)\chi(a)$} could be different of zero. Instead, we will consider a more larger class of linear functionals.\\

\begin{defn}
If $\A$ is a $C^*$-algebra, a {\bf\index{state} state} is a positive linear functional of norm one, i.e.~a linear functional $\phi$ such that:
\begin{itemize}
\item $\phi(a^*a) \geq 0\qquad \forall a\in\A\qquad$($a^*a$ is called a positive element)
\item $\norm{\phi} = 1$, which implies $\phi(1)=1$ if $\A$ is unital
\end{itemize}
The space of states of $\A$ is denoted by $S(\A)$. Every state on a \mbox{$C^*$-algebra} is automatically continuous.
\end{defn}

The space $S(\A)$ is clearly a convex set, since for any $\phi_1, \phi_2 \in S(\A)$ and for $0 \leq \lambda \leq 1$ we have \mbox{$\lambda \phi_1 + (1-\lambda)\phi_2 \in S(\A)$}.

\begin{defn}
A {\bf\index{state!pure} pure state} is an element of $S(\A)$ that cannot be written as a convex combination of two other states. So pure states are extreme points of $S(\A)$.
\end{defn}

To show the interest of states of $C^*$-algebras, we need to introduce the theory of representation.

\begin{defn}
A {\bf\index{representation} representation} of a $C^*$-algebra $\A$ is a pair \mbox{$(\H,\pi)$} where $\H$ is a Hilbert space and $\pi$ is a *-homomorphism from $\A$ into $\B(\H)$, the $C^*$-algebra of bounded operators on $\H$. A representation is called {\bf\index{representation!faithful} faithful} if $\pi$ is injective.
\end{defn}

\begin{defn}\label{reprirr}
A representation \mbox{$(\H,\pi)$} is {\bf\index{representation!irreducible} irreducible} if there is no closed subspace of $\H$ which is invariant under the action of $\pi(\A)$ except the trivial spaces $\set{0}$ and $\H$.
\end{defn}

\begin{defn}\label{reprcycl}
A representation \mbox{$(\H,\pi)$} is {\bf\index{representation!cyclic} cyclic} if there is an element $\xi\in\H$ called a {\bf\index{cyclic vector} cyclic vector} such that the orbit \mbox{$\set{\pi(a)\xi : a\in\A}$} is dense in $\H$.
\end{defn}

We can remark that any non-cyclic vector would generate by the closure of its orbit a closed invariant space under the action of $\pi(\A)$ different from $\H$, and different from $\set{0}$ if the vector is not null. So a trivial consequence from Definitions \ref{reprirr} and \ref{reprcycl} is that every irreducible representation is cyclic, with every non-zero vector of $\H$ being cyclic, and reciprocally, if every non-zero vector is cyclic, the representation is irreducible since there is no nontrivial invariant space. We have also another characterization of irreducible representation.

\begin{lemma}[Schur]\label{propcommut}
A representation \mbox{$(\H,\pi)$} is irreducible if and only if \mbox{$\pi(\A)' = \setC . \text{id}$} where $\pi(\A)'$ is the commutant of $\pi(\A)$, i.e.~the set of elements in $\B(\H)$ which commute with each element in $\pi(\A)$.
\end{lemma}
\begin{proof}
Let us assume that \mbox{$\pi(\A)' = \setC . \text{id}$} and let us take a closed subspace $\K\subset\H$ invariant under the action of $\pi(\A)$. Then the projector $P_\K$ on $\K$ commutes with $\pi(\A)$, and consequently is of the form $P_\K = \lambda\, \text{id}$ with $\lambda\in\setC$. Since the projector is idempotent, $\lambda \in \set{0,1}$ so $\K$ is a trivial space. 

To prove the the reverse property, let us suppose that \mbox{$(\H,\pi)$} is an irreducible representation and let us take an element $T\in\B(\H)$ which commutes with $\pi(\A)$. Since $T$ can be decomposed into a Hermitian and an antiHermitian part, we can restrict to Hermitian operators and take the spectral decomposition \mbox{$T = \int \lambda\, dP(\lambda)$}. Each projector $P(\lambda)$ commutes with $\pi(\A)$ so it must be equal to $0$ or $1$, the spectrum is then restricted to one point and \mbox{$T = \lambda \,\text{id}$}.
\end{proof}

\begin{defn}
Two representations \mbox{$(\H_1,\pi_1)$} and \mbox{$(\H_2,\pi_2)$} are {\bf\index{representation!equivalent} equivalent} if there exists a unitary operator \mbox{$U : \H_1 \rightarrow \H_2$} such that \mbox{$\pi_2(a) = U \pi_1(a)\, U^*$} for any $a\in\A$.
\end{defn}

States of a $C^*$-algebra and representations of that algebra are  related by the following construction called the GNS-representation.

\begin{thm}[Gel'fand-Naimark-Segal representation]
Given a state $\phi$ of a $C^*$-algebra $\A$, there is a cyclic representation \mbox{$(\H_\phi,\pi_\phi)$}  of $\A$ with cyclic vector $\xi_\phi$ such that \mbox{$\phi(a) = \scal{\xi_\phi , \pi_\phi(a) \xi_\phi}$} for every $a\in\A$.
\end{thm}

\begin{proof}
First, we can observe that
$$\scal{a,b} = \phi(a^*b)$$
defines a sesquilinear form which is positive semidefinite (i.e.~\mbox{$\phi(a^*a) \geq 0\  \forall a\in\A$}). So this form must respect the Cauchy--Schwarz inequality \mbox{$\phi(a^*b)^2 \leq \phi(a^*a)\;\phi(b^*b)$}.

Let us define the set $\N_\phi = \set{a\in\A\; |\; \phi(a^*a) = 0}$. Then $\N_\phi$ is a closed left ideal in $\A$. Indeed, for $x,y\in\N_\phi$ and $a\in\A$ we have:
\begin{itemize}
\item \mbox{$0 \leq \phi\prt{(x+y)^*(x+y)} = \phi(x^*x) + \phi(y^*y) + 2 \,\mathfrak{Re}\!\prt{\phi(x^*y)} \leq 0$} \mbox{$\implies \ x+y\in\N_\phi$}
\item \mbox{$\phi\prt{(ax)^*(ax)}^2 = \phi\prt{x^*a^*ax}^2 \leq \phi(x^*x)\;\phi\prt{(a^*ax)^*(a^*ax)} = 0 $}  \mbox{$\implies \ ax\in\N_\phi$}
\end{itemize}
where we have applied each time the Cauchy--Schwarz inequality.

So we can construct the quotient space $\A/\N_\phi$  which turns out to be a pre-Hilbert space with the positive definite Hermitian inner product $\scal{[a],[b]} = \phi(a^*b)$ which is independent of the representatives in the equivalence classes. The Hilbert space $\H_\phi$ is the completion of $\A/\N_\phi$ by the norm defined by the inner product.

To each $a\in\A$, we associate the operator \mbox{$\pi(a)$} on $\A/\N_\phi$ defined by:
$$\pi(a)\,[b] = [ab].$$

If we denote by $a \leq b$ the fact that $\sigma(b-a) \subset [0,\infty)$, then $a^*a \geq 0$ and if $\A$ is unital (otherwise we perform it throughout $\A^+$) by \mbox{Lemma \ref{rayspect}} we have:
$$0 \leq a^*a \leq \norm{a^*a} 1.$$
For $b\in\A$ the map \mbox{$a \leadsto b^*ab$} preserves positivity, so the inequality becomes:
$$ b^*a^*ab \leq \norm{a^*a} b^*1b$$
and by applying the positive linear functional $\phi$:
$$ \phi(b^*a^*ab) \leq \norm{a^*a} \phi(b^*b) = \norm{a}^2 \phi(b^*b)$$
$$\implies \frac{\norm{\pi(a)\,[b]}}{\norm{[b]}} = \frac{\sqrt{\phi(b^*a^*ab)}} {\norm{[b]}}\leq \norm{a} \frac{\phi(b^*b)}{{\norm{[b]}}} = \norm{a}$$
which implies that the operator \mbox{$\pi(a)$} is bounded on $\A/\N_\phi$, and can be uniquely extended to a bounded operator in $\H_\phi$. The *-homomorphism properties of $\pi_\phi$ can easily be checked.

The cyclic vector  is simply the unity $\xi_\phi = [1]$. If $\A$ is not unital, we can take $\xi_\phi = \lim_\alpha [u_\alpha]$ where $\set{u_\alpha}$ is an approximate unit, i.e.~an increasingly ordered net of positive elements in the closed unit ball of $\A$ such that \mbox{$\forall a\in\A,\ a = \lim_\alpha a \,u_\alpha$} (such approximate unit always exists). We trivially have \mbox{$\phi(a) = \phi(\xi_\phi^*\,(a\,\xi_\phi)) = \scal{\xi_\phi , \pi_\phi(a) \xi_\phi}$}.
\end{proof}

Of course any cyclic representation \mbox{$(\H,\pi)$} defines a state $\phi$ by the formula \mbox{$\phi(a) = \scal{\xi , \pi(a) \xi}$} with $\xi$ a cyclic vector of norm one. In this case, the GNS-representation \mbox{$(\H_\phi,\pi_\phi)$} is equivalent to \mbox{$(\H,\pi)$} by use of the unitary operator $U$ defined by \mbox{$U \pi(a)\xi = \pi_\phi(a)\xi_\phi$}. So we can say that the GNS-representation is surjective among the equivalence classes of cyclic representations.\\

From this GNS-representation, we can obtain the following theorem also from Gel'fand and Naimark, which can be interpreted as a generalization of the Theorem \ref{GelNai} for arbitrary $C^*$-algebras:

\begin{thm}[Gel'fand--Naimark]
Any C*-algebra $\A$ is isometrically \mbox{*-isomorphic} to a closed subalgebra of the algebra $\B(\H)$ of bounded operators on some Hilbert space $\H$.
\end{thm}

\begin{proof}
Let us take a non-zero element $a\in\A$. We denote by $\A_+\subset \A$ the closed convex cone of positive elements (i.e.~in the form $b^*b$). The element $-a^*a\in\A$ is clearly separate from $\A_+$. We can apply the \mbox{Hahn--Banach}\footnote{Hahn--Banach separation theorem:
Let $V$ be a topological vector space and \mbox{$A$, $B$} convex, non-empty subsets of $V$ with \mbox{$A \cap B = \emptyset$} and $A$ open, then there exists a continuous real linear map such that \mbox{$\lambda(a) < 0 \leq \lambda(b)\ \forall a\in A, b\in B$}.} separation theorem to find a real linear continuous form $f_a$ such that \mbox{$f_a(x) \geq 0$} for all \mbox{$a\in\A_+$} and \mbox{$f_a(-a^*a)<0$}. $f_a$ is positive by definition, with $f_a(a^*a)>0$ and can be rescaled to be a state. So we can take the GNS-representation \mbox{$(\H_{f_a},\pi_{f_a})$} with the cyclic vector $\xi_{f_a}$ and  find:
$$ \norm{\pi_{f_a}(a) \;\xi_{f_a}}^2 = \scal{\xi_{f_a} , \pi_{f_a}(a^*a) \; \xi_{f_a}} = f_a(a^*a) > 0.$$ 
So for every non-zero element $a\in\A$ we have a representation $(\H_a,\pi_a)$ with \mbox{$\pi_a(a)\neq 0$}. Then we can build the faithful representation:
$$ \prt{ \H = \bigoplus_{a\in\A\setminus\set{0}} \H_a\ ,\ \pi = \bigoplus_{a\in\A\setminus\set{0}} \pi_a  }\cdot$$

The isometry comes from the fact that every injective *-morphism of $C^*$-algebras is an isometry. To show this last assumption, let us consider an injective *-morphisms $f$ between two unital $C^*$-algebras $A$ and $B$ (we use the unitizations if it is not initially the case). Let us take a Hermitian element $x\in A$ and $A_x$, $B_{f(x)}$ the unital sub $C^*$-algebras generated by $x$ and $f(x)$. Because $x$ is Hermitian, those subalgebras are commutative and are isomorphic to the algebras of continuous functions on their spectrum by the Gel'fand--Naimark theorem (Theorem \ref{GelNai}). Let us consider the pullback map \mbox{$ f^* : \Delta(B_{f(x)}) \rightarrow \Delta(A_x)$}. This map is surjective, because if it is not the case there must exist a continuous function $\hat a$ non zero on all $\Delta(A_x)$ but zero on $f^*\prt{\Delta(B_{f(x)})}$, so a non-zero element $a\in A$ such that $f(a)=0$, which contradicts the injectivity of $f$. So by use of {Lemma \ref{rayspect}} we have for Hermitian elements
\begin{eqnarray*}
\norm{f(x)} &=& r(f(x)) = \sup_{\phi\in\Delta(B_{f(x)})} \abs{\phi(f(x))} = \sup_{\phi\in\Delta(B_{f(x)})} \abs{f^*(\phi)(x)}\\
&=& \sup_{\psi\in\Delta(A_x)} \abs{\psi(x)} = r(x) = \norm{x}
\end{eqnarray*}
and for non Hermitian elements
\begin{equation*}
\norm{x}^2 = \norm{x^*x} = \norm{f(x^*x)} = \norm{f(x)}^2.\qedhere
\end{equation*}
\end{proof}

We have seen that any commutative $C^*$-algebra is isometrically \mbox{*-isomorphic} to an algebra of continuous functions by the first \mbox{Gel'fand--Naimark} theorem, and that in the general case a $C^*$-algebra is isometrically \mbox{*-isomorphic} to an algebra of bounded operators on a Hilbert space by the second one. In the commutative case we have defined the spectrum to be the space of characters. The following theorems will allow us to make a new definition compatible with the commutative case.

\begin{thm}
If $\A$ is a commutative $C^*$-algebra, the GNS-representation \mbox{$(\H_\phi,\pi_\phi)$} is irreducible if and only if $\phi$ is a character.
\end{thm}

\begin{proof}
Since $\A$ is commutative, $\pi_\phi(\A) \subset \pi_\phi(\A)' $ (commutant) so $\pi_\phi(\A) \subset \setC .\text{id}$ by Lemma \ref{propcommut}. That means that every representation is just a \mbox{*-homomorphism} \mbox{$\pi_\phi : \A \rightarrow \setC$}, so \mbox{$\phi(a) = \scal{1,\pi_\phi(a)1} = \pi_\phi(a)$} is a character.
\end{proof}

\begin{thm}
For any $C^*$-algebra $\A$, the GNS-representation \mbox{$(\H_\phi,\pi_\phi)$} is irreducible if and only if $\phi$ is a pure state.
\end{thm}

\begin{proof}
First, we begin by considering an equivalent definition of pure states. We say that a state $\phi$ is a pure state if and only if $f\neq 0$ and, for every positive continuous form $\psi$ on $\A$ such that $\phi$ is a majorant of $\psi $, $\psi = \lambda\, \phi$ for some $\lambda\in[0,1]$.

To establish this new definition, let us take a pure state $\phi$ majorant of a positive continuous form $\psi_1$, so $\phi = \psi_1 + \psi_2$ with $\psi_2$ another positive continuous form. If $\norm{\psi_1} = \lambda$ then $\norm{\psi_2} = 1 - \lambda$, and $\phi$ can be written as the convex combination
$$\phi = \lambda\prt{\frac{\psi_1}{\lambda}} +  (1-\lambda)\prt{\frac{\psi_2}{1-\lambda}} $$
and since $\phi$ is extremal we must have 
\vspace{-0.2em}                                                      
$$\phi = \prt{\frac{\psi_1}{\lambda}} = \prt{\frac{\psi_2}{1-\lambda}} \quad\implies\quad \psi_1 = \lambda\, \phi.$$
\vspace{-1em}                                                      

Let us assume that, for each $\psi$ whose $\phi$ is a majorant, \mbox{$\psi = \mu\, \phi$}, $\mu\in[0,1]$, and let us assume \mbox{$\phi = \lambda\, \psi_1 + (1-\lambda)\, \psi_2$} with \mbox{$0<\lambda<1$} and \mbox{$\norm{\phi} = \norm{\psi_1}=\norm{\psi_2}=1$} (states). Then $\phi$ is a majorant of $\lambda\, \psi_1$, so \mbox{$\lambda\, \psi_1 = \mu\, \phi$}. By taking the norm, $\ \lambda\, \norm{\psi_1} = \mu\, \norm{\phi}\ $ implies $\ \mu=\lambda\ $ and $\phi=\psi_1=\psi_2$.

Now we can come to the proof and assume $\phi$ to be a pure state. If $P$ is a projector in $\H_\phi$ (so Hermitian and idempotent) which commutes with $\pi_\phi(\A)$, then \mbox{$\psi(a) = \scal{  P\, \xi_\phi , \pi_\phi(a) P\, \xi_\phi}$} defines a positive continuous form on $\A$ with $\phi$ as majorant, so $\psi = \lambda \phi$ with $\lambda\in[0,1]$. So \mbox{$\scal{ \xi_\phi ,  \pi_\phi(a)\,P^2\, \xi_\phi} = \scal{  \xi_\phi ,  \pi_\phi(a)\, \lambda\,\xi_\phi}$} and because $\xi_\phi$ is a cyclic vector we obtain \mbox{$P = \sqrt{\lambda}\,\text{id}$} which gives $P = 0$ or $P = 1$ by idempotence. So the only invariant spaces are trivial.

For the reverse, let us assume that \mbox{$(\H_\phi,\pi_\phi)$} is an irreducible representation. If $\psi$ is a continuous positive form on $\A$ with $\phi$ as majorant, by an extension of Radon--Nikodym theorem (cf. \cite{DixmierALG}), $\psi$ can be written as \mbox{$\psi(a) = \scal{ T \,\xi_\phi , \pi_\phi(a) \,T \,\xi_\phi}$} with some Hermitian operator $T\in\pi_\phi(\A)'$ such that \mbox{$0 \leq T \leq \text{id}$} (such that $T$ and $\prt{\text{id}-T}$ are both semidefinite positive). By the Lemma \ref{propcommut}, \mbox{$\pi_\phi(\A)' = \setC . \text{id}$} so $T= \lambda\, \text{id}$ with $\lambda$ constraint in the interval $[0,1]$. Hence \mbox{$\psi = \lambda^2 \phi$} which proves that $\phi$ is a pure state.
\end{proof}

\vspace{-0.2em}                                                      

Since every irreducible representation is cyclic, the \mbox{GNS-representation} gives a surjection between the pure states and the equivalence classes of irreducible representations. We can suggest the following definition.

\vspace{-0.2em}                                                      

\begin{defn}\label{defspectrum2}
Let $\A$ be a $C^*$-algebra. The {\bf\index{spectrum!of a $C^*$-algebra} spectrum $\Delta(\A)$} of $\A$ is the set of all equivalence classes of irreducible representations\footnote{One could also define the spectrum as the space of kernels of irreducible representations, which gives a similar space in the commutative case but not necessarily in the noncommutative one.} of $\A$.
\end{defn}

With this definition, the spectrum corresponds to the space of pure states on the algebra $\A$. If the algebra is commutative, then the pure states are just the characters of the algebra.                                                     

We now have our extension of the notion of spectrum. We can remember that the spectrum, in the commutative case, is identified with the {\it points} of a Hausdorff space, and possibly with the points of a manifold if there is sufficient continuity and differentiability. So the natural extension is to consider pure states as {\it points} of some kind of manifold or space, and to identify the elements of the algebra to {\it functions} on this manifold. Of course this is not realistic geometrically, since by this way we obtain a space of functions which do not commute! So we have to do here some mind gymnastic, and imagine that a noncommutative algebra could be a space of noncommutative functions on a certain space given by the pure states, a space with a geometry which becomes noncommutative.\\

\vspace{0.5em}                                                      

Let us translate out this reasoning into a figure, by updating our old Figure \ref{Gelf_fig} from the Section \ref{gelfsect} to a new one (Figure \ref{GelfNC_fig}). Now, only the right side is correctly defined, with an algebraic space of operators  and a set of pure states. The left side is a virtually "noncommutative" manifold with a virtually set of noncommutative functions on it.

\vspace{0.5em}                                                      

\begin{figure}[!ht]
\begin{center}
\includegraphics[width=10cm]{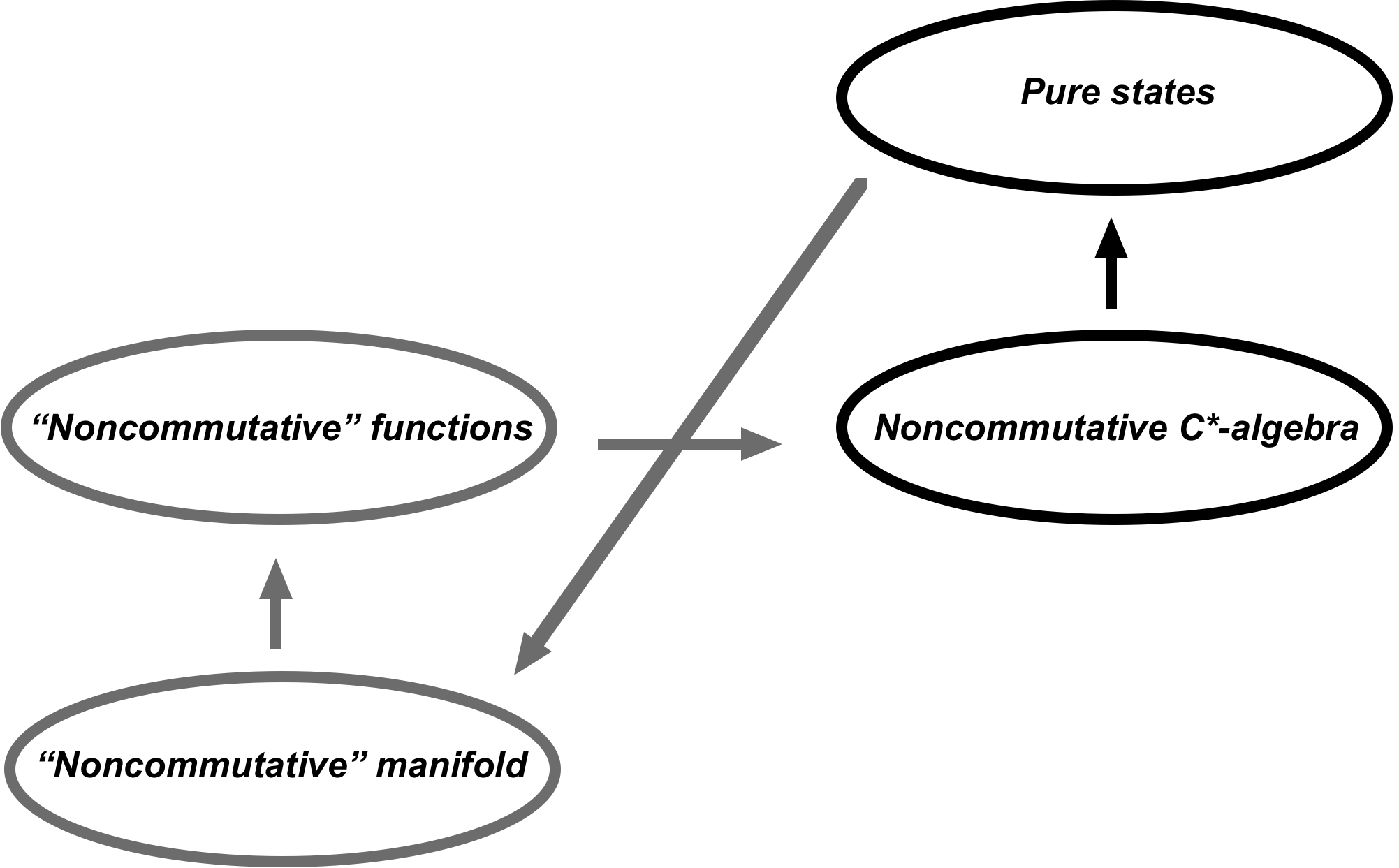}
\end{center}
\caption{The conception of noncommutative manifolds}\label{GelfNC_fig}
\end{figure}

\vspace{0.5em}                                                      

The beauty of this figure is that we can read it from left to right and from right to left.
\begin{itemize}
\item From left to right: We can see that the algebraic part comes directly from the geometric one. There is a priori no geometrical concept relating to noncommutative algebras, but we can translate geometrical concepts from the left side in the commutative case to the right side in the commutative case, and then extend them to noncommutativity. This is the primarily concern of noncommutative geometry: translating concepts  from commutative geometry to noncommutative algebras in order to create new noncommutative spaces with geometrical tools. This is our {\it algebraization of geometry}.
\item From right to left: The new tools translated in the algebraic formalism can be abstractly identified with a geometrical signification, since we can consider the space of pure states as our new geometrical space. So elements which can only be defined with the help of noncommutative algebraic framework can now have a signification as geometrical elements.  This is our {\it geometrization of algebra}.
\end{itemize}

The program of noncommutative geometry is simple in its concept: translating geometrical tools into an algebraic formalism, and extending these tools to the noncommutative world. This leads to the creation of a {\it dictionary} which gives correspondences between geometrical and algebraic concepts.\\

We have already seen an important element of this dictionary: the correspondence between locally compact spaces and commutative $C^*$-algebra. We have another fundamental one we will not develop here (see e.g.~\cite{Var}) which is given by the Serre--Swan theorem: there is an equivalence of categories between the category of vector bundles over a compact space $\M$ and the category of finitely generated projective\footnote{
A module $P$ is projective if there exists a module $Q$ such that the direct sum of the two is a free module $F = P \oplus Q$, where a free module is a module with a generating set of linear independent elements.
} modules over $C(\M)$.\\

A lot of other elements can be added, and we will review and construct some of them in the forthcoming chapters. We give here a non exhaustive extract of this dictionnary, which is growing up from year to year:\\

\begin{center}\begin{tabular}{rcl}
{\bf Geometry} &  & {\bf Algebra} 
\vspace{0.2cm}\\
\hline\\
locally compact space & $\longleftrightarrow$ & C*-algebra \\
compact & $\longleftrightarrow$ & unital \\
Alexandroff compactification & $\longleftrightarrow$ & unitization \\
Stone--Cech compactification & $\longleftrightarrow$ & multiplier algebra \\
point & $\longleftrightarrow $ & pure state \\
open subset & $\longleftrightarrow$ & ideal \\
dense open subset & $\longleftrightarrow$ & essential ideal \\
closed subset & $\longleftrightarrow$ & quotient algebra \\
surjection & $\longleftrightarrow$ & injection \\
injection & $\longleftrightarrow$ & surjection \\
homeomorphism & $\longleftrightarrow$ & automorphism \\
metrizable & $\longleftrightarrow$ & separable \\
Borel measure & $\longleftrightarrow$ & positive functional \\
probability measure & $\longleftrightarrow$ & state \\
measure space & $\longleftrightarrow$ & von Neumann algebra \\
vector field & $\longleftrightarrow$ & derivation \\
fiber bundle & $\longleftrightarrow $ & finite projective module \\
compact Riemannian manifold & $\longleftrightarrow$ & unital spectral triple \\
complex variable & $\longleftrightarrow $ & operator \\
real variable & $\longleftrightarrow $ & Hermitian operator \\
infinitesimal & $\longleftrightarrow $ & compact operator \\
integral & $\longleftrightarrow $ & Dixmier trace \\
range of a function & $\longleftrightarrow$ & spectrum of an operator \\
de Rham cohomology & $\longleftrightarrow$ & cyclic homology \\
& ... & \\
\end{tabular}\end{center}

\vspace{1em}                                                      

The second chapter of our dissertation will deal with the special case of compact Riemannian manifolds. We will see which elements can be constructed in order to translate the geometry of a compact Riemannian manifold into an algebraic formalism, and then how we can construct from that a model which combines  a description of the standard model and Euclidean gravity.\\

 The third chapter will be dedicated to the problem of generalization of these elements to non-compact Lorentzian manifolds.\\

There are many books or general articles dealing with the basis of noncommutative geometry. We give here a list of those which have helped us to write this chapter: \cite{C94,MC08,DixmierALG,Var,Kalk,Landi97,Landi02,Mass96}.


\cleardoublepage
\hbox{} \vspace*{\fill} \thispagestyle{empty}
\chapter[Euclidean noncommutative geometry]{The current model of Euclidean noncommutative geometry}\label{chaprie}

In this chapter, we introduce the theory of noncommutative geometry sometimes referred as noncommutative geometry "à la Connes". From about 30 years now, A. Connes and his collaborators have developed a huge amount of tools for noncommutative geometry. Among those tools we can highlight the following elements:
\begin{itemize}
\item Algebraization and generalization to noncommutative spaces of some elements from the differential structure of Riemannian geometry
\item Algebraization and generalization to noncommutative spaces of the notion of Riemannian distance
\item Construction of a noncommutative model which combines Euclidean gravitation and a classical standard model of particles
\end{itemize}
We will develop these elements throughout this chapter. A first section will be dedicated to the machinery of differential calculus in noncommutative geometry, aiming at the creation of the structure of "spectral triple". The second section will be an overview of the application of the notion of spectral triple to reproduce the standard model of particles combined with Euclidean gravity.


\newpage

\section[Noncommutative differential Riemannian geometry]{Noncommutative differential Riemannian geometry}\label{secgrav}

\vspace{0.5em}                                                      

We remember the main concern of noncommutative geometry: transcribing existing geometrical elements into an algebraic formalism in order to extend them to noncommutative algebras. From now and for the whole chapter we will consider extensions of elements from compact Riemannian manifolds.\\

We will work mainly with an algebra $\A \subset \B(\H)$ of bounded operators acting on some infinite-dimensional separable Hilbert space $\H$. When the algebra will be supposed to be commutative it will be the algebra $C^\infty(\M)$ of continuous functions over a compact $n$-dimensional Riemannian manifold $\M$ (with $n\geq 3$) with metric $g$ (and since the completion $\overline{C^\infty(\M)} = C(\M)$ is a $C^*$-algebra, $C^\infty(\M)$ can always be interpreted as an algebra of bounded operators thanks to the GNS-representation).\\

 Since our goal is not to rewrite a complete book on the subjet we will not present all the details of this theory, and we refer the reader to the classical books \cite{C94,MC08,Var}.\\

\subsection{Noncommutative infinitesimals}

\vspace{0.5em}                                                      

One main idea is to define in fine a notion of integral on noncommutative spaces. Because the functions are replaced by operators in noncommutative geometry, such integral should take the form of a trace of some kind of infinitesimal operators, so we will start to define what can be a noncommutative infinitesimal. Moreover, we want to define an order for those infinitesimals so that the integral will be defined for each infinitesimal of order one and will vanish for the others.\\

By their name, infinitesimal operators should be operators which are {\it as small as possible}. Since it is impossible for a non-null operator $T$ to require $\norm{T}< \epsilon$ for any $\epsilon>0$, where \mbox{$\norm{\,\cdot\,}$} is the operator norm\footnote{$\norm{T} = \sup_{\phi\in\H,\;\norm{\phi}\leq 1} \norm{T\,\phi} = \sup_{\phi\in\H,\;\phi\neq 0} \frac{\norm{T\,\phi}}{\norm{\phi}}$.} on $\H$, we can try to require this while removing some finite-dimensional subspace of $\H$. So we propose the following definition:

\begin{defn}
An {\bf\index{infinitesimal} infinitesimal} is an operator $T$ on $\H$ such that, for all $\epsilon > 0$, there exists a finite-dimensional subspace $E\subset \H$ such that $\norm{T_{\vert  E^\perp}} < \epsilon$.
\end{defn}

We would like to characterize  more precisely these infinitesimals by using the spectral theory of compact operators, which has the following well-known results:

\begin{thm}[Spectral theorem]
Let $T\in\K(\H)$, where $\K(\H)$ is the space of compact operators on $\H$. Then its spectrum $\sigma(T)$ is either a finite set, or a discrete sequence $\set{\mu_n}_{n\geq 0}$ such that $\mu_n\rightarrow 0$.
\end{thm}

\begin{thm}[Spectral decomposition]
Let $T\in\K(\H)$. Then $T$ admits an expansion convergent in norm given by
$$ T = \sum_{n\geq 0} \mu_n \,| \psi_n\rangle \langle \phi_n |$$
where $\mu_n\in\sigma(\sqrt{T^*T})$ are the singular values of  $\,T$ ordered as a decreasing sequence and $\set{\psi_n}_{n\geq 0}$, $\set{\phi_n}_{n\in\geq 0}$ are two orthonormal sets.
\end{thm}

\begin{proof}
There is a unique partial isometry $U$ such that $T = U\sqrt{T^*T}$ (polar decomposition). $T^*T$ is a positive self-adjoint compact operator which admits the usual spectral decomposition for self-adjoint operators:
$$ T^*T = \sum_{n\geq 0} \mu_n^2 \,| \phi_n\rangle \langle \phi_n |$$
with $\mu_n^2\in\sigma(T^*T)$ and $\phi_n$ the associate eigenvectors. We set \mbox{$\psi_n=U\phi_n$} to obtain the decomposition for non-self-adjoint compact operators.
\end{proof}

We can notice that, as a consequence of the decreasing order, we have $\mu_0 = \norm{T}$.

\vspace{1em}                                                      

\begin{prop}
$T$ is an infinitesimal if and only if $T\in\K(\H)$.
\end{prop}

\begin{proof}
Since the singular values $\mu_n\in\sigma(\sqrt{T^*T})$ form a decreasing sequence with limit zero, then $\forall \epsilon > 0$ there exists $k\in\setN$ such that
$$\norm{T-\sum_{n\leq k} \mu_n \,| \psi_n\rangle \langle \phi_n |} = \norm{\sum_{n> k} \mu_n \,| \psi_n\rangle \langle \phi_n |} < \epsilon.$$
We can take $E = \text{span}\set{\phi_n : n\leq k}$ to have $\norm{T_{\vert  E^\perp}} < \epsilon$.
\end{proof}

We want to characterize infinitesimals by their order, and since every infinitesimal has a decreasing sequence of singular values, we want to use the rate of decay of these singular values to determine the order of the infinitesimal.

\begin{defn}
An infinitesimal $T$ with singular values $\mu_n(T)$ is of {\bf\index{infinitesimal!order of} order $\alpha$} $\in\setR^+$ if \mbox{$\mu_n(T) = O(n^{-\alpha})$}, i.e.~if
$$\exists\, C < \infty : \ \mu_n(T) \leq C\,n^{-\alpha}\quad \forall n\in\setN.$$
\end{defn}

This definition is coherent since by \cite{Simon} we have the property
$$ \mu_{n+m}(T_1 T_2) \leq \mu_n(T_1)\ \mu_m(T_2)$$
so if $T_1$ is an infinitesimal of order $\alpha$ and $T_2$ an infinitesimal of order $\beta$, $T_1 T_2$ is an infinitesimal of order at most $\alpha + \beta$.\\

\subsection{The Dixmier trace}

We have suggested that a noncommutative integral should take the form of a trace on infinitesimal operators. We want to define a trace on the space of compact operators which neglects operators of order greater than one, so a linear functional \mbox{$\tr : \K(\H) \rightarrow \setC$} with $\tr(T_1T_2)=\tr(T_2T_1)$ such that $\tr(T)\neq0$ only if $\mu_n(T) = O(\frac 1n)$.\\ 

A first idea would be to take  the usual trace for compact operators
$$ \norm{T}_1 = \tr(\abs{T}) = \sum_{n \geq 0} \mu_n(T) $$ 
in the space $\L^1 = \set{T : \sum \mu_n(T) < \infty}$. This trace defines a new norm, with the usual properties of norms, and can be rewritten as
$$\norm{T}_1 = \lim_{N\rightarrow\infty} \sigma_N(T)$$
if we define the partial sums
$$\sigma_N(T) = \sum_{n = 0}^{N-1} \mu_n(T).$$
However, the space of infinitesimals of order 1 is larger than \mbox{$\mathcal L^1$} since the partial sums can be logarithmically divergent:
$$\mu_n(T) \leq C\,\frac 1n\quad \forall n\in\setN \quad\implies\quad \sigma_N(T) \leq C \ln N.$$

J. Dixmier \cite{Dixmiertrace} has shown that there exists a trace which corresponds to an extraction of the coefficient $C$:
\begin{equation}\label{extractC}
C = \lim_{N \rightarrow \infty} \frac{1}{\ln N} \sum_{n = 0}^{N-1} \mu_n(T) = \lim_{N \rightarrow \infty} \frac{\sigma_N(T)}{\ln N} .
\end{equation}
However the formula (\ref{extractC}) does not define a trace since linearity and convergence are not guaranteed. We must instead introduce what is called the Dixmier trace. We will partially follow \cite{CM95} for the construction of this trace. In the sequel we will assume $T\geq 0$ since we will be able to extend the trace to non-positive operators by linearity.

\begin{lemma}\label{lemmasup}
\begin{eqnarray*}
\sigma_N(T) &= &\sup\set{\norm{TP_E}_1 : E\subset \H\ \text{and}\  \dim E = N}\\
&=&\sup\set{\tr\prt{TP_E} : E\subset \H\ \text{and}\  \dim E = N}
\end{eqnarray*}
where $P_E$ is the projector and $\tr$ is the usual trace.
\end{lemma}
\begin{proof}
Since $TP_E \leq T$ we have \mbox{$\mu_n(TP_E) \leq \mu_n(T)$} for all $n$ and because \mbox{$\dim E = N$}:
$$ \norm{TP_E}_1 = \sum_{n=0}^{\infty} \mu_n(TP_E)  \leq  \sum_{n=0}^{N-1} \mu_n(T) = \sigma_N(T).$$
For the second expression, \mbox{$\tr(TP_E) \leq \tr(\abs{TP_E}) = \norm{TP_E}_1 \leq \sigma_N(T).$} The supremum is each time attained by taking \mbox{$E = \text{span}\set{\phi_n : n < N}$.}
\end{proof}

We define a continuous extension of the notion of partial sums, with the following proposition.

\begin{prop}
The following function of $\lambda\in\setR^*$ matches $\sigma_N(T)$ at integer values:
$$\sigma_\lambda(T) = \inf\set{\norm{R}_1 + \lambda \norm{S} : R + S = T , \ R,S\in\K(\H)}.$$
\end{prop}
\begin{proof}
First, we can remark that $\sigma_N$ obeys the triangle inequality
$$\sigma_N(R+S) \leq \sigma_N(R) + \sigma_N(S)$$
 since $\norm{\,\cdot\,}_1$ is a well defined norm respecting the triangle inequality and \mbox{$\norm{(R+S)P_E}_1 \leq \norm{RP_E}_1 + \norm{SP_E}_1$} for every subspace $E$ of dimension $N$. We also have that \mbox{$\sigma_N(T) \leq N \norm{T}$} since \mbox{$\sigma_N(T) \leq N \mu_0(T) = N \norm{T}$}.
 
Then, if $T = R + S$ we have
$$\sigma_N(T) \leq \sigma_N(R) + \sigma_N(S) \leq \norm{R}_1 + N \norm{S}.$$
The get the infimum, we can take $E = \text{span}\set{\phi_n : n < N}$ and set the decomposition $T = R + S$ with \mbox{$R = (T-\mu_N(T) 1) P_E$} and \mbox{$S=\mu_N(T) P_E + T(1-P_E)$} which gives \mbox{$\norm{R}_1 = \sigma_N(T) - N \mu_N(T)$} and \mbox{$\norm{S}=\mu_N(T)$}.
\end{proof}

\vspace{1em}                                                      

\begin{lemma}
If $T_1,T_2\geq 0$ and $\lambda\in\setR^*$, then
\begin{equation}\label{largelambda}
\sigma_\lambda(T_1 + T_2) \leq \sigma_\lambda(T_1 ) + \sigma_\lambda(T_2) \leq \sigma_{2\lambda}(T_1 + T_2).
\end{equation}
\end{lemma}
\begin{proof}
We just have to prove this inequality for integer values, and extend by piecewise linearity. The first part is just the triangle inequality. For the second, let us take two $n$-dimensional subspaces $E_1, E_2$ and a $2n$-dimensional subspace $E$ such that $E_1 + E_2 \subset E$, then
$$\tr\prt{T_1E_1} + \tr\prt{T_2E_2}  \leq \tr\prt{T_1E}  + \tr\prt{T_2E}  = \tr\prt{(T_1+T_2)E}. $$
We conclude by taking the supremum over all triplets $E_1,E_2,E$ and by using Lemma \ref{lemmasup}.
\end{proof}

\vspace{1em}                                                      

\begin{defn}
The {\bf\index{ideal!Dixmier} Dixmier ideal} is the norm subspace defined by:
$$\L^{1+}(\H) = \set{T\in\K(\H) : \norm{T}_{1+} = \sup_{\lambda>e} \frac{\sigma_\lambda(T)}{\ln\lambda} < \infty}\cdot$$
\end{defn}

By looking on the equation (\ref{largelambda}), we can extrapolate that $\frac{\sigma_\lambda(T)}{\ln\lambda}$ should behave like an additive functional while $\lambda$ is growing. We can look at the {\bf\index{Ces\`aro mean} Ces\`aro mean} of this quantity defined by
$$\tau_\lambda(T) = \frac 1{\ln\lambda} \int_a^\lambda \frac{\sigma_u(T)}{\ln u} \frac{du}{u} $$
for $\lambda\geq a$ and with $a>e$. Once more the convergence \mbox{$\lim_{\lambda\rightarrow\infty} \tau_\lambda(T)$} and the linearity are not guaranteed, but the behaviour of the Ces\`aro mean for large $\lambda$ is given by the following lemma:

\begin{lemma}
If $T_1, T_2 \in\L^{1+}(\H)$ are such that $\forall\lambda\geq a$, \mbox{$\sigma_\lambda(T_1) \leq C_1 \ln \lambda$} and \mbox{$\sigma_\lambda(T_2) \leq C_2 \ln \lambda$}, then
$$
\abs{\tau_\lambda(T_1+ T_2) - \tau_\lambda(T_1) - \tau_\lambda(T_2)} \leq \prt{C_1 + C_2} \frac{(\ln\ln\lambda + 2)\ln 2}{\ln \lambda}.
$$
\end{lemma}

The proof is only technical and can be found in  \cite{CM95}.\\

For every positive element $T\in\L^{1+}(\H)$, $\tau_\lambda(T)$ is a bounded function for the parameter $\lambda \in [a,\infty)$, i.e.~$\tau(T) = \tau_{(\cdot)}(T)\in C_b([a,\infty))$ where $C_b$ denotes the space of continuous bounded functions. Since \raisebox{0cm}[1.2\height][0.8\height]{$ \frac{(\ln\ln\lambda + 2)\ln 2}{\ln \lambda}$} is bounded and vanishes at infinity, it is an element of the ideal $C_0([a,\infty))\subset C_b([a,\infty))$. So in order to have an additive map, we define
$$\overline\tau(T) = [\tau(T)] \in C_b([a,\infty))/C_0([a,\infty))$$
as an element of the quotient space of bounded continuous functions modulo functions vanishing at infinity.\\

Now the space $\L^{1+}(\H)$ is generated by its positive elements since every self-adjoint element $T = U\abs{T}$ can be expressed as the difference of two positive operators \mbox{$T = P_+ \abs{T}P_+ - P_- \abs{T}P_-$} with $P_\pm$ the projectors on the eigenspaces of $U$ with respective eigenvalues $\pm 1$, and since self-adjoint elements generate the whole space by the formula \mbox{$T = \frac 12 (T+T^*) + \frac i2 (iT^* -iT)$}. So $\overline\tau$ can be extended to $\L^{1+}(\H)$ by setting \mbox{$\overline\tau(T) = \overline\tau(P_+ \abs{T}P_+) - \overline\tau(P_- \abs{T}P_-)$} on self-adjoint elements and \mbox{$\overline\tau(T) = \frac 12 \overline \tau(T+T^*) + \frac i2 \overline\tau(iT^* -iT)$} for non-self-adjoint elements. From the additivity of $\overline\tau$ on positive elements, one can check that $\overline\tau$ is linear on $\L^{1+}(\H)$, with \mbox{$\overline\tau(T_1T_2) = \overline\tau(T_2T_1)$}.\\

The space $C_b([a,\infty))/C_0([a,\infty))$ is a commutative $C^*$-algebra, so has a good number of states generated from the characters, each state being a positive linear map \mbox{$\omega : C_b([a,\infty)) \rightarrow \setC$} vanishing on $C_0([a,\infty))$ and with $\omega(1)=1$.  A state $\omega$ corresponds to a generalized limit for bounded but not necessarily convergent functions. Since $\overline\tau$ is a map \mbox{$\L^{1+}(\H) \rightarrow C_b([a,\infty))/C_0([a,\infty))$}, we can combine it with $\omega$ to obtain a linear functional $\L^{1+}(\H) \rightarrow \setC$.\\                                                      

\begin{defn}
The {\bf\index{Dixmier trace} Dixmier trace} associated with the state $\omega$ is the trace operator:
$$\tr_\omega : \L^{1+}(\H) \rightarrow \setC : T \leadsto \tr_\omega(T) = \omega(\overline\tau(T)).$$
\end{defn}

\vspace{1em}                                                      

By its definition, Dixmier trace gives a non-null value only for sequence of singular values with logarithmic divergence, and a null value on all infinitesimals of order greater than one.

\begin{prop}
If $\lim_{\lambda\rightarrow\infty} \tau_\lambda(T)$ exists, and only in this case, the Dixmier trace is independent of the choice of $\omega$ and \mbox{$\tr_\omega(T) = \lim_{\lambda\rightarrow\infty} \tau_\lambda(T)$}.
\end{prop}
\begin{proof}
If $\lim_{\lambda\rightarrow\infty} \tau_\lambda(T) = L(T)$ then \mbox{$\prt{\tau(T) - L(T) 1}\in C_0([a,\infty))$}, and since every $\omega$ vanishes on $C_0([a,\infty))$ we have \mbox{$\omega\prt{\tau(T) - L(T)1} = 0$}, so \mbox{$\tr_\omega(T) = \omega\prt{\tau(T)} = L(T)$}.

Conversely if $ \tau_\lambda(T)$ has two distinct limit points we can find two states $\omega_1$ and $\omega_1$ whose values on $ \tau(T)$ are different.
\end{proof}

The Dixmier trace can be seen as a kind of generalized limit 
$$\tr_\omega(T) = \lim_{N\rightarrow\,\omega} \frac{\sigma_N(T)}{\ln N} $$
which corresponds to the usual limit when this one exists.

\begin{defn}
Operators such that the Dixmier trace is independent of the choice of $\omega$ are called {\bf\index{operator!measurable} measurable}.
\end{defn}

The Dixmier trace is define for every operator in the space $\L^{1+}(\H)$. We can extend the definition of this space to $\L^{p+}(\H)$ spaces. 

\begin{defn}
The spaces $\L^{p+}(\H)$ are defined by:
$$\L^{p+}(\H) = \set{T\in\K(\H) : \norm{T}_{p+} = \sup_{\lambda>e} \frac{\sigma_\lambda(T)}{\lambda^{\frac{p-1}{p}}} < \infty}\cdot$$
\end{defn}

\begin{prop}
If $T\in\L^{p+}(\H)$ is a positive operator, then \mbox{$T^p\in\L^{1+}(\H)$}.\footnote{The converse is not necessarily true.}
\end{prop}

\begin{proof}
There exists a constant $C>0$ such that 
$$\sigma_\lambda(T) \leq C\,\lambda^{\frac{p-1}{p}} = C\,\frac{p-1}{p} \int_0^\lambda \frac{1}{s^{\frac 1p}} ds.$$
 So we can find a constant $C' \geq C$ such that $\mu_n(T) \leq C'\,\frac{p-1}{p} \frac{1}{(n+1)^{\frac 1p}}$ for each $n\geq 0$, and since $T$ is positive:
 $$\mu_n(T^p) =  \mu_n(T)^p \leq \prt{C'\,\frac{p-1}{p}}^p \frac{1}{(n+1)} $$
 and for some $C'' \geq  \prt{C'\,\frac{p-1}{p}}^p$:
 \begin{equation*}
 \sigma_\lambda(T^p) \leq C''\int_0^\lambda \frac{1}{s} ds = C''\,\ln\lambda.\qedhere
 \end{equation*}
\end{proof}

By this proposition, for every operator $T\in\L^{p+}(\H)$, the Dixmier trace $\tr_\omega(\abs{T}^p)$ is well defined.\\

\vspace{1em}                                                      

\subsection{Spin geometry on Riemannian manifolds}\label{spinsec}

We need to introduce the theory of spin manifolds and Dirac operators. We give here a concise review of the fundamental notions, and we refer the readers who are unfamiliar with spin geometry  to \cite{Ati64,Lawson}.

\begin{defn}
If $V$ is a vector space over $\setR$ and $q$ a symmetric bilinear form on $V$, the {\bf\index{algebra!Clifford} Clifford algebra} $\text{Cl}(V,q)$ is the most general algebra generated by $V$ under the condition
$$uv + vu = 2\,q(u,v) \  \lequi \ u^2 = q(u,u) 1 \quad \forall u,v\in V.$$
The  {\bf\index{algebra!Clifford!complexified} complexified Clifford algebra} is just given by complexification \mbox{$\Cl(V,q) = \text{Cl}(V,q) \otimes \setC\cong  \text{Cl}(V \otimes \setC,q)$} with $q$ extended bilinearly.
\end{defn}

One way to construct a Clifford algebra is to consider the tensor algebra \mbox{$\sum_{r=0}^\infty \bigotimes_r V$} (the algebra of tensors of any rank with tensor product as multiplication) and to take the quotient by the ideal generated by elements of the form \mbox{$v\otimes v - q(v,v) 1$} for all $v \in V$.

\begin{defn}
The {\bf\index{fiber bundle!Clifford} Clifford bundle} $\Cl(T\M)$ over a compact Riemannian manifold $\M$ with metric $g$ is the fiber bundle  whose fibers are the complexified Clifford algebras generated by the tangent bundle over $\M$, i.e.~\mbox{$F_p = \Cl(T_p(\M),g_{|p})$}. Alternatively, the  {\bf Clifford bundle} $\Cl(T^*\M)$ is generated by the cotangent bundle, i.e.~\mbox{$F_p = \Cl(T^*_p(\M),g^{-1}_{|p})$}. Those two Clifford bundles are isomorphic by extending the isomorphism between tangent and cotangent spaces given by $g$, so we can define the {\bf Clifford bundle} $\Cl(\M)$ to be alternatively one of them.
\end{defn}

\begin{defn}
A {\bf\index{module!Clifford} Clifford module} over $\M$ is a finitely generated projective $C(\M)$-module $\Gamma(W)$, where $W$ is a vector bundle over $\M$, together with a representation of the Clifford algebra $\Cl(\M)$ onto the space of endomorphisms of $W$, i.e.~a \mbox{$C(\M)$-linear} homomorphism \mbox{$c : \Gamma\prt{\Cl(\M)} \rightarrow \Gamma(\text{End }W)$}. The action of the Clifford algebra can be represented by a left multiplication \mbox{$c(\alpha)\,\omega = \alpha \cdot \omega$}, so a Clifford module can be considered as a $C(\M)$-module and a left ${\Gamma(\Cl(\M)})$-module.
\end{defn}

The Levi--Civita connection on the tangent bundle, or similarly on the cotangent bundle, extends uniquely as a connection $\nabla^{\Cl}$ on ${\Cl(\M)}$. 

\begin{defn}
For any self-adjoint Clifford module $\Gamma(W)$ (i.e.~Clifford module with a self-adjoint representation $c^\dagger (\alpha) = c(\alpha^*)$ where the adjoint $c^\dagger$ is relative to a Hermitian pairing in \mbox{$\Gamma(W)$}) we define a {\bf\index{connection!Clifford} Clifford connection} $\nabla$ by requiring the Leibniz rule:
$$\nabla(c(v)s) = c(\nabla^{\Cl}v)\,s + c(v) \,\nabla s \qquad \forall v\in\Gamma({\Cl(\M)}),\ \forall s\in\Gamma(W).$$
\end{defn}

We can rewrite the Clifford action \mbox{$c : \Gamma\prt{\Cl(\M)} \rightarrow \Gamma(\text{End }W)$} as an operator \mbox{$\hat c : \Gamma\prt{\Cl(\M)} \otimes \Gamma(W) \rightarrow \Gamma(W)$} by setting \mbox{$\hat c(v\otimes s) = c(v)\,s$}. This will be useful for the following definition:

\begin{defn}
The (generalized) {\bf\index{operator!Dirac} Dirac operator} associated to the connection $\nabla$ and the self-adjoint Clifford action $c$ is defined by\footnote{The $-i$ factor is added here to obtain an essentially self-adjoint operator (see below) instead of an essentially skew-self-adjoint operator.}:
$$ D = -i(\hat c \circ \nabla).$$
\end{defn}

This definition can easily be illustrated with the use of Dirac matrix.

\begin{defn}
If \mbox{$(x^1,\dots,x^n)$} is a local basis of $\M$, the (curved) {\bf\index{Dirac matrix} Dirac matrix} (or curved gamma matrix) are the Hermitian elements \mbox{$\gamma^\mu(x) = c(dx^\mu) \in \Cl(\M)$} which respect the relations:
$$\gamma^\mu \gamma^\nu + \gamma^\nu \gamma^\mu = 2 g^{\mu\nu}.$$
\end{defn}

The Dirac operator associated to the covariant derivative \mbox{$\nabla_\mu = \partial_\mu + \omega_\mu$} can be written locally as
$$D = -i c(dx^\mu)\nabla_\mu = -i\gamma^\mu\prt{\partial_\mu  + \omega_\mu}$$
and in case of a flat connexion, the Dirac operator is simply 
$$D = -i\gamma^\mu \partial_\mu = -i /\!\!\!\partial.$$

We can be interested in the definition of the Dirac operator whenever the compact Riemannian manifold $\M$ carries a spin structure. We will denote by $E$ the tangent bundle over $\M$ with metric $g$.

\begin{defn}
A {\bf\index{spin structure} spin structure}\footnote{See  \cite{Lawson} for more details or alternative definitions about spin structures, and for existence conditions of spin structures on Riemannian manifolds.} on $E$ is a principal $\text{Spin(n)}$-bundle $P_{\text{Spin}}(E)$ with a double covering \mbox{$\xi : P_{\text{Spin}}(E) \rightarrow P_{\text{SO}}(E)$} such that
$$ \xi(pg) = \xi(p)\,\iota(g) \qquad \forall p\in P_{\text{Spin}}(E),\ \forall g\in\text{Spin(n)}$$
where
\begin{itemize}
\item $\text{Spin(n)}$ is the Lie group which is a double covering of the group $SO(n)$ composed by the invertible elements of the real Clifford algebra $\text{Cl}(\setR^n,g)$, with \mbox{$\iota : \text{Spin(n)} \rightarrow SO(n)$} being the double covering map (i.e with kernel $\setZ_2$)
\item $ P_{\text{SO}}(E)$ is the principal $SO(n)$-bundle of orthonormal frames of $E$, i.e.~with each fiber $F_x$ being the set of all orthonormal ordered basis at the point $x\in\M$
\end{itemize}
\end{defn}

\vspace{0.5em}                                                      

A manifold $\M$ which carries a spin structure on its tangent bundle $E$ is called a {\bf\index{manifold!spin} spin manifold}.

\vspace{0.5em}                                                      

\begin{defn}
A (complex) {\bf\index{fiber bundle!spinor}\index{module!spinor} spinor bundle} $S$ of $E$ is a bundle of the form \mbox{$S = P_{\text{Spin}}(E) \times_\mu V$} where $V$ is a complex left $\Cl(\setR^n,g)$-module (typically $\Cl(\setR^n,g)$ itself with left multiplication) and $\mu$ is the complex representation of \mbox{$\text{Spin}(n)$} (seen as a subspace of the algebra $\Cl(\setR^n,g)$) given by left multiplication.
\end{defn}

\vspace{0.5em}                                                      

One can check that the sections of a spinor bundle $S$ form a module over the sections of the Clifford bundle $\Cl(\M)$, so \mbox{$\Gamma(S)$} can be seen as a self-adjoint Clifford module with a Clifford action $c$.\\

Once more we can take the Levi--Civita connection considered as a connection $1$-form on $P_{\text{SO}}(E)$, lift it to $ P_{\text{Spin}}(E)$ by the use of $\xi$, and then define a unique Hermitian connection on $S$, called the {\bf\index{connection!spin} spin connection} $\nabla^S$, with the Leibniz rule:
$$\nabla^S(c(v)s) = c(\nabla^{\Cl}v)\,s + c(v) \,\nabla^S s \qquad \forall v\in\Gamma({\Cl(\M)}),\ \forall s\in\Gamma(S).$$

We define by this way the Dirac operator on a spinor bundle to be the Dirac operator associated to the spin connection:
$$ D = -i(\hat c \circ \nabla^S).$$
This operator is sometimes called the {\bf\index{operator!Atiyah--Singer} Atiyah--Singer} operator. Its square is related to the Laplacian by the Lichnerowicz formula:

\begin{thm}[Lichnerowicz \cite{Lich}]
\begin{equation}\label{Lich}
D^2 = \Delta^S + \frac{1}{4}R
\end{equation}
where $\Delta^S = -g^{\mu\nu}\prt{\nabla^S_\mu\nabla^S_\nu - \Gamma^\lambda_{\mu\nu} \nabla^S_\lambda}$ is the Laplacian operator associated to the spin connection and $R$ is the scalar curvature.
\end{thm}

Now we can construct a Hilbert space where the Dirac operator acts as an unbounded essentially self-adjoint operator.

\begin{defn}
{\bf $\H = L^2(\M,S)$} is the Hilbert space of square integrable sections of the spinor bundle $S$ over $\M$, i.e.~the completion of the prehilbert space given by the spinor module \mbox{$\Gamma(S)$} under the positive definite Hermitian form
$$(\psi,\phi) = \int_\M \psi^*\phi\, d\mu_g$$
with $\psi^*\phi$ the usual inner product in $\setC^{m}$ (with $m = 2^{\left[\frac n2\right]}$ being the rank of the spinor bundle) and $d\mu_g = \sqrt{\det g} \;d^nx$ the Riemannian density on $\M$.
\end{defn}

\begin{thm}
The Dirac operator on a spinor bundle $S$ is an unbounded essentially\footnote{Essentially self-adjoint means that the unique extension of $D$ on $\H = L^2(\M,S)$ is a self-adjoint operator.} self-adjoint operator on the space $\H = L^2(\M,S)$.
\end{thm}

The proof of this theorem can be found in \cite{Var,Lawson}.\\

The next result is of great interest for Riemannian noncommutative geometry. It shows that the commutator of a Dirac operator acts as a differential on the space of functions on $\M$.

\begin{thm}\label{commudirac}
If $D$ is the Dirac operator on the spinor  module $\Gamma(S)$ and if we consider $C^\infty(\M)$ as an algebra of operators acting on $\Gamma(S)$ by scalar multiplication, then 
$$\forall a\in C^\infty(\M),\qquad [D,a] = -i\,c(da).$$
\end{thm}

\begin{proof}
For any section $s\in\Gamma(S)$, we have
\begin{eqnarray*}
i\,[D,a] \,s &=& i\,D(as) - i\,a\,D(s)\\
&=& \hat c\prt{ \nabla^S(as)} - a\,\hat c\prt{\nabla^S s}\\
&=& \hat c\prt{ \nabla^S(as) - a\,\nabla^S s}\qquad\text{since c is $C(\M)$-linear}\\
&=& \hat c\prt{ \nabla^S(a)\otimes s +a\,\nabla^S s - a\,\nabla^S s}\\
&=& c\prt{ \nabla^S(a) } \,s\qquad\text{}
\end{eqnarray*}
and since the covariant derivative of a function $a$ corresponds to its directional derivative, we have \mbox{$\nabla^S_x\, a = x(a) = da(x) $}, so
\begin{equation*}
[D,a]\,s = -i\,c(da)\,s.\qedhere
\end{equation*}
\end{proof}

This result is quite surprising. It shows that the operator $[D,a]$ is a bounded operator acting on $\H$ which is no more than a representation of the usual differential operator $da$. Its norm is given by \mbox{$\norm{[D,a]} = \sup_{x\in\M}\norm{da(x)}$} which is finite since $a$ is smooth and $\M$ is compact. The fact the $[D,a]$ acts as a multiplicative operator can also be checked from its local expression:
$$i[D,a]s = \gamma^\mu \partial_\mu (as) - a \gamma^\mu \partial_\mu s = \prt{\gamma^\mu \partial_\mu a}s .$$

Further results can be obtained from the Dirac operator on compact Riemannian manifolds, and more precisely from a kind of inverse $\abs{D}^{-1}$ of this operator. However, in order to define this inverse, we need to introduce pseudodifferential operators. For a complete introduction to the theory of pseudodifferential operators, the reader can refer to \cite{Taylor,Treves}.\\

The notion of pseudodifferential operator comes from the following observation: every differential operator of order $d$ on the Riemannian bundle $E$ is a linear map \mbox{$p :\Gamma(E) \rightarrow \Gamma(E)$} which can be locally written as:
$$P(x) = \sum_{\abs{\alpha} \leq d} a_\alpha(x) \,\partial^\alpha$$
with $\alpha = \set{\alpha_1, \dots, \alpha_n}$  a multi-index of cardinality \mbox{$\abs{\alpha} = \sum_{k=1}^n \alpha_k$} ($n$ is here the dimension of the manifold $\M$), $a_\alpha \in C^\infty(\M)$ for each $\alpha$ and \mbox{$\partial^\alpha = \partial_1^{\alpha_1} \circ\cdots\circ \partial_n^{\alpha_n}$}. By use of the Fourier transform, we can write:
\begin{eqnarray}
P(x) f(x) &=& \frac{1}{(2\pi)^n} \int e^{ix\cdot \xi} \,p(x,\xi)\,\hat f(\xi)\,d\xi \nonumber\\
&=& \frac{1}{(2\pi)^n} \int \int e^{i(x-y)\cdot \xi} \,p(x,\xi)\,f(y)\,dy\,d\xi\label{pseudo}
\end{eqnarray}
for  $f\in\Gamma(E)$ and with $\xi\in\setR^n$, and where
\begin{equation}\label{princip}
p(x,\xi) =  \sum_{\abs{\alpha} \leq d} a_\alpha(x) \xi^\alpha.
\end{equation}

The function $p(x,\xi)$ is called the {\bf\index{symbol} symbol} of $P$. The idea of the construction of pseudodifferential operators is to enlarge the class of symbols. For this goal, we say that a symbol is of order $d$, and we write $p\in\text{Sym}^d$, if for every multi-index $\alpha,\beta$ there is a constant $C_{\alpha\beta}$ such that
$$ \abs{\partial_x^\beta \,\partial_\xi^\alpha \,p(x,\xi)} \leq C_{\alpha\beta} \prt{1+\abs{\xi}}^{d-\abs{\alpha}}.$$

\begin{defn}
A {\bf\index{operator!pseudodifferential} pseudodifferential operator} of order $d$ is an operator $P$ in the form defined by (\ref{pseudo}) with a symbol $p\in\text{Sym}^d$.
\end{defn}

\vspace{0.5em}                                                      

There is no problem to define pseudodifferential operators of negative order. A pseudodifferential operator of order $-\infty$, i.e.~with a symbol \mbox{$p\in\text{Sym}^{-\infty} = \cap_d\; \text{Sym}^d$}, is called a {\bf\index{operator!smoothing} smoothing operator}.

\vspace{0.5em}                                                      

\begin{defn}
The {\bf\index{symbol!principal} principal symbol} of a pseudodifferential operator of order $d$ is the highest order part of the formula (\ref{princip}):
$$\sigma^P(\xi) = \sum_{\abs{\alpha} = d} a_\alpha\, \xi^\alpha.$$
\end{defn}

From the principal symbol, we can define a kind of trace in the space of pseudodifferential operators of order $-n$, introduced by Wodzicki \cite{Wodzicki}.

\vspace{0.5em}                                                      

\begin{defn}
Let $P$ be a pseudodifferential operator of order $-n$. The {\bf\index{Wodzicki residue} Wodzicki residue}\footnote{The initial formulation of the Wodzicki residue was without the coefficient.} is defined by
\begin{equation}\label{Wodzicki}
\text{Res}_W\,P = \frac{1}{n(2n)^n} \int_\M \int_{S^{n-1}} \tr_E\,\sigma^P(\xi)\  d\mu_g\,d\xi
\end{equation}
where \mbox{$S^{n-1} =\set{\xi : \norm{\xi}=1}$} is the standard unit sphere and $\tr_E$ is the pointwise matrix trace on $\text{End(E)}$.
\end{defn}

\begin{defn}
A pseudodifferential operator is {\bf\index{operator!elliptic} elliptic} if its principal symbol $\sigma^P(\xi)$ is invertible for each non-zero $\xi$.
\end{defn}

The invertibility of the principal symbol of elliptic operators leads to the useful property (see \cite{Treves} for the proof):

\begin{prop}
If $P$ is an elliptic pseudodifferential operator of order $d$, then there exists an inverse $P^{-1}$ modulo smoothing operator, i.e.~another elliptic pseudodifferential operator $P^{-1}$ of order $-d$ such that \mbox{$\prt{PP^{-1} - \mathbb I}$} and \mbox{$\prt{P^{-1}P - \mathbb I}$} are both smoothing operators. Such inverse is called a {\bf\index{parametrix} parametrix}.
\end{prop}

Now we can return to our analysis of the Dirac operator, within the light of the theory of pseudodifferential operators. We have the following results:

\begin{prop}
If $D$ is the Dirac operator  on the Hilbert space \mbox{$\H = L^2(\M,S)$} of square integrable spinor sections over the \mbox{$n$-dimensional} compact Riemannian manifold $\M$, then:
\begin{enumerate}
\item $D$ is a pseudodifferential operator of order $1$
\item The principal symbol of $D$ is given by the Clifford multiplication \mbox{$\sigma^D(\xi) = c(\xi)$}, which implies  \mbox{$\sigma^D(\xi)^2 = c^2(\xi) = g(\xi,\xi)$}
\item $D$ is an elliptic pseudodifferential operator
\item $D$ possesses an inverse $D^{-1}$ modulo smoothing operation
\item The operator $\abs{D}^{-1}$ is compact and lives in $\in \L^{n+}$, so the Dixmier trace \mbox{$\tr_\omega(\abs{D}^{-n})$} is well defined
\end{enumerate}
\end{prop}

The proofs of these properties can be found in \cite{Var}. 1 is obvious by the definition of the connection. 2 comes from the local expression of this connection. 3 is a consequence of 2, and 4 a consequence of 3. 5 is obtained by use of the Lichnerowicz formula (\ref{Lich}) and by studying the properties of the Laplacian operator. It is related to the fact that every self-adjoint elliptic differential operator over a compact Riemannian manifold has a real discrete spectrum which tends rapidly to infinity \cite{Lawson}.\\

\vspace{1em}                                                      

\subsection{Noncommutative integral}

\vspace{1em}                                                      

We now return to our goal of defining a noncommutative integral. The main key is the theorem introduced by A. Connes in \cite{C88} establishing a correspondence between the Dixmier trace and the Wodzicki residue (\ref{Wodzicki}).

\begin{thm}[Connes trace theorem]
If $P$ be a pseudodifferential operator of order $-n$ acting on the Riemannian bundle $E$, then the Dixmier trace $\tr_\omega(P)$ is independent of $\omega$ and corresponds to the residue
$$\tr_\omega(P) = \text{Res}_W\,P.$$
\end{thm}

The proof is quite long and can be found in \cite{C88} or \cite{Var}.\\

We know that the Dirac operator leads to a pseudodifferential operator $\abs{D}^{-n}$ of order $-n$. Since any $f\in\A = C^\infty(\M)$ acts as a bounded multiplicative operator, we can create the  pseudodifferential operator $f\abs{D}^{-n}$ which is also of order $-n$ and in $ \L^{1+}$. We have then the following result:

\begin{thm}
\begin{equation}\label{comint}
\int_\M f(x) \,d\mu_g =  c_n \tr_\omega(f\abs{D}^{-n})\qquad \forall f\in\A
\end{equation}
with $c_n = 2^{n-[\frac n2]-1}\pi^{\frac n2}n\,\Gamma(\frac n2)$.\footnote{$\Gamma(n)$ is the  Gamma function.}
\end{thm}
\begin{proof}
The principal symbol of the operator $f\abs{D}^{-n}$ is \mbox{$\sigma(\xi)=f(x)\norm{\xi}^{-n}$}. Because the Wodzicki residue works on the unit sphere  \mbox{$S^{n-1} =\set{\xi : \norm{\xi}=1}$}, this symbol reduces to a matrix  \mbox{$\sigma(\xi)=f(x)\, \mathbb I_{2^{\left[\frac n2\right]}}$} where $2^{\left[\frac n2\right]}$ is the dimension of the fibers, so we get:
\begin{eqnarray*}
\tr_\omega(f\abs{D}^{-n}) &=&  \frac{1}{n(2n)^n} \int_\M \int_{S^{n-1}} \tr_E\,  f(x) \,\mathbb I_{2^{\left[\frac n2\right]}} \  d\mu_g\,d\xi\\
&=&   \frac{1}{n(2n)^n}   \int_{S^{n-1}} \tr_E \,\mathbb I_{2^{\left[\frac n2\right]}} \, d\xi\ \ \int_\M f(x)\  d\mu_g\,\\
&=&   \frac{2^{\left[\frac n2\right]}}{n(2n)^n}   \int_{S^{n-1}} \, d\xi\ \ \int_\M f(x)\  d\mu_g\\
&=& \frac{1}{c_n}\int_\M f(x)\  d\mu_g
\end{eqnarray*}
with $c_n$ determined from the area of the unit sphere  \mbox{$\int_{S^{n-1}} \, d\xi = \frac{2\pi^{\frac n2}}{\Gamma\prt{\frac n2}}$}.\\
\end{proof}

\begin{prop}
The action functional given by the Dixmier trace of $\abs{D}^{2-n}$ is proportional to the Einstein--Hilbert action
$$ S(D) = \tr_\omega(\abs{D}^{2-n}) \sim \int_\M R\ d\mu_g.$$
\end{prop}

This last proposition was conjectured by A. Connes in \cite{C92} and verified mainly by "brute force" in \cite{Kalau, Kastler}.\\

The formula (\ref{comint}) is our main interest. It shows that we can construct a complete algebraic translation of the concept of geometrical integration. In this formula, the operator ${D}^{-1}$ plays the role of a kind of line element $ds$, so the metric is dictated from the Dirac operator. Of course this is just a tautology since the Dirac operator is directly defined from the metric.\\

However this opens the door to a possible generalization to noncommutative spaces. Indeed, if we take an arbitrary unital algebra $\A$ acting on a Hilbert space $\H$ as bounded multiplicative operators, we just have in order to define integrals to choose an operator $D$ acting as a self-adjoint operator at least on a dense subset of $\H$ and which corresponds to an elliptic pseudodifferential operator of order one with compact inverse $D^{-1}$.\\

Of course requiring the existence of an inverse is too strong, since for the usual Dirac operator we only have inverse modulo smoothing operator. However, we can use the fact that for every self-adjoint operator, eigenvalues are real and that the resolvent $R(\lambda) = (D-\lambda 1)^{-1}$ for all $\lambda \notin \sigma(D)$ (so in particular $\forall\lambda \notin \setR$) are bounded operators on $\H$. So we can required the resolvent to be compact $\forall\lambda \notin \sigma(D)$. One can check that requiring this condition for a particular $\lambda$ is sufficient. 

\begin{defn}
An operator is with {\bf\index{resolvent!compact} compact resolvent} if its resolvent $R(\lambda)$ is compact for $\lambda \notin \sigma(D)$.
\end{defn}

Such operator is automatically unbounded if $\H$ is infinite dimensional. As a consequence of spectral theory, for any self-adjoint operator $D$ with compact resolvent there is an orthonormal basis $\set{\phi_k}$ for $\H$ consisting of simultaneous eigenvectors for all the resolvents $R(\lambda)$. These vectors are also eigenvectors of $D$ with eigenvalues \mbox{$D\phi_k = \mu_k\phi_k$} forming a discrete sequence such that \mbox{$\lim_{k\rightarrow\infty}\abs{\mu_k} = \infty$}. By this fact, the kernel of $D$ must be finite dimensional.

\begin{defn}
The inverse {\bf $D^{-1}$} of a self-adjoint operator $D$ with compact resolvent is defined as the inverse operator on the orthogonal complement of the finite dimensional kernel of $D$ and $0$ otherwise.
\end{defn}

So now we can choose an arbitrary self-adjoint operator $D$ with compact resolvent and use its inverse $D^{-1}$. The way to characterize the first order condition will be shown latter.

\begin{defn}
$D$ is {\bf\index{operator!finitely summable} finitely summable} (or $n^+$-summable) if there exists a positive integer $n$ such that $D^{-1} \in \L^{n+}$.
\end{defn}

\begin{defn}
If $D$ is $n^+$-summable and measurable, the {\bf\index{noncommutative integral} noncommutative integral} of $a\in\A$ is defined as
$$\int\!\!\!\!\!\!-\ a \abs{D}^{-n} = \tr_\omega(a \abs{D}^{-n}) .$$
\end{defn}

This shows that the choice of the operator $D$ completely determines the integral aspect of noncommutative algebras by use of the spectral information given by the operator $D$.\\

\subsection{Noncommutative differential calculus}\label{diffcalc}

We want  to characterize differential operators of any order for noncommutative spaces. The idea is to construct a representation of differential forms which can be extended to noncommutative spaces, but we need to introduce the more general background of differential graded algebras.

\begin{defn}
A {\bf\index{algebra!graded} graded algebra} $\bA$ is a direct sum of associative algebras $\bA^r$:
$$\bA = \bigoplus_{r=0}^\infty \bA^r$$
such that the multiplication operation satisfies
$$a\in \bA^r \text{ and } b\in \bA^s \implies ab\in \bA^{r+s}.$$
It is a {\bf\index{algebra!differential} differential graded algebra} if it is equipped with a linear map \mbox{$d : \bA \rightarrow \bA$} of degree $1$ (which means \mbox{$d : \bA^r \rightarrow \bA^{r+1}$})  with the following properties:
\begin{itemize}
\item $d^2 = 0$
\item $d(ab) = (da) \,b + (-1)^{r} a\, db\quad $ if $a\in\bA^r\ $ (graded Leibniz rule)
\end{itemize}
A {\bf\index{differential calculus} differential calculus} on the algebra $\A$ is a differential graded algebra \mbox{$(\Omega\A,d)$} with $\Omega^0\A = \A$.
\end{defn}

The space of differential forms \mbox{$\bigoplus_{k=0}^\infty \Omega^k(\M)$} with the exterior derivative is clearly  a particular case of differential calculus on $C^\infty(\M)$, and is called the de Rham differential calculus. We want to introduce the most general differential calculus on an algebra $\A$, called the universal differential algebra.

\begin{defn}
The {\bf\index{algebra!differential!universal} universal differential algebra} on  $\A$ is the differential graded algebra \mbox{$(\Omega\A,d)$} where
\begin{itemize}
\item $\Omega^0\A = \A$
\item $\Omega^1\A$ is the $\A$-bimodule subspace of $\A \otimes \A$ generated by \mbox{$da = 1 \otimes a - a \otimes 1$}, which implies the condition \mbox{$d(ab)= a\,db + b\,da$}, and with the identification \mbox{$a_o \otimes \bar a_1 = a_0\,da_1$} with $\bar a_1$ taken modulo $\setC$ (i.e.~$\bar a$ is the image of $a$ by the quotient map \mbox{$\A \rightarrow \A/\setC$})
\item $\Omega^n\A =\underbrace{\Omega^1\A \otimes \cdots \otimes \Omega^1\A}_{n \text{ times}}$ is the space generated by elements of the form \mbox{$a_0 \otimes \bar a_1 \otimes \cdots \otimes \bar a_n = a_0 \,da_1 \dots da_n $}, with the product rule
$$\prt{a_0\,da_1}\prt{b_0\,db_1} = \prt{a_0\,\prt{da_1}\,b_0}\,db_1 = a_0\,d\!\prt{a_1b_0}\,db_1 - a_0\,a_1\,db_0\,db_1$$
$$\prt{a_0\,da_1\dots da_n}\prt{b_0\,db_1\dots db_n} =   \prt{a_0\,\prt{da_1\dots da_n}\,b_0}\,db_1\dots db_n$$
\item $d$ is extended to every $\Omega^n\A$ in the unique following manner: $d(a_0 \otimes \bar a_1 \otimes \cdots \otimes \bar a_n) =1 \otimes \bar a_0 \otimes \cdots  \otimes \bar a_{n} $, which implies \mbox{$d(a_0\,da_1\dots da_n) = da_0\,da_1\dots da_n$} and $d^2 = 0$ since $\bar 1 = 0$
\end{itemize}
\end{defn}

The name {\it universal} is justified by the following property: if \mbox{$(\tilde\Omega\A,\tilde d)$} is another differential graded algebra with $\tilde\Omega^0\A = \A$ and $\tilde\Omega^n\A$ generated by the $\tilde da$, then \mbox{$(\tilde\Omega\A,\tilde d)$} is isomorphic to a quotient of \mbox{$(\Omega\A,d)$} \cite{Mass96}. The de Rham differential algebra is by this way a particular case of quotient differential graded algebra.\\

In the special case where $\A = C^\infty(\M)$ for a compact Riemannian manifold $\M$ with spin structure, we have from the Theorem \ref{commudirac} that the Dirac operator acts as a derivative by its commutator $[D,\,\cdot\,]$, and that $\forall a\in \A$, $[D,a] = -ic(da)$ is a bounded operator acting multiplicatively on the Hilbert space \mbox{$\H = L^2(\M,S)$}. So we can construct a representation of the universal algebra $\Omega\mathcal{A}$ into the algebra of bounded operators on $\H$ by:
$$\pi : \Omega\mathcal{A} \rightarrow \mathcal B(\mathcal{H}) : \quad\pi(a_0 d a_1 \cdot\cdot\cdot d a_p) = a_0 [D,a_1] \cdot\cdot\cdot[D,a_p]\qquad a_0,...,a_p\in\mathcal{A}.$$

Once more we want to extend this construction to any unital algebra $\A\in\B(\H)$ with a self-adjoint operator $D$ such that $[D,\,\cdot\,]$ acts as a first order derivative. To obtain this first order condition we will simply require that $[D,a]$ should be a bounded operator for each $a\in\A$.\\

The first idea is to define a noncommutative differential algebra by taking all elements of the form \mbox{$a_0 [D,a_1] \cdot\cdot\cdot[D,a_p]$}, so the full image \mbox{$\pi(\Omega\mathcal{A})$}. Nevertheless, this definition leads to an unpleasant problem: the fact that $\pi(\omega)=0$ for $\omega\in\Omega\mathcal{A}$ does not imply  $\pi(d\omega)=0$. The forms $\omega\in\Omega\mathcal{A}$ such that $\pi(\omega)=0$ and $\pi(d\omega) \neq 0$ are called {\bf\index{junk form} junk forms}.\\

If $\A = C^\infty(\M)$, an example of such junk form is given by
$$\omega = fdf - (df)f$$
 with $f\in\A$. Indeed, we have that
$$ \pi(\omega) = f[D,f] - \prt{[D,f]}f = -i\gamma^\mu\prt{f\,\partial_\mu f - \prt{\partial_\mu f} f} = 0,$$
but by the graded Leibniz rule,
$$d\omega = df df + f d^2 f - \prt{d^2 f} f + df df =2\,df\,df$$\vspace{-2em}                                                      
\begin{eqnarray*}
\implies \pi(d\omega) &=& 2[D,f][D,f] = -i\,2\,\gamma^\mu\gamma^\nu \,\partial_\mu f \partial_\nu f\\
&=&  -i \gamma^\mu\gamma^\nu \,\partial_\mu f \partial_\nu f -i \gamma^\nu\gamma^\mu \,\partial_\nu f \partial_\mu f\\
&=& -i\prt{ \gamma^\mu\gamma^\nu + \gamma^\nu\gamma^\mu} \,\partial_\mu f \partial_\nu f\\
&=& -i\,2\,g^{\mu\nu} \,\partial_\mu f \partial_\nu f \neq 0.
\end{eqnarray*}
So we have to remove such forms by performing a suitable quotient.

\begin{prop}
Let us define
$$J_0^p = \set{\omega \in\Omega^p\A\ :\ \pi(\omega)=0}$$
and $J_0 = \bigoplus_{p=0}^\infty J_0^p$, then the subspace \mbox{$J = J_0 + dJ_0$} is a graded two-sided differential ideal of $\Omega\mathcal{A}$.
\end{prop}

\begin{proof}
$J^p = J \cap \Omega^p{\A}$ is generated by elements of the form \mbox{$\omega = \omega_1 + d\omega_2$} with $\omega_1 \in J_0^p$ and $\omega_2\in J_0^{p-1}$. Then for every $\eta \in\Omega^q\A$ we have \mbox{$\pi(\omega_i\eta) = \pi(\omega_i) \pi(\eta) = 0$}, $i\in\set{1,2}$, by the homomorphism property of the representation, and then:
\begin{eqnarray*}
\omega\eta &=& \omega_1 \eta + (d\omega_2)\eta\\
&=& \underbrace{\omega_1 \eta - (-1)^{p-1} \omega_2 d\eta}_{\in \,J_0^{p+q}} + d(\underbrace{\omega_2 \eta}_{\in \,J_0^{p+q-1}}) \\
\end{eqnarray*}
which implies $\omega\eta \in J^{p+q} \subset J$, and similar for $\eta\omega$. So $J$ is a graded two-sided ideal, with the differential property coming from \mbox{$dJ = dJ_0 + d^2 J_0 = dJ_0$}.
\end{proof}

\begin{defn}
The {\bf\index{algebra!differential!noncommutative} noncommutative differential algebra} is defined as the quotient space:
\vspace{-0.5em}                                                      
$$\Omega_D\A = \Omega\A / J.$$
\end{defn}

If we explicit the structure of this algebra gradually \cite{C94}, we have:
\begin{itemize}
\item $\Omega^0_D\A \cong \A$
\item $\Omega^1_D\A \cong \pi\prt{ \Omega^1\A}$
\item $\Omega^p_D\A \cong \pi\prt{\Omega^p\A}/\pi\prt{dJ_0^{p-1}}$
which is composed of sums of operators of the form
$$\omega_p = a_0 [D,a_1] \cdot\cdot\cdot[D,a_p]$$
modulo the subbimodule generated by operators
$$\tilde\omega_p = [D,b_0] [D,b_1] \cdot\cdot\cdot[D,b_{p-1}] \ \text{ such that }\  b_0 [D,b_1] \cdot\cdot\cdot[D,b_{p-1}]= 0.$$
\end{itemize}

Of course the derivative is extended to the quotient space and is explicitly given on basic elements by
$$d\prt{\,a_0 [D,a_1] \cdot\cdot\cdot[D,a_p]\,} = [D,a_0][D,a_1] \cdot\cdot\cdot[D,a_p].$$

We can wonder to what the differential algebra $\Omega_D\A$ corresponds when $D$ is the Dirac operator  on the Hilbert space $\H = L^2(\M,S)$ of square integrable spinor sections over $\M$? The result is that this algebra is simply isomorphic to the usual de Rham differential algebra \cite{Landi97}.\\

\subsection{Noncommutative Riemannian distance}\label{riemdist}

Having set a way to construct integrals and differential forms by using a Dirac-like operator, we are showing now the fact that these data are sufficient to recover  the notion of Riemannian distance. The establishment of the distance formula is very important for us since it will be one of our main concerns in the Chapter \ref{chaplor} when we will discuss the Lorentzian generalization of the theory, so we will take the time here to study the details of the proof. The Riemannian distance was introduced by A. Connes in \cite{C94,C96}.\\

The proof of this Riemannian distance is too often reduced in the literature to approximated considerations. In particular, the fact that the usual Riemannian distance is not a fully differentiable function is usually swept under the carpet. So we will try here to establish the result as carefully as possible.\\

First, we must review what exactly a Riemannian distance is. Let us consider a compact Riemannian manifold $(\M,g)$ and two points $p$ and $q$ on it. For reason of simplicity, we will suppose the manifold to be connected (otherwise we just have to consider the distance between connected points). We know that there exists at least one piecewise smooth curve $\gamma : [0,1] \rightarrow \M$ with $\gamma(0) = p$ and $\gamma(1) = q$.

\begin{defn}
The {\bf\index{curve!length} length} of a piecewise smooth curve $\gamma$ is given by
$$l(\gamma) = \int_0^1 \abs{\dot \gamma(t)}\,dt  = \int_0^1 \sqrt{g_{\gamma(t)}\prt{\dot \gamma(t),\dot \gamma(t)}}\,dt.$$
\end{defn}

\begin{defn}
The {\bf\index{distance!Riemannian} Riemannian distance} between $p$ and $q$ is the quantity
$$d(p,q) = \inf \set{l(\gamma) : \gamma \text{ piecewise smooth curve with } \gamma(0) = p,\ \gamma(1) = q}.$$
\end{defn}

The Riemannian distance has the usual properties of a distance function:
\begin{itemize}
\item $d(p,q) \geq 0$ for all $p,q\in\M$
\item $d(p,q) = 0$ if and only if $p = q$
\item $d(p,q) = d(q,p)$ for all $p,q\in\M$
\item $d(p,r) \leq d(p,q) + d(q,r)$  for all $p,q,r\in\M$\quad(triangle inequality)
\end{itemize}

We will construct the noncommutative counterpart in three steps: first one we will show that there exists a formulation of the Riemannian distance independent of any path consideration, second one we will show that this formulation can be expressed with the only use of the Dirac operator $D$, and last one we will look to how this can be generalized to noncommutative spaces.

\subsubsection*{Step 1: Path independent formulation}

If we set $\A=C^\infty(\M)$, then for each function $f \in \A$ and each piecewise smooth curve $\gamma$ from $p$ to $q$, we have:
\begin{equation}\label{rdist1}
f(q) - f(p) = f(\gamma(1)) - f(\gamma(0)) = \int_0^1 \dv{}{t} f(\gamma(t)) \, dt = \int_0^1 df_{\gamma(t)}(\dot\gamma(t)) \, dt
\end{equation}
by using the second fundamental theorem of calculus.\\

If we denote the gradient operation by $\nabla f$, then the gradient is related to the differential by \mbox{$g(\nabla f, v) = df(v)$} for every vector field $v$.\footnote{Since $f$ is a complex-valued function, the metric $g$ is here the unique \mbox{$C^\infty(\M)$-sesquilinear} extension of the real-valued metric, i.e.~\mbox{$g(\alpha u,\beta v) = \bar \alpha \beta g(u,v)$} for $u,v\in T\M$ real-valued vectors fields and $\alpha,\beta\in C^\infty(\M)$, in order to guarantee its positive definiteness. This leads to the fact that the isomorphism between $\nabla f$ and $df$ -- called musical isomorphism -- is actually antilinear.} Since $\dot\gamma(t)$ is a tangent vector for each $t\in[0,1]$ except at most on a discrete set, the integral (\ref{rdist1}) is equivalent to:
$$ 
\int_0^1 df_{\gamma(t)}(\dot\gamma(t)) \, dt = \int_0^1  g_{\gamma(t)}( \nabla f_{\gamma(t)},\dot\gamma(t)) \, dt.
$$

Since the Cauchy--Schwartz inequality is valid for the metric $g$, we obtain:
\begin{eqnarray}
\abs{ f(q) - f(p) } &\leq& \int_0^1  \abs{g_{\gamma(t)}( \nabla f_{\gamma(t)},\dot\gamma(t))}\nonumber\\
&\leq& \int_0^1  \abs{ \nabla f_{\gamma(t)}} \abs{\dot\gamma(t)} \, dt \label{rdist5}\\
&\leq& \norm{  \nabla f}_\infty \  \int_0^1 \abs{\dot\gamma(t)} \, dt\nonumber  \\
&=& \norm{  \nabla f}_\infty\  l(\gamma)\nonumber
\end{eqnarray}
where \mbox{$ \norm{  \nabla f}_\infty = \sup_{x\in\M} \abs{ \nabla f_x} $}.\\

Since this result is valid for any piecewise smooth curve from $p$ to $q$, it is valid for the infimum, so \mbox{$\abs{ f(q) - f(p) } \leq \norm{  \nabla f}_\infty d(p,q)$}. If we restrict our set of functions to those which respect $\norm{  \nabla f}_\infty \leq 1$, we have 
\begin{equation}\label{rdist8}
d(p,q) \ \geq\  \sup\set{  \abs{f(q)-f(p)} \ :\  f \in \mathcal A,\   \norm{  \nabla f}_\infty \leq 1 }.
\end{equation}

If we want an equality, we have to find a function $f \in \mathcal A$ such that $ \norm{  \nabla f}_\infty \leq 1$ and $\abs{f(q)-f(p)} = d(p,q)$, and because the first condition is invariant by addition of a constant function, we can fix $f(p)=0$ and search for a function satisfying \mbox{$f(q)=d(p,q)$}.\\

The trivial solution is the distance function itself  as a function of its second argument \mbox{$f(\,\cdot\,) = d_p(\,\cdot\,) = d(p,\,\cdot\,)$}. However, this function is not in $\A=C^\infty(\M)$, since its derivative is not correctly defined on $p$ and on a set called the cut locus\footnote{Roughly speaking, the cut locus is the set of points which can be reached by more than one geodesic starting from the point $p$. We will take more time to discuss on the cut locus in the Chapter \ref{chaplor} about the Lorentzian case. In the Riemannian case, the study of the cut locus is not mandatory since the Lipschitz continuity of the distance function implies that the cut locus has measure zero. Informations on the cut locus for Riemannian manifolds can be found e.g.~in \cite{Cha}.}. Nevertheless, this function is still continuous on the whole manifold, and is Lipschitz continuous with best Lipschitz constant $1$ since
$$\abs{f(x)-f(y)} = \abs{d(p,x)-d(p,y)} \leq 1\  d(x,y)\quad\forall x,y\in\M$$
by the reverse triangle inequality. More precisely, this function belongs to the algebra of bounded Lipschitz continuous functions, since the manifold is compact, and is automatically almost everywhere differentiable (a.e.~differentiable) by Rademacher's theorem. The second fundamental theorem of calculus is not valid in general for a.e.~differentiable functions, but is still valid for the class of Lipschitz continuous functions, so the construction in (\ref{rdist1}) holds.\footnote{Properly speaking, the construction (\ref{rdist1}) holds for curves for which the function $f$ is differentiable except for a discrete number of points on the curve. This is not a problem since every potentially problematic curve could be approximated by other curves respecting (\ref{rdist1}) and the result is still valid for the infimum among all curves.} Moreover, since $1$ is the best Lipschitz constant for $d_p$, we have \mbox{$\norm{\nabla d_p} = 1$} whenever this gradient exists, which leads us to use the notion of essential supremum.

\begin{defn}
If $\mu$ is a measure on a space $X$ and $h:X\rightarrow\setR$ a function with real values, the {\bf\index{essential supremum} essential supremum} of $f$ on $X$ is defined by
$$\text{\rm ess} \sup f = \inf\set{a\in\setR : \mu(\set{x : f(x) > a}) = 0}.$$
\end{defn}

The distance $d_p$ belongs to a dense subalgebra $\A_L \subset C(\M)$ composed of functions which are bounded Lipschitz continuous, and respect the condition $\text{\rm ess} \sup \norm{  \nabla d_p} \leq 1$. Since a modification of $\nabla f$ in (\ref{rdist5}) for a discrete number of points does not affect the value of the integral, we can extend the formula (\ref{rdist8}) to all functions in $\A_L$ under the condition $\text{\rm ess} \sup \norm{  \nabla f} \leq 1$, with the equality given by the distance $d_p$. At the end, we have the formula:
\begin{equation}\label{rdistnabla}
d(p,q) = \sup\set{  \abs{f(q)-f(p)} \ :\  f \in \A_L,\  \text{\rm ess} \sup \norm{  \nabla f} \leq 1 }.
\end{equation}

This formula is very interesting on its own by the fact that it is totally independent of any particular choice of piecewise smooth curve. The path dependance on the initial definition has been removed.

\subsubsection*{Step 2: Operatorial formulation}

The condition \mbox{$\text{\rm ess} \sup \norm{  \nabla f} \leq 1$} for $ \in \A_L$ can be replaced by the following one: \mbox{$\norm{[D,f]} \leq 1$}.\\

To prove that, let us take a function $f\in\A$. As a consequence of the \mbox{Theorem \ref{commudirac}} and of the boundedness of the manifold $\M$, we have that \mbox{$[D,f] = -i\,c(df)$} is a densely defined bounded operator on $\H$, so its norm is given by:
\begin{eqnarray}
\norm{[D,f]}^2 = \norm{c(df)}^2 &=&  \sup_{\phi\in\H,\;\phi\neq0} \frac{(c(df)\phi,c(df)\phi)}{(\phi,\phi)}\nonumber\\
&=&  \sup_{\phi\in\H,\;\phi\neq0} \frac{ \int_\M c(df^*)\,c(df)\, \phi^*\phi \,d\mu_g}{\int_\M \phi^*\phi \,d\mu_g }\cdot\label{rdist15}
\end{eqnarray}
since we work with a self-adjoint Clifford action.\\

The integral formulation (\ref{rdist15}) shows that we can extend this calculation to  a.e.~differentiable functions $f\in\A_L$ without modification of the norm value $\norm{[D,f]}$.\\

By the relations of the Clifford algebra, and since $g$ stands here for the sesquilinear metric, we have:
\begin{equation}\label{rdist20}
c(df^*)\,c(df) =  g^{-1}(df,df) =  g(\nabla f,\nabla f).
\end{equation}

If we insert (\ref{rdist20}) in (\ref{rdist15}), we obtain:
\begin{eqnarray*}
\norm{[D,f]}^2  &=&\sup_{\phi\in\H,\;\phi\neq0} \frac{ \int_\M g(\nabla f,\nabla f)\;\phi^*\phi \,d\mu_g}{\int_\M \phi^*\phi \,d\mu_g }\\
&=& \underset{x\in\M}{\text{\rm ess} \sup} \ g_x(\nabla f_x,\nabla f_x) =  \text{\rm ess} \sup \norm{  \nabla f}^2.
\end{eqnarray*}

So the distance formula (\ref{rdistnabla}) becomes:
\begin{equation}\label{rdistdiracL}
d(p,q) = \sup\set{  \abs{f(q)-f(p)} \ :\  f \in \A_L,\  \norm{[D,f]} \leq 1 }.
\end{equation}

Actually, the role of the operator $[D,f]$ is more important than giving an upper bound for the gradient, it completely defines the Lipschitz algebra $\A_L$. Indeed, using a result of \cite{Weaver}, we have that the space of bounded continuous functions $f$ such that $[D,f]$ is bounded corresponds to the space of bounded Lipschitz functions. So we can use the closure algebra $\bar\A = C(\M)$ whose Lipschitz subalgebra $\A_L$ is given by the boundedness condition of $[D,f]$, and the final commutative version of the Riemannian formula (\ref{rdistdiracL}) is:
\begin{equation}\label{rdistdirac}
d(p,q) = \sup\set{  \abs{f(q)-f(p)} \ :\  f \in \bar\A,\  \norm{[D,f]} \leq 1 }.
\end{equation}

This time, the formula only involves the closure of the algebra $\A$ (as a normed algebra acting on the Hilbert space $\H = L^2(\M,S))$ and the Dirac operator $D$ densely defined on $\H$.

\subsubsection*{Step 3: Noncommutative generalization}

The distance formula (\ref{rdistdirac}) gives all the necessary ingredients to construct a distance for noncommutative spaces. Indeed, if we take a possibly noncommutative unital algebra $\A$ which is a pre-$C^*$-algebra, acting as usual on a Hilbert space $\H$ as bounded multiplicative operators, we have that $\bar\A$ is a $C^*$-algebra and we can consider its spectrum $\Delta(\bar\A)$. Then every self-adjoint operator $D$ densely defined on $\H$ with the condition that \mbox{$[D,a]$} is bounded for each $a\in\A$ gives rise to a well defined distance function between every pure states \mbox{$\xi,\eta\in\Delta(\bar\A)$}:
\begin{equation}\label{rdistnc}
d(\xi,\eta) = \sup\set{  \abs{\xi(a)-\eta(a)} \ :\  a \in \bar\A,\  \norm{[D,a]} \leq 1 }.\footnote{In the literature, the condition $a\in\bar\A$ is often replaced by the simpler one $a\in\A$, but this will correspond to the usual Riemannian distance (\ref{rdistdirac}) only if one can find a sequence $f_n\in C^\infty(\M)$ respecting $ \norm{[D,f_n]} \leq 1$ and converging to the usual distance function $d_p$.}
\end{equation}
The Lipschitz algebra is defined as \mbox{$\A_L =\set{a\in\bar\A : [D,a] \in\B(\H)}$}.\\

Actually, this function can be extended to the complete space of states (convex combinations of pure states)\footnote{We could mention here the works of M. A. Rieffel which provide an extended study of the Riemannian distance function on the space of states \cite{Rieffel}.}. There is no problem for the algebra $\A$ to be a discrete algebra, so we have a notion of distance valid even for discrete spaces.\\

We can check that this function respects all the conditions for a Riemannian distance.
\begin{itemize}
\item $d(\xi,\eta) \geq 0$ for all $\xi,\eta\in\Delta(\bar\A)$ is trivial
\item $d(\xi,\eta) = 0$ if and only if $p = q$ comes from the fact that two distinct states must differ at least on one $a \in \bar\A$
\item $d(\xi,\eta) = d(\eta,\xi)$ for all $\xi,\eta\in\Delta(\bar\A)$ is also trivial
\item The triangle inequality is valid since for all $\xi,\eta,\rho\in\Delta(\bar\A)$ we have
\begin{eqnarray*}
d(\xi,\rho) &=& \sup_{a \in \bar\A}\set{  \abs{\xi(a)-\rho(a)} \ :\norm{[D,a]} \leq 1 }\\
&\leq& \sup_{a \in \bar\A}\set{  \abs{\xi(a)-\eta(a)} + \abs{\eta(a)-\rho(a)} \ :\norm{[D,a]} \leq 1 }\\
&\leq& \sup_{a \in \bar\A}\set{  \abs{\xi(a)-\eta(a)} \ :\norm{[D,a]} \leq 1 }\\
&&+\   \sup_{a \in \bar\A}\set{  \abs{\eta(a)-\rho(a)} \ :\norm{[D,a]} \leq 1 }\\
&=&  d(\xi,\eta) + d(\eta,\rho)
\end{eqnarray*}
\end{itemize}

When $\A=C^\infty(\M)$ and $D$ is the Dirac operator on the Hilbert space $\H = L^2(\M,S)$ of square integrable spinor sections over $\M$, the formula (\ref{rdistnc}) is completely equivalent to the usual Riemannian distance (\ref{rdistdirac}) by using the isomorphism given by the Gel'fand transform.\\

\vspace{1em}                                                      

\subsection{Spectral Triples}\label{ST}

\vspace{1em}                                                      

We have almost completed the construction of the basic elements of Riemannian noncommutative geometry. By setting the noncommutative equivalents for integration, differential forms and distance, we have each time referred to abstract key ingredients as an algebra $\A$, a Hilbert space $\H$ and a Dirac-like operator $D$, with some conditions to guarantee the good definition of these objects as well as the correspondence in the commutative case to compact Riemannian spin manifolds. Those three ingredients form together the fundamental spaces of noncommutative geometry, called spectral triples.

\vspace{1em}                                                      

\begin{defn}\label{spectrip}
A {\bf\index{spectral triple} Spectral Triple} \mbox{$(\A,\H,D)$} is the data of:
\begin{itemize}
\item A Hilbert space $\mathcal{H}$
\item A unital pre-$C^*$-algebra $\mathcal{A}$ with a representation as bounded multiplicative operators on $\mathcal{H}$ 
\item A self-adjoint operator $D$ densely defined on $\mathcal{H}$ with compact resolvent such that all commutators $[D,a]$ are bounded for every $a\in\A$
\end{itemize}
\end{defn}

\vspace{1em}                                                      

Whenever the algebra $\A$ is non abelian, a spectral triple is clearly a noncommutative space, with those elements automatically defined:
\begin{itemize}
\item The spectrum \mbox{$\Delta(\bar\A)$} corresponds to a noncommutative geometrical space
\item Elements in the form \mbox{$a_0 [D,a_1] \cdot\cdot\cdot[D,a_p]$}, \mbox{$a_0, a_1, \dots , a_p \in\A$} generate a graded algebra of differential forms
\item The subspace $\set{a\in\bar\A : [D,a] \in\B(\H)} \subset \bar\A$ defines the Lipschitz algebra of the spectral triple
\item The function \mbox{$d(\xi,\eta) = \sup\set{  \abs{\xi(a)-\eta(a)} \ :\  a \in \bar\A,\  \norm{[D,a]} \leq 1 }$} defines a Riemannian distance on the space of states of $\bar\A$
\end{itemize}

\vspace{1em}                                                      

One can define alternatively the spectral triple directly on the closure algebra \mbox{$(\A = \bar\A,\H,D)$} and require the commutators $[D,a]$ to be bounded only on a dense subalgebra of $\A$ (the Lipschitz algebra). We do not favor this definition since it can lead to some confusion between the algebras $C(\M)$ and $C^\infty(\M)$ on the commutative case.\\

Additionally, some specific properties can be imposed to the spectral triples. As this theory is still developing, the conditions imposed on it are also developing  following the needs. We present here the properties as given in \cite{MC08}.

\vspace{1em}                                                      

\begin{defn}
A spectral triple is {\bf\index{spectral triple!finitely summable} finitely summable} (or \mbox{$n^+$-summable}) if there exists a positive integer $n$ such that the resolvent of $D$ has characteristic values \mbox{$\mu_k = O(k^{-n})$}. In this case, $D^{-1} \in \L^{n+}$, and $n$ is the {\bf\index{metric!dimension} metric dimension} of the spectral triple.
\end{defn}

\vspace{1em}                                                      

So this is the rate of growth of the eigenvalues of the Dirac operator which gives the information on the dimension of the spectral triple. Discrete spectral triples, defined with a finite dimensional Hilbert space $\H$ (and so with a finite number of eigenvalues for the Dirac operator), can be seen as noncommutative space of dimension zero, while non-finitely summable spectral triples represent infinite-dimensional noncommutative spaces.\\

For every finitely summable spectral triple of metric dimension $n$, we have a noncommutative integral defined by
$$\int\!\!\!\!\!\!-\ a \abs{D}^{-n} = \tr_\omega(a \abs{D}^{-n})$$
for every measurable element $a \abs{D}^{-n}$ with $a\in\A$.

\vspace{1em}                                                      

\begin{defn}
A spectral triple is {\bf\index{spectral triple!even} even}  if there exists a \mbox{$\mathbb Z_2$}-grading $\gamma$ such that \mbox{$[\gamma,a] = 0\ \forall a\in\mathcal{A}$} and \mbox{$\gamma D + D \gamma = 0$}.
\end{defn}
Spectral triples which are not even are simply called odd spectral triples.

\vspace{1em}                                                      

\begin{defn}\label{realst}
A spectral triple is {\bf\index{spectral triple!real} real of KO-dimension $n \in\mathbb Z_8$} if there exists an antilinear isometry \mbox{$J :\mathcal{H} \rightarrow \mathcal{H}$}  such that:
\begin{itemize}\itemsep=0pt
\item $J^2 = \epsilon$
\item $J D = \epsilon' D J$
\item $J \gamma = \epsilon'' \gamma\quad$ (if the spectral triple is even)  
\item $[a,b^ \circ] = 0\ \quad\forall a,b\in\mathcal{A}$
\item $[[D,a],b^ \circ] = 0\ \quad\forall a,b\in\mathcal{A}$
\end{itemize}
where $b^ \circ =Jb^*J^{-1}$ and where the numbers $\epsilon$, $\epsilon'$ and $\epsilon''$ are taken in the set $\left\{{-1,1}\right\}$ depending on the value of $n\!\!\mod 8$:
\begin{center}
\begin{tabular}{|c|rrrrrrrr|}
 \hline
  n &0 &1 &2 &3 &4 &5 &6 &7 \\
\hline \hline
$\epsilon $  &1 & 1&-1&-1&-1&-1& 1&1 \\
$\epsilon'$ &1 &-1&1 &1 &1 &-1& 1&1 \\
$\epsilon''$&1 &{}&-1&{}&1 &{}&-1&{} \\  
\hline
\end{tabular}
\end{center}
\end{defn}

Elements $b^\circ=Jb^*J^{-1}$ are the elements of the {\bf\index{algebra!opposite} opposite algebra} of $\A$, i.e.~the algebra \mbox{$\A^ \circ = \set{ a^\circ : a\in\A}$} which respects the inverse product \mbox{$a^\circ b^\circ = (ba)^\circ$}. The KO-dimension does not really play the role of a dimension, but rather the role of a signature of the noncommutative space.\\

We know that for every $n$-dimensional compact Riemannian spin-manifold, a spectral triple can be constructed by taking the Hilbert space $\H = L^2(\M,S)$ of square integrable spinor sections over $\M$, the pre-$C^*$-algebra $\A=C^\infty(\M)$ and the Dirac operator $D$. This spectral triple has metric dimension $n$. If $n$ is even, then the spectral triple is even by taking the chirality element \mbox{$\gamma = \gamma^1\dots\gamma^n$} as a \mbox{$\mathbb Z_2$}-grading, with the Dirac matrix $\gamma^\mu=c(dx^\mu)$. An antilinear isometry $J$ giving a real structure is known as a charge conjugation operator.\\

One can wonder if the converse is also true, so if every spectral triple with commutative algebra, maybe under suitable conditions, corresponds to a compact Riemannian spin-manifold, with $\A$ being the space of smooth functions on it? The answer is yes, and is known as the Connes' reconstruction theorem. This theorem was first conjectured in \cite{C96}, and was recently proven in details in \cite{C08} under slightly stronger axioms.


\newpage

\section[Noncommutative standard model]{Noncommutative standard model}

In the first part of this chapter we have seen that noncommutative geometry provides an interesting mathematical tool in the form of spectral triples which can lead to translate Euclidean gravity into an algebraic formalism (Euclidean has here the meaning of a Riemannian signature) as well as giving a geometrical interpretation to noncommutative algebraic spaces, with the definitions of distance, integrals and differential forms.\\

In the Chapter \ref{chapterintro}, we have asked the question of the existence of new mathematical tools which can describe gravitation and include in the same time a description of the other fundamental interactions. Noncommutative geometry gives a positive answer to this question, by modeling the standard model of particles (at least on a classical level) of the form of a spectral triple, and then by creating a product of spectral triples representing a minimal coupling between Euclidean gravity and a classical standard model.\\

So to conclude this chapter on Euclidean noncommutative geometry, we will give a quick overview of the construction of this noncommutative standard model (noncommutative is used here in the sense that this model is based on noncommutative geometry). A complete version of this theory (with massive neutrinos) was given in \cite{MC207} by A.~H.~Chamseddine, A.~Connes and M.~Marcolli, and largely discussed in \cite{MC08}. The latest version of this theory can be found in \cite{CC10}. An interesting physical presentation can also be found in \cite{Schuk}.\\

First of all, we must explain how the product geometry is constructed. The idea is similar to a Kaluza--Klein theory. The spectral triple is taken to be the product of a spectral triple associated to the commutative geometry of a compact $4$-dimensional Riemannian spin manifold $\M$ with a spectral triple associated to a finite noncommutative geometry $F$.\\

The gravitational part \mbox{$(\A_\M,\H_\M,D_\M) = (C^\infty(\M),L^2(\M,S),D)$} is constructed in the same way as in the Section \ref{secgrav}, with $D$ the Dirac operator associated to the spin structure.  We can add the chirality element\footnote{The indices $\set{0,1,2,3}$ are generally used even if the signature is positive.}  \mbox{$\gamma_\M = \gamma^0\gamma^1\gamma^2\gamma^3$} and the $4$-dimensional charge conjugation operator $J_\M$ defined  by \mbox{$J_\M = C : \psi \rightarrow i\gamma^2\gamma^0 \bar\psi$} to obtain a real even spectral triple of KO-dimension $4$ modulo $8$ : \mbox{$\M = (\A_\M,\H_\M,D_\M,\gamma_\M,J_\M)$}.\\

The finite geometry is given by another real even spectral triple \mbox{$F = (\A_F,\H_F,D_F,\gamma_F,J_F)$}. Then the product geometry \mbox{$\M \times F = (\A,\H,D,\gamma,J)$} is constructed in the following way:
\begin{itemize}
\item $\mathcal{A} = \mathcal{A}_\M \otimes \mathcal{A}_F$
\item $\mathcal{H} =  \mathcal{H}_\M \otimes \mathcal{H}_F$
\item $D = D_\M \otimes 1 + \gamma_\M \otimes D_F$
\item $\gamma =  \gamma_\M \otimes \gamma_F$
\item $J  = J_\M \otimes J_F$
\end{itemize}
One can check that those products give rise to a well defined spectral triple. The KO-dimension of the product geometry is the sum modulo 8 of the KO-dimensions of both spaces. We can note that a real structure is not mandatory in order to construct a product of spectral triple, but one at least must be even.\\

The first and most simple model of product geometry was introduced in \cite{CH93} as the product of a spin manifold and a two-points space. Then, the  finite algebra was chosen to be $M_n(\setC)$, the algebra of complex $n\times n$ matrix, and the result was comparable to an Einstein--Yang--Mills system \cite{CH93/2}. The construction of the complete standard model requires to find an algebra such that the group of "diffeomorphisms" of the finite noncommutative geometry corresponds to the gauge groups of the standard model \cite{CC97,C96}.\\

To see what the term diffeomorphism can mean for a noncommutative geometry, we can notice that in the Riemannian commutative case the group $\text{Diff}(\M)$ of diffeomorphisms of $\M$ is in a one to one correspondence with the group $\text{Aut}(C^\infty(\M))$ of automorphisms (endomorphisms which are isometric) of the algebra of smooth functions on $\M$. Indeed, each $\varphi\in\text{Diff}(\M)$ is associated to the isomorphism $\alpha : C^\infty(\M) \rightarrow C^\infty(\M)$ by the relation \mbox{$\alpha(f) = f\circ\varphi^{-1}$} for each $f\in C^\infty(\M)$. So a diffeomorphism for a noncommutative space would just correspond to an automorphism of the noncommutative algebra.\\

If we define the unitary group for a unital involutive noncommutative algebra $\A$ by \mbox{$\U=\set{ u\in\A : uu^* = u^*u = 1}$}, then the space of {\bf\index{inner automorphism} inner automorphisms}
$$\text{Int}(\A) = \set{\alpha \in\A : \exists \,u\in\U, \;\alpha(a) = u\,a\,u^*\ \forall a\in\A}$$
is a non-empty subgroup of $\text{Aut}(\A)$, and corresponds to  internal symmetries of the noncommutative space. If we take the finite spectral triple \mbox{$F = (\A_F,\H_F,D_F,\gamma_F,J_F)$}, then the unitary group of $\A_F$ defines an action on $\H_F$ by the adjoint representation:
$$\text{Ad}(u) \,\xi = u\,\xi\,u^* = u(u^*)^\circ \,\xi\qquad \forall \xi\in\H\quad \forall u\in\A\ :\ uu^* = u^*u = 1$$
with $(u^*)^\circ = J u J^{-1}$ in the opposite algebra. So in order to reproduce the standard model, we would ask the algebra $\A_F$ to be chosen such that the adjoint representations act as elements in the gauge group $U(1)\times SU(2) \times SU(3)$.\\

Those inner automorphisms can be interpreted as gauge transformations on the noncommutative space. Those transformations affect the Dirac operator $D$ as a kind of inner fluctuations of the metric. Actually, we can replace the Dirac operator $D$ by a covariant formulation:
$$D_A = D + A + \epsilon'JAJ^{-1}$$
where $A$ is a Hermitian $1$-form \mbox{$A=\sum_i a_i [D,b_i]$}, \mbox{$a_i,b_i \in\A$}, $A=A^*$.
The covariant property is given by the following proposition \cite{MC08}:
\begin{prop}
If \mbox{$(\A,\H,D)$} is a real spectral triple with antilinear isometry $J$, then for any gauge potential \mbox{$A\in\Omega^1_D$} with \mbox{$A=A^*$} and any unitary \mbox{$uu^* = u^*u = 1$}, $u\in\A$, one has
$$\text{Ad}(u) \prt{D + A + \epsilon'JAJ^{-1} }  \text{Ad}(u^*) = D + A' + \epsilon'JA'J^{-1}$$
where
$$ A' = u[D,u^*] + uAu^*.$$
\end{prop}
To summarize, we have the gauge transformations
$$\xi\in\H \rightarrow u\,\xi\,u^* = u(u^*)^\circ \,\xi,\qquad A\in\Omega^1_D \rightarrow u[D,u^*] + uAu^*$$
with a Hermitian gauge field $A\in\Omega^1_D$ which gives rise to the existence of gauge bosons in noncommutative geometry.\\

The action proposed by A. Connes and A. H. Chamseddine \cite{CC97} is based on this covariant Dirac operator, with a cut-off parameter $\Lambda$ which fixes the mass scale and a positive even functional $f$:
\begin{equation}\label{spectaction}
S =\tr\left({ f\left({\frac{D_A}{\Lambda}}\right)}\right)\cdot
\end{equation}
This action is called the {\bf\index{spectral action} spectral action}, and requires to assume the following principle:
\begin{defn}
{\bf\index{spectral action!principle} Spectral action principle}:  The physical action depends only on the spectum of the Dirac operator. 
\end{defn}

In the special case of gravity, this spectral invariance is however a stronger condition than the usual diffeomorphism invariance since there exist manifolds which are isospectral without being isometric.\\

Typically, the action (\ref{spectaction}) is calculated by use of the Lichnerowicz' formula (Theorem \ref{Lich}) on the square \mbox{$\frac{D_A^2}{\Lambda^2}$} and by computing the trace with a method of heat kernel expansion \cite{CC97}. This action only handles the bosonic part. In order to account for the fermionic part an additional term must be added:
$$S =\tr\left({ f\left({\frac{D_A}{\Lambda}}\right)}\right)\ +\ \frac 12 \scal{J\xi,D_A\xi}\cdot$$

The construction of the finite geometry is quite complicated, and can be found in \cite{MC08}. We present here a sketch of this construction.\\

First we define the algebra
$$\mathcal{A}_{LR} = \mathbb C \oplus \mathbb H_L\oplus \mathbb H_R \oplus M_3(\mathbb C)$$
where $\mathbb H$ is the algebra of quaternions, and $\mathbb H_L$ and $\mathbb H_R$ are just two copies of this algebra, labelled left and right (and they will have a correspondence with left-handed and right-handed fermions).\\

We will use irreducible representations of each element of this direct sum of algebras, and denote these representations with the following notations:
\begin{itemize}
\item ${\bf 1}$ is the $1$-dimensional irreducible representation of $\mathbb C$, and ${\bf 1^\circ}$ is its opposite (the irreducible representation of the opposite algebra)
\item ${\bf 2}$ is the $2$-dimensional irreducible (complex) representation of $\mathbb H$, and ${\bf 2^\circ}$ is its opposite, with ${\bf 2}_L$ and ${\bf 2}_R$ begin the distinction between  $\mathbb H_L$ and $\mathbb H_R$
\item ${\bf 3}$ is the $3$-dimensional irreducible representation of $M_3(\mathbb C)$, and ${\bf 3^\circ}$ is its opposite
\end{itemize}

Then $M_F = \mathcal E \oplus \mathcal E^ \circ $ with
$$\mathcal E ={{\bf 2}_L \otimes {\bf 1^ \circ}} \oplus{{\bf 2}_R \otimes{\bf 1^ \circ}} \oplus { {\bf 2}_L \otimes{\bf 3^ \circ} } \oplus {{\bf 2}_R \otimes{\bf 3^ \circ}} $$
is a $\mathcal{A}_{LR}$-bimodule. The decomposition $\mathcal E \oplus \mathcal E^ \circ $ corresponds to the physical interpretation between particles and anti-particles inside one generation of fermions. To obtain the Hilbert space, we sum this bimodule  a number of times corresponding to the number of generations of fermions. So we define the Hilbert space by
$$\H_F =  M_F \oplus M_F \oplus M_F = \H_f \oplus \H_{\bar f}$$
with a suitable trace, and where $\H_f = \mathcal E \oplus \mathcal E \oplus \mathcal E $ and $\H_{\bar f} = \mathcal E^ \circ\oplus \mathcal E^ \circ\oplus \mathcal E^ \circ$.\\

The algebra $\A_F$ is chosen as a subalgebra of $\mathcal{A}_{LR}$ such that the existence of a Dirac operator is guaranteed under real spectral triple conditions (Definition \ref{realst}) and under the extra condition that the Dirac operator intertwines the subspace $\H_f$ and $\H_{\bar f}$. To guarantee this intertwining, the finite algebra is chosen to be
$$\mathcal{A}_F = \mathbb C \oplus \mathbb H \oplus M_3(\mathbb C).$$

The antilinear isometry $J_F$ is defined from its action on $M_F = \mathcal E \oplus \mathcal E^ \circ $ by
$$J_F (\xi,\bar\eta) = (\eta,\bar\xi)\qquad \forall \xi,\eta \in \mathcal E$$
and the $\setZ_2$-grading by using the $\setZ_2$-grading given by $\mathbb{H}$:
$$\gamma_F =  c - J_F c J_F,\quad c=(0,1,-1,0)\in\mathcal{A}_{LR}.$$

From these settings, it is possible to construct a general form for a Dirac operator acting on  $\H_F$ and respecting the real even spectral triple conditions, form that we do not want to explicit here. All details about the classification of the Dirac operators can be found in \cite{MC08}. We will just make the remark that the free elements in the construction of the Dirac operator correspond to the Yukawa parameters of the standard model.\\

The finite geometry is of KO-dimension $6$ modulo $8$, so the product geometry is of KO-dimension $2$ modulo $8$. We have then those nice results \cite{MC08}:
\begin{itemize}
\item If \mbox{$\U(\A_F) = \set{ u\in\A_F : uu^* = u^*u = 1}$} is the unitary group of the finite algebra $\A_F$ and \mbox{$SU(\A_F) = \set{u\in\U{\A_F} : \det(u) = 1}$} its special orthonormal subgroup, where $\det(u)$ is the determinant of the adjoint action of $u$ on $\H_F$, then modulo a finite abelian group the group $SU(\mathcal{A}_F)$ is of the form:
$$SU(\mathcal{A}_F) \cong U(1) \times SU(2) \times SU(3)$$
\item The action of the $U(1)$ subgroup is modulo a finite abelian group of the form \mbox{$u(\lambda) = (\lambda^\mu,1,\lambda^\nu 1_3) \in SU(\mathcal{A}_F)$}, $\mu,\nu\in\setR$, $\lambda\in U(1)$, and the corresponding adjoint action $\text{Ad}(u) = u\prt{u^*}^\circ$ is then a multiplication of the basis vectors of $\H_F$ by powers of $\lambda$. If we denote by $|\!\!\uparrow\rangle $ and $|\!\!\downarrow \rangle$ the basis of ${\bf 2}$ such that the action of $\lambda$ is diagonal on the basis, then the different powers of $\lambda$ can be written under the following table:
$$\begin{array}{ccccc}
&|\!\!\uparrow\rangle \otimes {\bf 1^\circ}  &|\!\!\downarrow \rangle \otimes {\bf 1^ \circ}  &|\!\!\uparrow \rangle \otimes {\bf 3^ \circ}  &|\!\!\downarrow \rangle \otimes {\bf 3^ \circ}\\
{\bf 2}_L & -1 & -1 & \frac 13 & \frac 13\\
{\bf 2}_R & 0 & -2 & \frac 43 & -\frac 23\\
\end{array}$$
Those powers correspond to the hypercharges of the fermions of the standard model.
\item The inner fluctuations of the metric for the product geometry $\M\times F$ can be separated into two parts, since $D$ is of the form \mbox{$D_\M \otimes 1 + \gamma_\M \otimes D_F$}. The fluctuations of the continuous part $D_\M \otimes 1$ correspond to a $U(1)$ gauge field, a $SU(2)$ gauge field and a $U(3)$ gauge field, which reduces to a $SU(3)$ gauge field if we restrict to gauge fields such that $\tr(A)=0$, so the 12 gauge bosons of the standard model are recovered. The fluctuations of the discrete part $\gamma_\M \otimes D_F$ correspond to an arbitrary quaternion-valued function which corresponds to a Higgs field.\\
\end{itemize}

To conclude, we have seen here that the framework of spectral triples can be used to construct a model which combines both Euclidean gravitation and a classical standard model. By this way, the standard model of particles is conferred a geometrical interpretation as a noncommutative geometrical space. Of course, in order to have a real theory of all the four fundamental interactions, one needs to solve the following two quite annoying problems:
\begin{itemize}
\item The standard model presented here is only defined at a classical level, so a way to quantify this model must still be found. Few works about this problem have been done, mainly on particular examples of noncommutative spaces (\cite{Bes07,Hale,HS,HS2,Man, Rov99}), but a complete quantization is still out of sight.
\item Since only Euclidean gravity is involved, so with positive signature, this model has no real physical interpretation at this time. The attempts  to generalize this theory to Lorentzian manifolds is the complete topic of our Chapter \ref{chaplor}.
\end{itemize}


\cleardoublepage
\hbox{} \vspace*{\fill} \thispagestyle{empty}
\chapter[Lorentzian noncommutative geometry]{An attempt to generalize noncommutative geometry to Lorentzian geometry}\label{chaplor}

In the Chapter \ref{chaprie} we have introduced Connes' noncommutative geometry. This theory provides a good mathematical background which allows us to construct algebraic noncommutative spaces with geometrical interpretation. In particular the theory can be used to construct a model which combines Euclidean gravitation and a classical standard model.\\

However this theory is only developed in the case of compact Riemannian manifolds, and the gravitational model obtained has a positive signature. This is not satisfactory since physical theories like general relativity are based on spacetimes with Lorentzian signature. So the theory of noncommutative geometry is not at this time a real physical theory, but only a technical background having its interest mainly at a mathematical level. If we want this theory to lead one day to a real physical unification between gravitation and the other fundamental interactions, a complete Lorentzian counterpart must be produced, but such a fully complete theory is still out of sight. \\

Connes' theory of noncommutative geometry is more than $25$ years old now, but the consideration on the hyperbolic case possesses only half this age. Moreover, if the Riemannian formulation has led to a quite good amount of literature, only a few attempts have been made about the Lorentzian formulation.\\

Generalizing the theory to Lorentzian spaces is where our research comes in. Of course our goal is not to find a complete solution to this wide open problem, but to propose significative improvements in this direction, with the hope that they can become some parts of a final solution.\\

We will begin this chapter by introducing different problems generated by the addition of a Lorentzian signature, and by making a review of the first existing results. Then the following will be a detailed report on our contributions about this problem.\\


\newpage

\section[Current attempts]{Generalization to Lorentzian spaces:\\ Current attempts}

We have said that the theory of noncommutative geometry is only valid at this time for Riemannian manifolds, but we have not expressed at which places the extension to Lorentzian manifolds can be problematic. It is impossible to give at this time an exhaustive list  of all problems and difficulties which could arise from the Lorentzian signature since the attempts to bypass some of them can lead to new ones. However, we will try to give a reasonable list of the main problems that one can encounter while evolving in the Lorentzian generalization.\\

\subsection{Problems coming from the generalization}\label{prob}

The problems encountered can be divided in two kinds: technical problems and conceptual problems, but some conceptual problems can also induce new technical problems.

\subsubsection*{Technical problems}

The technical problems come mainly from the definition of a spectral triple involving the three elements $\A$, $\H$ and $D$. $\H = L^2(\M,S)$ is the Hilbert space of square integrable sections of the spinor bundle $S$ over $\M$ with a Hermitian inner product defined by:
$$(\psi,\phi) = \int_\M \psi^*\phi\, d\mu_g.$$
In the Riemannian case, $D$ is an essentially self-adjoint elliptic operator with real discrete spectrum. In the Lorentzian case, $D$ is not an essentially self-adjoint operator any more, which implies a more complicated spectrum, and it is not elliptic any more, which implies several issues as the non-compactness of its resolvent, an ill-definition of its inverse modulo smoothing operation or singularities appearing in the domain of smoothness of the Dirac operator. As a consequence, all the elements defined in the Section \ref{secgrav} as integration, differential forms and Lipschitz algebra are not automatically conserved in the Lorentzian case.\\

So if one wants to conserve a similar notion of spectral triple, an important adaptation to the space $\H$ and to the Dirac operator $D$ must be performed, with if possible the conservation of similar notions of integration or differential forms.

\subsubsection*{Conceptual problems}

The main conceptual problem of the generalization is the introduction in Lorentzian geometry of the notion of causality. Roughly speaking\footnote{All elements about Lorentzian geometry will be mathematically defined in the Section \ref{Lorgeo}.}, causality is the fact that a point (an event) of the spacetime can only be in a physical relation with specific points of the spacetime, with some points being under the influence of that point, and some other points influencing that point. Such notion, which can be compared to a notion of partial order on the space of points, is totally absent from Connes' theory.\\

The notion of distance on Lorentzian spaces is completely dependent on this notion of partial causal order. Indeed, the distance between two points must be positive only if the two points are causally related, so the distance between two causally unrelated distinct points must be zero. Moreover, the distance must not be symmetric any more, since the distance from a point $p$ to a point $q$ is positive only if $q$ is in the following of $p$ for this order (this corresponds to $q$ being in the future of $p$). The distance defined in (\ref{rdistnc}) is clearly symmetric and is not null for any two distinct points, so in no way it can give rise to a Lorentzian distance.\\

So we can see that two important and related concepts -- Lorentzian distance and causality -- are totally absent of the theory of noncommutative geometry and must be introduced in a Lorentzian generalization.

\subsubsection*{A new technical problem: the non-compactness}

The introduction of causality leads  to the showing up of another important technical problem: the non-compactness of the manifold. Indeed, the current theory is mainly set for compact manifolds, with the compactness condition being very useful to avoid technical complications. However it is well known that Lorentzian compact manifolds do not accept a well defined causal structure. Actually, Lorentzian compact manifolds imply the existence of points which are in the future of themselves (see \cite{Beem,Gallo,Tipler}), which is once more not very physically realistic.\\

We have presented the Section \ref{secgrav} only for the compact case, as it is usually done. As we have said, the compactness condition is most of the time more a simplicity condition than a problematic one. For example, in the non-compact case the algebra $\A$ must be non-unital, and must correspond to functions vanishing at infinity in the commutative case. Since the Gel'fand--Naimark theorem has a non-unital version, this is not a priori a problem to consider non-unital algebras, even if sometimes a unitization could be necessary. The definition of spectral triple (Definition \ref{spectrip}) must be adapted in order to account for non-unital algebras, with the replacement of the condition on the compactness of the resolvent of $D$ by the condition:
\begin{equation}\label{diracncompact}
a \prt{D - \lambda 1}^{-1}\ \text{ is compact }\quad \forall a\in\A, \ \forall\lambda\notin\sigma(D).
\end{equation}

Constructions of non-compact noncommutative spaces can be found e.g.~in \cite{Gayral} with the Moyal plane. Extensions of the distance function can also be set for non-compact spaces \cite{DAM}, and especially for the Moyal plane \cite{Ca,CW}.\\

However not all the difficulties necessarily occur in a Lorentzian generalization since some elements which are problematic for the non-compact case are not even present in the Lorentzian case. For example, the boundedness condition of the Lorentzian distance function $d_p$ becomes irrelevant since this function is not  Lipschitz any more in the Lorentzian case, so this difficulty is actually replaced by an even more difficult one.\\

Nevertheless, the mandatory non-compactness condition of the manifold will lead to a quite large number of technical difficulties while trying to generalize the theory to Lorentzian spaces. This is mainly due to an important consideration: the fact that causality is a local concept, and that locality is something which is not a priori present in algebraic structures. So we can often be facing the problem of enlarging a local concept to the whole manifold, which could  be straightforward with a presence of finite open covering but not with non-compact manifolds.\\

As a consequence, global considerations on the manifold will be useful tools, especially about causality. This leads to the fact that most of the approaches to generalize the theory to Lorentzian manifold are done under the hypothesis of global hyperbolicity. This is logical since globally hyperbolic spacetimes carry a global information on causality by the existence of a global time function, as well as many useful properties that we will expose in the Section \ref{Lorgeo}.\\

\subsection{Review of the literature}

Literature on the subject of generalization of Connes' noncommutative geometry to Lorentzian spaces is so small at this time that we should be able to make an almost exhaustive review over it, of course at the best of our knowledge.\\

There exist different approaches, but no one with a complete solution. Actually, many of those approaches could consist each one as a piece from a big puzzle, with a final solution which would combine different aspects from diverse attempts. This is coherent since we have seen that the generalization to Lorentzian spaces needs to solve different problems, in relation to each other but with some particular distinct elements.\\

All existing attempts could be classified in three groups, even if some of those approaches were developed independently. So these groups do not really represent some known axes of research, but the problems they try to solve are quite similar, and in result the mathematical tools developed are often related in some way.

\subsubsection*{Group 1: Hamiltonian noncommutative geometry}

The first group does not really consist of approaches to create an equivalent Lorentzian version of Connes' noncommutative geometry, but it consists more of attempts to apply the existing Riemannian model to Lorentzian spaces.\\

A first way is to consider the Hamiltonian formulation (ADM formalism \cite{ADM}) of general relativity for globally hyperbolic spacetimes. In this formalism, the spacetime is divided in a $3+1$ decomposition, so it is foliated into a family of $3$-dimensional Riemannian  manifolds governed by a Hamiltonian equation, together with fitting conditions given by Lagrange multipliers called lapse and shift. Since the foliation gives rise to a family of Riemannian manifolds indexed by a time function $t$, one can consider a family of spectral triples \mbox{$(\A_t,\H_t,D_t)$} associated with each slice. Actually, this is sufficient to consider a unique Hilbert space $\H$ which is isomorphic to each $\H_t$ and then to study the evolution of $\A_t$ and $D_t$ by using a time evolution operator.  This approach is done by E. Hawkings in \cite{Haw} and by T. Kopf in \cite{Kopf98,Kopf00}.\\

A development of this idea is given by T. Kopf  and M. Paschke in \cite{Kopf01,Kopf02} by considering the algebras $A_t$ as elements of a category, where the morphisms between them are isomorphisms  which are given by the time evolution operator. So the system uses a groupoid as input. This groupoid is one of the main problems of this attempt since it involves a huge amount of data, with no known conditions to restrict them. Moreover this model cannot give any information on elements like the metric, the distance, differential calculus, integration, or even any information about causality. A similar development is given in \cite{Pas04} in the framewok of quantum field theory.\\\

Another way is to apply Euclidean noncommutative geometry not to the foliated space but to the space of connections on it, and more precisely to the space of connections defined in the theory of quantum gravity (see the Section \ref{QG}). This leads to an intersection between Connes' noncommutative geometry and quantum gravity proposed  by J. Aastrup, J. M. Grimstrup and R. Nest \cite{AaG06,AaG07,AaG091,AaG092,AaG093,AaG094,AaG112,AaG11}.\\

We will not develop these theories here since there are more an adaptation of the Riemannian model to some particular spaces referring to Lorentzian geometry than a complete search for a Lorentzian counterpart.

\vspace{1em}                                                      

\subsubsection*{Group 2: Pseudo-Riemannian spectral triples}

The goal of the second group is clear: adapting the construction of spectral triples \mbox{$(\A,\H,D)$} in order to account for pseudo-Riemannian signatures, et especially Lorentzian ones.\\

The main work about that is given by A. Strohmaier \cite{Stro} who introduces the concept of Krein space in order to recover the self-adjointness of the Dirac operator $\H$, and in the same time suggests an elliptic adaptation of this operator. Since this element is really interesting and could play an important role in a hypothetic complete generalization, we will expose the basis of this theory in the Section \ref{Krein}.\\

These pseudo-Riemannian spectral triples have been constructed onto some examples, as the Lorentzian noncommutative torus \cite{Stro}, the  Lorentzian noncommutative  cylinder \cite{Suij} or the Lorentzian noncommtuative $3$-sphere \cite{Pas}. The last citation seems to be an element of an extended research conducted mainly by M. Paschke about Lorentzian spectral triples, but without further publications at this time.\\

However, this adaptation is only technical since there is no way until now to recover the causal information from a Lorentzian spectral triple, because the only distance function available is Riemannian. Moreover, some informations from the Euclidean case as the noncommutative integral are translated in this formalism under the condition of compactness of the Lorentzian manifold, but we know that compact Lorentzian manifolds must be avoided. Likewise, most of the examples are based on compact Lorentzian manifolds. So this theory should need  further developments in order to clarify which elements can be extended to non-unital algebras and how this could be done, as it was initiated in \cite{Suij}.\\

To these pseudo-Riemannian spectral triples, we should add the work by J. Barrett \cite{Barret} which presents an adaptation of the spectral triple product of the standard model coupled with gravity with a Lorentzian signature (more precisely with a modification of the \mbox{KO-dimension}).

\subsubsection*{Group 3: Causal noncommutative geometry}

The last group is our main interest. The goal is not to technically adapt Euclidean noncommutative geometry to Lorentzian spaces, but to work on the conceptual elements that are introduced by Lorentzian signatures and do not currently exist in the Euclidean model. The main element is the establishment of a Lorentzian distance function. The notion of causality is also important, as well as the notion of {\it time}, still unknown in noncommutative geometry. Of course causality and Lorentzian distance are two strongly related concepts.\\

The first lines in this direction is written by G.N. Partfionov and R. R. Zapatrin in \cite{PZ} where they try to obtain a first conceptual formulation of what could be a Lorentzian distance function, with an example in Minkowski spacetime.\\

The first technical formulation of a Lorentzian distance for globally hyperbolic spacetimes is given by V. Moretti in \cite{Mor}, which can be considered as the starting point of the introduction of causality in noncommutative geometry. A path independent Lorentzian distance is given using local  conditions and an operatorial formulation is proposed by use of the Laplace--Beltrami--d'Alembert operator. Then a possibility of noncommutative generalization of the causal local elements is sketched. This approach is extremely complicated and it has never given rise to further developments, but many elements could be considered as of great interest.\\

It is at this point that our research in the subject takes place. In \cite{F2} we present a more developed conceptual formulation of a global Lorentzian distance function for globally hyperbolic spacetimes. Then in \cite{F3} we present a technical realization of this global Lorentzian distance function with the presentation of a path independent formulation. Despite a different approach, the function obtained possesses some similarities with the function proposed by V. Moretti, except the important fact that the local conditions are replaced by global ones, so this formulation could more easily give rise to a noncommutative generalization. The establishment of this function will be the entire topic of the Section \ref{secdist}.\\

As we have said, causality is a concept directly related to the Lorentzian distance, but since there is no counterpart in noncommutative geometry it deserves a research on its own. A work by F. Besnard \cite{Bes} presents a noncommutative generalization of the concept of causal order. Actually, a noncommutative generalization of the notion of completely separated ordered spaces is given on the compact case, but with no guarantee that the order corresponds to one induced by a Lorentzian structure. We will give a quick look to this theory in the Section \ref{seccaus}. In this same section, we will present some unpublished research about a key element in causal noncommutative geometry: the set of causal functions. We will be mainly interested in the establishment of a normed algebraic structure which can accept those functions, and which will lead to an extension of the notion of Lorentzian spectral triple including a temporal element.\\

\subsection{Pseudo-Riemannian Spectral Triples}\label{Krein}

We give here a review of the new tools introduced by A. Strohmaier in order to adapt the construction of spectral triples to pseudo-Riemannian manifolds. We will just be interested here in the results and we refer to  \cite{Stro} for the different proofs. The elements about Dirac operators in pseudo-Riemannian geometry can also be found in \cite{BaumG,BaumE}.\\

A pseudo-Riemannian spectral triple is, in the same way as a Riemannian spectral triple, a triple \mbox{$\prt{\A,\H,D}$} which corresponds in the commutative case to the algebra \mbox{$\A = C^\infty_0(\M)$} over a pseudo-Riemannian spin manifold $\M$ of signature \mbox{$(p,q)$} (with $q\geq 1$), to the Hilbert space $\H$ consisting of square integrable sections of the spinor bundle over $\M$ and on which there exists a representation of $\A$ as multiplicative bounded operators, and to the Dirac operator \mbox{$D = -i(\hat c \circ \nabla^S)$} acting on the space $\H$.\\

The Hilbert space $\H$ is endowed with the positive definite Hermitian structure
$$(\psi,\phi) = \int_\M \psi^*\phi\, d\mu_g$$
where $d\mu_g = \sqrt{\abs{\det g}} \;d^nx$ is the pseudo-Riemannian density on $\M$.\\

However, this structure does not admit any Dirac self-adjoint operator. Instead, a Dirac operator is an essentially Krein-self-adjoint operator if we transform $\H$ into a Krein space. We will give here the basic notions about Krein spaces. For further informations on Krein spaces, we refer the reader to \cite{Bog}.

\begin{defn}
An {\bf\index{inner product!indefinite} indefinite inner product} on a vector space $V$ is a map $V \times V \rightarrow \mathbb C$ which satisfies 
$$(v,\lambda w_1 + \mu w_2) = \lambda (v,w_1) + \mu(v, w_2), \qquad \overline{(v,w)} = (w,v).$$
An indefinite inner product is non-degenerate if
$$(v,w)=0 \quad\forall v\in V \ \ \Rightarrow\ \  w = 0.$$
\end{defn}

Let us suppose that $V$ can be written as the direct sum of two orthogonal spaces \mbox{$V = V^+ \oplus V^-$} such that the inner product is positive definite on $V^+$ and negative definite on $V^-$. Then the two spaces $V^+$ and $V^-$ are two pre-Hilbert spaces by the induced inner product (with a multiplication by $-1$ on the inner product for the second one).

\begin{defn}
If the two subspaces $V^+$ and $V^-$ are complete in the norm induced on them and if the indefinite inner product on $V$ is non-degenerate, then the space $V = V^+ \oplus V^-$ is called a \mbox{\bf\index{Krein space} Krein space}. The  indefinite inner product is called a \mbox{\bf\index{inner product!Krein} Krein inner product}.
\end{defn}

\begin{defn}
For every decomposition \mbox{$V = V^+ \oplus V^-$} the operator \mbox{$\J = \text{id}_{V^+} \oplus -\text{id}_{V^-}$} respecting the property $\J^2 = 1$ is called a {\bf\index{fundamental symmetry} fundamental symmetry}. Such operator defines a positive definite inner product (called the $\J$-product) on $V$ by \mbox{$\scal{\,\cdot\,,\,\cdot\,}_\J = \prt{\,\cdot\,,\J\,\cdot\,}$}.
\end{defn}

Each fundamental symmetry of a Krein space $V$ defines a Hilbert space structure, and two norms associated with two different fundamental symmetries are equivalent. So it is natural to defined the space of bounded operators $\B(V)$ as the space of bounded operators on the Hilbert space defined for any fundamental symmetry.

\begin{defn}
If $A$ is a densely defined linear operator on $V$, the {\bf\index{operator!Krein-adjoint} Krein-adjoint $A^+$} of $A$ is the adjoint operator defined for the Krein inner product \mbox{$(\,\cdot\,,\,\cdot\,)$}. An operator $A$ is called {\bf\index{operator!Krein-self-adjoint} Krein-self-adjoint} if \mbox{$A=A^+$}.
\end{defn}

Of course, for any fundamental symmetry $\J$ we can define an adjoint $A^*$ for the $\J$-product $\scal{\,\cdot\,,\,\cdot\,}_\J$. In this case, the Krein-adjoint is related to it by \mbox{$A^+ = \J A^* \J$}, and an operator $A$ is Krein-self-adjoint if and only if $\J\!A$ or $A\,\J$ are self-adjoint for the $\J$-product.\\

Now the question is how we could define a Krein space structure from a spin manifold with pseudo-Riemannian metric. This is done by using spacelike reflections.

\begin{defn}
A {\bf\index{spacelike reflection} spacelike reflection $r$} is an automorphism of the vector bundle $T\M$ such that:
\begin{itemize}
\item $g(r\,\cdot\,,r\,\cdot\,) = g(\,\cdot\,,\,\cdot\,)$
\item $r^2=\text{id}$
\item $g^r(\,\cdot\,,\,\cdot\,) = (\,\cdot\,,r\,\cdot\,)$ is a positive definite metric on $T\M$
\end{itemize}
\end{defn}

It is clear that, for every pseudo-Riemannian metric of signature $(p,q)$, the tangent bundle can be split into an orthogonal direct sum \mbox{$T\M = T\M_+^p \oplus T\M_-^q$} where the metric is positive definite on the $p$-dimensional bundle $T\M_+^p$ and negative on the $q$-dimensional bundle $T\M_-^q$, and so a spacelike reflection is automatically associated by defining \mbox{$r(v_{+|x} \oplus v_{-|x}) = v_{+|x} \oplus -v_{-|x}$}. This splitting is  transposed to the cotangent bundle \mbox{$T\M^* = T\M_+^{*p} \oplus T\M_-^{*q}$} by isomorphism.

\begin{prop}\label{propslr}
For each spacelike reflection $r$, there is an associated fundamental symmetry $\J_r$ defined from the Clifford action $c$ on a local oriented orthonormal basis \mbox{$\set{e_1,e_2,\dots,e_q}$} of $T\M_-^{*q}$ by
$$\J_r = i^{\frac{q(q+1)}{2}} c(e_1)c(e_2)\dots c(e_q) =  i^{\frac{q(q+1)}{2}}\gamma^1\gamma^2\dots\gamma^q.$$
\end{prop}

This definition is independent of the choice of the local basis. With such fundamental symmetry, the space  $\H$  of square integrable sections of the spinor bundle becomes a Krein space endowed with the indefinite inner product:
$$(\psi,\phi) = \int_\M \psi^* \J \phi\, d\mu_g.$$
 Actually, this operation is similar to a Wick rotation, but performed at an algebraic level.\\
 
In the special case of a $4$-dimensional Lorentzian manifold, with signature \mbox{$(-,+,+,+)$} and with local coordinates \mbox{$(x_0,x_1,x_2,x_3)$}, a fundamental symmetry is just given by \mbox{$J = i\gamma^0 = ic(dx^0)$}.\footnote{With this choice of signature, the Dirac matrix $\gamma^0$ is such that $\prt{\gamma^0}^2=-1$ and $\prt{\gamma^0}^*=-\gamma^0$, so $J = i\gamma^0$ respects the conditions of a fundamental symmetry. The other Dirac matrix respect $\prt{\gamma^i}^2=1$ and $\prt{\gamma^i}^*=\gamma^i$ for $i=1,2,3$.}\\

We have then the following result concerning the Dirac operator:
\begin{prop}\label{cpt1}
If there exists a spacelike reflection such that the Riemannian metric $g^r$ associated is complete, then the Dirac operator $D$ is essentially Krein-self-adjoint. In particular, if $\M$ is compact, then $D$ is always essentially Krein-self-adjoint.
\end{prop}

From all these properties, we can introduce the definition of a pseudo-Riemannian spectral triple:\\

\begin{defn}
A {\bf\index{spectral triple!pseudo-Riemannian} pseudo-Riemannian Spectral Triple} \mbox{$(\A,\H,D)$} is the data of:
\begin{itemize}
\item A Krein space $\mathcal{H}$
\item A pre-$C^*$-algebra $\mathcal{A}$ with a representation as bounded multiplicative operators on $\mathcal{H}$ and such that \mbox{$a^*=a^+$}
\item A Krein-self-adjoint operator $D$ densely defined on $\mathcal{H}$ such that all commutators $[D,a]$ is bounded for every $a\in\A$
\end{itemize}
\end{defn}

\vspace{1em}                                                      

In addition, it is natural to assume the existence of a fundamental symmetry $\J$ which commutes with all elements in $\A$. In this case, $\mathcal{A}$ becomes a subalgebra of $\B(\H)$ and the involution $a^*$ corresponds to the adjoint for the Hilbert space defined by $\J$.\\

Similarly to the Section \ref{ST} we can define an even condition and a real condition for a pseudo-Riemannian spectral triple. The condition on the boundedness of every commutator $[D,a]$ allows us to construct a similar noncommutative differential algebra generated by elements in the form \mbox{$a_0 [D,a_1] \cdot\cdot\cdot[D,a_p]$} as in the Section \ref{diffcalc}.\\

The $n^+$-summable condition is however a bit different. Indeed, in the commutative Riemannian case, the Dirac operator $D$ was elliptic, and this is not the case any more for Lorentzian geometry since its principal symbol satisfies the relation \mbox{$\sigma^D(\xi)^2 = c^2(\xi) = g(\xi,\xi)$} with $g$ the Lorentzian metric, and so it is no more invertible. To obtain a differential elliptic operator of order $1$, we define:
$$\Delta_\J = \prt{ [D]_\J^2 + 1}^{\frac 12}$$
with $[D]_\J^2 = \frac{DD^*+D^*D}{2}$ being the formally square of $D$ under the fundamental symmetry $\J$. This operator is elliptic of order $1$ since \mbox{$\sigma^{\Delta_\J}(\xi)^2  = g^r(\xi,\xi)$}, and is self-adjoint for the $\J$-product, so we can consider the rate of decay of its inverse modulo smoothing operation.

\begin{defn}\label{PRST}
A pseudo-Riemannian spectral triple is {\bf\index{spectral triple!pseudo-Riemannian!finitely summable} finitely summable} (or \mbox{$n^+$-summable}) if there exists a positive integer $n$ such that $a\,{\Delta_\J}\!\!\!\!^{-n} \in \L^{1+}$ for all $a\in\A$.
\end{defn}

One can check that this definition is independent of the choice of the fundamental symmetry $\J$. When the pseudo-Riemannian spectral triple is constructed from a compact  pseudo-Riemannian spin manifold, we have those additional results:

\begin{prop}\label{cpt2}
If $\M$ is a $n$-dimensional compact orientable pseudo-Riemannian spin manifold and if \mbox{$(\A,\H,D)$} is its pseudo-Riemannian spectral triple associated, then  \mbox{$(\A,\H,D)$} is $n^+$-summable, and for each smooth fundamental symmetry $\J$ commuting with al elements in $\A$ we have:
\begin{equation}
\int_\M f(x) \,d\mu_g =  c_n \tr_\omega(f{\Delta_\J}\!\!\!\!^{-n})\qquad \forall f\in\A
\end{equation}
with $c_n = 2^{n-[\frac n2]-1}\pi^{\frac n2}n\,\Gamma(\frac n2)$. Moreover, the signature $(p,q)$ can be recovered from the spectral data since the following formula holds:
\vspace{0.5em}                                                      
$$ \tr_\omega(fD^2{\Delta_\J}\!\!\!\!^{-n-1}) = {(-1)}^{q} \frac{n-2q}{n}\tr_\omega(f{\Delta_\J}\!\!\!\!^{-n})\qquad \forall f\in\A.$$
\end{prop}

\vspace{1em}                                                      

So we can see that Strohmaier's theory of pseudo-Riemannian noncommutative geometry allows us to technically translate the spectral characterization of the Dirac operator to pseudo-Riemannian manifolds. Some examples of construction of particular pseudo-Riemannian spectral triples can be found in \cite{Pas,Stro,Suij}.\\

\vspace{1em}                                                      

However, these adaptations are not sufficient to build a complete generalization of the theory. Among the remaining problems, we can notice the following ones:
\begin{itemize}
\vspace{0.5em}                                                      
\item Some elements (as the Propositions \ref{cpt1} and \ref{cpt2}) refer to compact pseudo-Riemannian manifolds. Once more this compactness condition is probably  present in order to avoid additional technical difficulties, but as we have already said compact manifolds must be discarded  in the Lorentzian case. So a complete study of pseudo-Riemannian spectral triples in the explicit case of non-compact manifolds would be of a great help.
\vspace{0.5em}                                                      
\item The general definition of pseudo-Riemannian spectral triples (Definition \ref{PRST}) does not include any explicit characterization on the signature of the space. In particular it is difficult to distinguish those referring to Lorentzian spectral triples to the other possibilities of signature, at least when the spectral triple is not directly constructed from a compact Lorentzian manifold (i.e.~when the Proposition \ref{cpt2} cannot be applied). Maybe a characterization of the dimension of the Riemannian subspaces could be set from some spectral informations extracted from a modification of $\Delta_\J$. In the Section \ref{secTLST}, we will suggest an extended definition of Lorentzian spectral triples with a new way to fix the signature.
\vspace{0.5em}                                                      
\item  Except under the conditions of the Proposition \ref{cpt2}, there is no guarantee that the noncommutative integral \mbox{$\int\!\!\!\!\!\!-\ a\, {\Delta_\J}\!\!\!\!^{-n} = \tr_\omega(a\,{\Delta_\J}\!\!\!\!^{-n})$} for $a\in\A$ is independent of the choice of the fundamental symmetry $\J$. Further conditions on the acceptable fundamental symmetries should probably be imposed in order to guarantee the existence of such integral (and maybe by this way guaranteeing a characterization of the signature for every pseudo-Riemannian spectral triple).
\vspace{0.5em}                                                      
\item Causal elements, and in particular the definition of a Lorentzian distance, are completely missing from the definition of pseudo-Riemannian spectral triples. Actually, a pseudo-Riemannian spectral triple is endowed with a Riemannian distance (using the Riemannian formula (\ref{rdistnc}) with the norm given by the $\J$-product) which just corresponds in the commutative case to the distance associated with the Riemannian metric $g^r$ coming from the Wick rotation. So in some way the causal informations are lost while working on the Hilbert space defined by the $\J$-product.
\end{itemize}

Nevertheless, the introduction of Krein spaces is a good way to guarantee the existence of Dirac operators for Lorentzian noncommutative geometries. New elements must be introduced in order to recovered the information on causality, but Krein spaces (or more general similar structures as we will present in the Section \ref{seccaus}) should represent an important piece in the generalization problem of noncommutative geometry.


\newpage

\section[Lorentzian distance function]{Establishing a global formulation for the Lorentzian distance function}\label{secdist}

In this section, we  develop some results we have presented in \cite{F2,F3} about the construction of a Lorentzian distance function in noncommutative geometry. Since the establishment of some of them requires a good technical background in Lorentzian geometry, we begin by a review of some known and less known definitions and properties on the subject. A complete introduction to Lorentzian geometry can be found in the general books \cite{Beem,ONeill} or in \cite{Hawk}.\\

\subsection{Lorentzian geometry}\label{Lorgeo}

From here, we will work with a  $n$-dimensional Lorentzian $C^\infty$-manifold $\M$ with metric $g$ and with signature \mbox{$(-,+,+,+,\dots)$}. The vectors of the tangent space $T_p\M$ at one point $p\in\M$ can be classified in $3$ groups following the sign of the metric.

\begin{defn}
A tangent vector $x\in\T_p\M$ is said to be {\bf\index{tangent vector!spacelike} spacelike} if \mbox{$g_p(x,x)>0$}, {\bf\index{tangent vector!timelike} timelike} if \mbox{$g_p(x,x)<0$} and {\bf\index{tangent vector!null} null} if \mbox{$g_p(x,x)=0$}. Vectors fields are spacelike, timelike or null if their vectors at each point are respectively spacelike, timelike or null.  A vector is {\bf\index{tangent vector!causal} causal} if timelike or null.\footnote{Timelike vectors are sometimes called chronological vectors, null vectors are sometimes called lightlike vectors and causal vectors are sometimes called nonspacelike vectors.}
\end{defn}

\begin{defn}
A Lorentzian manifold is {\bf\index{manifold!Lorentzian!time-orientable} time-orientable} if it admits a smooth timelike vector field $T$. Timelike vectors $x\in T_p\M$ are said to be {\bf\index{tangent vector!timelike!future directed} future directed} if \mbox{$g_p(x,T_p)<0$}  and {\bf\index{tangent vector!timelike!past directed} past directed} if \mbox{$g_p(x,T_p)>0$}. Future and past directed vectors are  two equivalent classes which can be switched, the choice of one of them as future directed is a choice of {\bf\index{time orientation} time orientation}.
\end{defn}

We will consider all Lorentzian manifolds to be time-orientable (this is actually the usual definition of spacetime).\\

We can notice that two timelike vectors $x,y\in T_p\M$ have the same orientation (i.e.~are both future directed or past directed) if and only if \mbox{$g_p(x,y)<0$}, and that the notion of orientation can be extended to non-zero null vectors. As a consequence, the set of timelike vectors with the same orientation forms a convex cone, called {\bf\index{time cone} time cone}.\\

Since the inner product $g(\,\cdot\,,\,\cdot\,)$ is not positive definite in Lorentzian spaces, there is no positive norm \mbox{$\norm{x} = \sqrt{g_p(x,x)}$} available. As a consequence, Cauchy--Schwarz and triangle inequalities are not valid any more, since their proofs rely on the positivity of the norm. Thankfully, some counterpart exists but only when both vectors are timelike.

\begin{prop}\label{cs}
Let $v$ and $w$ be two timelike vectors at some point~$p$, and let us denote by \mbox{$\scal{\,\cdot\,,\,\cdot\,} = g_p(\,\cdot\,,\,\cdot\,)$} the indefinite inner product at~$p$. Then we have the following properties:
\begin{itemize}
\item $\abs{\scal{v,w}} \geq \abs{\scal{v,v}}^{\frac 12} \,\abs{\scal{w,w}}^{\frac 12}$, with equality if and only if $v$ and $w$ are collinear \vspace{-0.2cm}\begin{flushright} ("wrong way" Cauchy--Schwarz  inequality)\end{flushright} \vspace{-0.2cm}
\item If $u$ and $v$ have the same orientation (in the same time cone), then \mbox{$\abs{\scal{v,v}}^{\frac 12} + \abs{\scal{w,w}}^{\frac 12} \leq \abs{\scal{v+w,v+w}}^{\frac 12}$}, with equality if and only if $v$ and $w$ are collinear \vspace{-0.2cm}\begin{flushright} ("wrong way" triangle inequality)\end{flushright} \vspace{-0.2cm}
\end{itemize}
\end{prop}

\begin{proof}
A technical proof can be found in \cite{ONeill}. For more diversity, we present here a personal proof based on geometrical considerations.

The usual Cauchy--Schwartz inequality in Riemannian spaces is a consequence of the invariance of the sign of  \mbox{$\norm{v+tw}^2 = \norm{v}^2 + 2t\scal{v,w} + t^2\norm{w}^2$} for $t$ going over $\setR$. In Lorentzian spaces, since timelike vectors are separated in two distinct cones, the combination $v+tw$ must go from one time cone (for $t\rightarrow\infty$ with the orientation of $w$ prevailing) to the other time cone (for $t\rightarrow-\infty$ with the orientation of $-w$ prevailing) without passing through the zero vector if $u$ and $v$ are not collinear, so $v+tw$ must be spacelike for some~$t$. As a consequence \mbox{$\scal{v+tw,v+tw} = \scal{v,v} + 2t\scal{v,w} + t^2 \scal{w,w}$} must change sign twice and its discriminant must be positive:
$$4 \scal{v,w}^2 - 4 \scal{v,v} \scal{w,w} \geq 0 \implies \scal{v,w}^2 \geq \prt{-\scal{v,v}} \,\prt{-\scal{w,w}}.$$
This discriminant is null if and only if $v+tw=0$ for some $t$.

The triangle inequality is a direct consequence, using the additional hypothesis that \mbox{$\scal{v,w} < 0$}:
\begin{eqnarray*}
 \prt{\abs{\scal{v,v}}^{\frac 12} + \abs{\scal{w,w}}^{\frac 12}}^2 &=& -\scal{v,v} + 2\abs{\scal{v,v}}^{\frac 12} \abs{\scal{w,w}}^{\frac 12} -\scal{w,w}\\
&\leq& -\scal{v,v} - 2\scal{v,w} -\scal{w,w}\\
&=& -\scal{v+w,v+w}.
\end{eqnarray*}\vspace{-1.3cm}\[\qedhere\]
\end{proof}

\begin{cor}\label{cscor}
The wrong way Cauchy--Schwarz and triangle inequalities still hold for causal vectors.
\end{cor}
\begin{proof}
This is trivial since any causal vector $v$ which is not timelike respects $\scal{v,v}=0$.
\end{proof}

\begin{defn}
A smooth curve or a piecewise smooth curve is said to be {\bf\index{curve!spacelike}\index{curve!timelike}\index{curve!null}\index{curve!causal} spacelike, timelike, null or causal} if its tangent vector, where defined, is respectively spacelike, timelike, null or causal everywhere. A timelike, null or causal curve can be future or past directed if its tangent vector is respectively future or past directed everywhere.
\end{defn}

\begin{defn}
If $p$ and $q$ are two points on $\M$, we have the possible following causal relations:
\begin{itemize}
\item $p \preceq q$ means that $p = q$ or that there is a future directed causal curve from~$p$ to~$q$
\item $p \precc q$ means that there is a future directed timelike curve from~$p$ to~$q$
\item $p \prec q$ means that $p \preceq q$ and $p\neq q$. 
\end{itemize}
These causal relations determine the following causal sets:
\begin{itemize}
\item $J^+(p) = \set{q\in\M : p \preceq q}$ is the {\bf\index{future!causal} causal future} of $p$
\item $J^-(p) = \set{q\in\M : q \preceq p}$ is the {\bf\index{past!causal} causal past} of $p$
\item $I^+(p) = \set{q\in\M : p \precc q}$ is the {\bf\index{future!chronological} chronological future} of $p$
\item $I^-(p) = \set{q\in\M : q \precc p}$ is the {\bf\index{past!chronological} chronological past} of $p$
\end{itemize}
Those sets can also be defined for every set $U\subset\M$ as for example \mbox{$J^+(U) = \cup_{p\in U}J^+(p) $}.
\end{defn}

We have said in the Section \ref{prob} that compact Lorentzian manifolds do not admit a well defined causal structure. Actually, one would like a Lorentzian manifold to be {\bf\index{manifold!Lorentzian!causal} causal}, which means that there exists no closed causal curve. It was proved in \cite{Gallo,Tipler} that compact Lorentzian manifolds cannot be causal. Moreover, one could require a stronger condition on causality. A Lorentzian manifold is {\bf\index{manifold!Lorentzian!strongly causal} strongly causal} if for each $p\in\M$ there exists a neighborhood of $p$ which is crossed at most once by every timelike curve.\\

Most of the time, we will impose an ever stronger condition on the global causal behaviour of the manifold:

\begin{defn}
A Lorentzian manifold $\M$ is {\bf\index{manifold!Lorentzian!globally hyperbolic} globally hyperbolic} if
\begin{itemize}
\item $\M$ is strongly causal
\item For every $p,q\in\M$, \mbox{$J^+(p) \cap J^-(q)$} is compact
\end{itemize}
\end{defn}

A first consequence of the global hyperbolicity of a Lorentzian manifold $\M$ is that two points $p\prec q$ can always be joined by at least one geodesic of maximal length, i.e.~there exists a causal geodesic from $p$ to $q$ whose length is greater than or equal to  that of any other future directed causal curve from $p$ to $q$ \cite{Beem}. Such geodesics are called {\bf\index{geodesic!maximal} maximal geodesics}.

\begin{defn}
A causal (or timelike) curve is said to be {\bf\index{curve!causal!inextensible} inextensible} if it admits no future or past endpoint, where a future (past) endpoint of a curve \mbox{$\gamma : I \rightarrow \M$} is an element $e\in\M$ such that for every neighborhood $U$ of $e$ there exists $t_e\in I$ such that \mbox{$\gamma(t)\in U \ \forall t> t_e$} ($\forall t<t_e$).
\end{defn}

\begin{defn}
A {\bf\index{Cauchy surface} Cauchy surface} is a subset $S \subset \M$ which every inextensible causal curve intersects exactly once.
\end{defn}

\begin{thm}\label{thgeroch}
A spacetime $\M$ is globally hyperbolic  if and only if it admits a smooth Cauchy surface $S$. In this case, $\M$ is diffeomorphic to the decomposition \mbox{$\setR \times S$}.
\end{thm}

The proof of this theorem is originated from R. Geroch \cite{Ger70}, but was only performed at a topological level (i.e.~with a homeomorphism to a non necessarily smooth Cauchy surface). The proof of the smoothness of the Cauchy surface is more recent and was given by A.N. Bernal and M. S\'anchez \cite{BS03,BS04,BS06}. In fact, the smoothness of the Cauchy surface comes from the existence of a Cauchy temporal function.

\begin{defn}\label{timefunction}
A {\bf\index{function!time}  time function} is a function strictly increasing along each future directed causal curve. A {\bf\index{function!time!Cauchy} Cauchy time function} is a  time function whose level sets are Cauchy surfaces. A {\bf\index{function!temporal}\index{function!temporal!Cauchy} (Cauchy) temporal function} is a smooth (Cauchy) time function with past-directed timelike gradient everywhere.
\end{defn}

\begin{cor}\label{corgeroch}
Let $\M$ be a globally hyperbolic spacetime. Then $\M$ admits a Cauchy temporal function.
\end{cor}

\begin{prop}\label{propJJ}
If $\M$ is globally hyperbolic and $S\subset\M$ is a Cauchy surface, then for all $p\in\M$, \mbox{$J^+(p) \cap J^-(S)$} and \mbox{$J^-(p) \cap J^+(S)$} are compact.
\end{prop}

This last result is quite intuitive and comes from the fact that the set of all continuous future directed causal curves  between $p$ and any point $q\in S$ is compact for a suitable topology. The proof is only technical and can be found in \cite{Wald}.

\begin{defn}\label{lordistdef}
The {\bf\index{distance!Lorentzian} Lorentzian distance}  on the spacetime $(\M,g)$ is the function $d : \M \times \M \rightarrow [0,+\infty) \cup \set{+\infty}$ defined by
$$
d(p,q) =
\begin{cases}
\quad \sup \left\{ l(\gamma) : 
\begin{array}{c}
\gamma \text{ future directed causal}\\
\text{piecewise smooth curve}\\
\text{with } \gamma(0)=p,\ \gamma(1)=q
\end{array}\right\} \ &\text{if } p \preceq q\\ 
\quad 0 &\text{if } p \npreceq q
\end{cases}
 $$
 where $l(\gamma) = \int \sqrt{ -g_{\gamma(t)}( \dot\gamma(t) ,  \dot\gamma(t) )}\ dt$ is the length of the curve.
\end{defn}

We can notice that, for globally hyperbolic spacetimes, the supremum is automatically obtained from a maximal geodesic joining $p$ and $q$. More generally, any future directed causal curve $\gamma$ from $p$ to $q$ is said to be maximal if \mbox{$l(\gamma)=d(p,q)$}.

\begin{prop}
The Lorentzian distance respects the following properties:
\begin{enumerate}
\item $d(p,q) \geq 0$ for all $p,q\in\M$
\item $d(p,q) > 0$ if and only if $p\precc q$
\item If $p\prec q \prec r$, then $d(p,r) \geq d(p,q) + d(q,r)$\begin{flushright} ("wrong way" triangle inequality)\end{flushright}
If $\M$ is causal, then we have the additional properties:
\item $d(p,p) = 0$
\item If $d(p,q) > 0$, then $d(q,p) = 0$
\end{enumerate}
\end{prop}

\begin{proof}
1. and 2. are immediate from the definition. 3. comes from the fact that the curves from $p$ to $r$ passing through $q$ are particular cases of the curves from $p$ to $r$. When the space is causal, $p\precc p$ cannot hold, and $q\precc p$, $p\precc q$ cannot hold together.
\end{proof}

\begin{cor}\label{Idpq}
For all $p\in\M$, \mbox{$I^+(p) = \set{q\in\M : d(p,q) > 0}$} and  \mbox{$I^-(p) = \set{q\in\M : d(q,p) > 0}$}.\\                                                     
\end{cor}

The Definition \ref{lordistdef} does not avoid the Lorentzian distance function to be infinite at some points. However, if we impose the condition of global hyperbolicity, we have the following behaviour for the distance function:

\begin{prop}
For a globally hyperbolic spacetime $\M$, the Lorentzian distance function is finite and continuous on $\M\times\M$. In particular, \mbox{$d_p(\,\cdot\,) = d(p,\,\cdot\,) : \M \rightarrow [0,+\infty)$} is a continuous function.
\end{prop}

The proof can be found in \cite{Beem,Hawk,ONeill}.\\

We need to introduce a barely known concept but which is a crucial element while dealing with distance functions, and which is called the cut locus. A large study of the cut locus for the Lorentzian distance function can be found in \cite{Beem}.\\

\vspace{1.5em}                                                      

Let us define the following elements:
\begin{itemize}
\item If $v\in T_p\M$ is a  future directed timelike vector, $c_v(t)$ denotes the unique geodesic in $\M$ such that $c_v(0)=p$ and $c_v'(0) = v$.
\item The {\bf\index{exponential map} exponential map} $\exp_p(v)$ of $v$ is given by \mbox{$\exp_p(v) = c_v(1)$} provided $c_v(1)$ is defined.
\item \mbox{$T_{-1}\M = \set{v\in T\M : g(v,v) = -1 \ \text{ and $v$ is future directed}}$} is the fiber bundle of unit future timelike vectors and $T_{-1}\M_{|p}$ is its fiber at $p$.
\item \mbox{$s : T_{-1}\M \rightarrow \setR \cup \set{\infty}$} is a function defined for $v\in T_{-1}\M_{|p}$ by
$$s(v) = \sup\set{t\geq 0 : d_p(c_v(t)) = t}.$$
\end{itemize}

We can remark that for \mbox{$0<t<s(v)$}, the geodesic \mbox{$c_v = [0,t] \rightarrow \M$} is the unique maximal geodesic (so the unique maximal timelike curve up to reparametrization) from $c_v(0)$ to $c_v(t)$. This leads to the fact that, if $s(v)<\infty$, the point $\exp(s(v)v)\in\M$ is the first point along the geodesic $c_v$ which can be reached by at least to different geodesics (or with infinitesimally neighboring geodesics which intersect at $p$ \cite{Hawk}). Such point is called the {\bf\index{cut point} cut point} of $c_v(0)$ along $c_v$.

\begin{defn}
If we define the following subset of future timelike tangent vectors:
$$\Gamma^+(p) = \set{s(v)v : v\in T_{-1}\M_{|p} \text{ and } 0 < s(v) < \infty}$$
then the {\bf\index{cut locus!timelike} future timelike cut locus} of $p$ is the set 
$$C_t^+(p) = \exp_p(\Gamma^+(p)) \subset I^+(p).$$
\end{defn}

To be complete, we must add the  {\bf\index{cut locus!null} future null cut locus} $C_N^+(p)$ of $p$ which is the set of points $\gamma(t_0)$ where $\gamma : [0,a) \rightarrow \M$ are future directed null  geodesics with $\gamma(0)=p$ and the extrema \mbox{$t_0 = \sup\set{ t \in [0,a) : d_p(\gamma(t)) = 0 }$} are considered only if $t_0< a$. The union \mbox{$C^+(p) = C_t^+(p) \cup C_N^+(p)$} is the {\bf\index{cut locus!nonspacelike} future nonspacelike cut locus} of $p$. The past cut locus can be defined dually.\\

In globally hyperbolic spacetimes, the cut locus corresponds to the domain of non-smoothness of the Lorentzian distance function. We have the following theorem from F. Erkeko\u{g}lu, E. Garc\'ia-R\'io and \mbox{D. N. Kupeli \cite{EGK}}, which we adapt here to globally hyperbolic spacetimes (the initial version is for strongly causal spacetimes, with no guarantee that the domain of smoothness extends to the whole chronological future):

\begin{thm}\label{graddp1}
If $\M$ is a globally hyperbolic spacetime, then $d_p$ is smooth on \mbox{$I^+(p) \setminus C^+(p)$} and satisfies the timelike eikonal equation \mbox{$g( \nabla d_p , \nabla d_p) = -1$}. Moreover, $\nabla d_p$ is a past directed vector field on \mbox{$I^+(p) \setminus C^+(p)$}.
\end{thm}

\begin{proof}
First, we can remark that, if we define the set 
$$ \tilde \Gamma^+(p) = \set{ tv : v\in T_{-1}\M_{|p} \text{ and } 0 < t < s(v)},$$
then  \mbox{$\exp_p( \tilde \Gamma^+(p) ) \cup C_t^+(p)= I^+(p)$} since by global hyperbolicity every point $q\in I^+(p)$ can be reached by a maximal geodesic beginning at $p$. The set \mbox{$I^+(p) \setminus C^+(p) = I^+(p) \setminus C_t^+(p) \subset \exp_p( \tilde \Gamma^+(p) )$} is an open set in globally hyperbolic spacetimes (see \cite{Beem}) and corresponds to the interior of $\exp_p( \tilde \Gamma^+(p) )$.

The exponential \mbox{$\exp_p : \text{int } \tilde \Gamma^+(p) \rightarrow I^+(p) \setminus C^+(p) $} is a diffeomorphism, so its inverse \mbox{$\exp_p^{-1} :  I^+(p) \setminus C^+(p)  \rightarrow \text{int } \tilde \Gamma^+(p) $} is smooth, and the distance function which can be expressed as \mbox{$d_p(q)= \prt{-g_p\prt{ \exp_p^{-1}(q) , \exp_p^{-1}(q)}}^{\frac 12}$} is also smooth on this open set.

Then let us take $v\in T_{-1}\M_{|p}$ and let \mbox{$c_v : [0,s(v)) \rightarrow I^+(p) \setminus C^+(p)$} be a unit future directed timelike geodesic. For \mbox{$t\in [0,s(v))$} we have the relation \mbox{$d_p(c_v(t)) = t$} which gives
\begin{equation}\label{gdpcv}
 \dv{}{t}\prt{d_p \circ c_v}(t) = 1 \implies g_{c_v(t)}\prt{ \,\prt{\nabla d_p \circ c_v}(t) \,,\,\dot c_v(t) \,} = 1.
 \end{equation}
Since $c_v(t)$ is orthogonal to the level sets of $d_p$ (this is a consequence of  Gauss Lemma saying that the exponential map is a radial isometry, whose Lorentzian version can be found in \cite{Beem}), $\dot c_v(t)$ must be collinear to $\prt{\nabla d_p \circ c_v}(t)$. We are in the equality case of the Cauchy--Schwartz inequality (Proposition \ref{cs}), and since \mbox{$g\prt{\dot c_v,\dot c_v} = -1$} we must have \mbox{$g\prt{\nabla d_p,\nabla d_p} = -1$}. Moreover, since $\dot c_v$ is future directed, the equation (\ref{gdpcv}) implies that $\nabla d_p$ must have the opposite orientation.
\end{proof}

\vspace{0.5em}                                                      

We have then the following theorem from V. Moretti:

\begin{thm}\label{graddp2}
Let $\M$ be a globally hyperbolic spacetime and take \mbox{$p\in\M$}, then the sets $J^+(p) \setminus I^+(p)$ and \mbox{$C^+(p)$} are closed with measure zero.
\end{thm}

\begin{proof}
The proof of the closeness can be found in \cite{Beem}. The proof of the measure zero of those sets is in \cite{Mor}.
\end{proof}

\begin{cor}\label{graddp3}
Let $\M$ be a globally hyperbolic spacetime and take \mbox{$p\in\M$}, then the Lorentzian distance function $d_p$ is a.e.~differentiable on $\M$, and \mbox{$g( \nabla d_p , \nabla d_p) = -1$} with $\nabla d_p$ past directed where $d_p>0$ and $\nabla d_p$ is defined.\\
\end{cor}

\vspace{1em}                                                      

\subsection{General form of a noncommutative Lorentzian distance function}\label{secgenform}

\vspace{0.5em}                                                      

As a first approach to generalize Connes' Riemannian distance function to Lorentzian spaces, we will try to propose a conceptual formulation of such a distance function.\\

Let us present a quick review of the concepts intervening for the Riemannian case. The Riemannian distance function
$$ d(p,q) = \sup\set{\abs{f(p) - f(q)} : f\in C(\M), \norm{[D,f]} \leq 1}$$
comes from the following equivalences: 
$$\norm{[D,f]} \leq 1\lequi \text{$f$ is Lipschitz with best constant 1}  $$
$$\lequi  \abs{f(x) - f(y)} \leq d(x,y)  \qquad \forall x,y $$
which leads to the following inequality:
$$ \sup\set{\abs{f(p) - f(q)} : f\in C(\M), \norm{[D,f]} \leq 1}\ \leq\ d(p,q).$$

Equality is obtained  by using the usual distance function \mbox{$f(z) = d_p(z)$} which is well Lipschitz thanks to the triangle inequality.\\

\vspace{0.5em}                                                      

In order to build a Lorentzian counterpart, we must adapt this function to the conditions for a Lorentzian distance:
\begin{itemize}
\item $d(x,x) = 0$
\item $d(x,y) > 0 \implies d(y,x) = 0$
\item If $x \prec y \prec z$, then $d(x,z) \geq d(x,y) + d(y,z)$\\                                                    
\end{itemize}

There are several ingredients we could replace or adapt:
$$ d(p,q) = \underbrace{\sup}_{\text{supremum?}}\set{\underbrace{\abs{f(p) - f(q)}}_{\text{difference?}} : \underbrace{\norm{[\underbrace{D}_\text{operator?},f]}}_{\text{Lipschitz condition?}} \underbrace{\leq 1}_{\text{constraint?}}} \cdot $$

\vspace{1em}                                                      

Our main idea is to create a conceptual Lorentzian counterpart whose form is as similar as possible to this formula. To make our work easier, we will add two hypotheses, which are the following ones:
\begin{itemize}
\item Work hypothesis 1: We suppose that there exists a way  to reproduce a Lipschitz-like condition in a Lorentzian framework.
\item Work hypothesis 2:  We will only consider the distance between two causally connected points.\\                                               
\end{itemize}

What does happen if we want to conserve the same Lipschitz condition as in the Riemannian case? We should have a function on this form, with $d_\sC$ some distance  on $\sC$:
\begin{equation}\label{hyp1}
 d(p,q) = \sup\set{d_\sC\!\prt{f(p), f(q)} : \text{\it $f$ Lipschitz}}
 \end{equation}
 where {\it $f$ Lipschitz} stands for a condition on $f$ implying that $f$ is Lipschitz with best Lipschitz constant $1$.\\

The main problem of this formulation is that the usual Lorentzian distance function $f(z) = d_q(z)$ is not a Lipschitz function any more, because of the wrong way triangle inequality:
\begin{eqnarray*}
x \prec y \prec z \ &\implies&\   d(x,z) \geq d(x,y) + d(y,z) \\
&\implies&\  \abs{d(x,z) - d(y,z)} \geq d(x,y)
\end{eqnarray*}
so the usual distance function cannot be considered among the class of {\it $f$ Lipschitz} functions.\\

Another concern is on the function $d_\sC$. Let $f$ be an arbitrary Lipschitz function with best Lipschitz constant 1 and assume that $d(x,y) > 0$, then if (\ref{hyp1}) is valid we have:
$$d_\sC\!\prt{f(x), f(y)} \leq  \ d(x,y).$$
By the conditions on Lorentzian distance, if we switch $x$ and $y$ we must obtain:
$$d_\sC\!\prt{f(y), f(x)} \leq  \ d(y,x) = 0 $$
$$\limplie d_\sC\!\prt{f(y), f(x)} = 0.$$

The last equation shows that, if we have a distance $d_\sC$ on $\setC$ such that $d_\sC\!\prt{a, b} = d_\sC\!\prt{b, a}$, then every Lorentzian distance of the form (\ref{hyp1}) will automatically be  a null function. This is an indication that, in the Lorentzian case, the old distance $d_\sC\!\prt{f(x), f(y)} = \abs{f(x)-f(y)}$ in $\setC$ should be replaced by a non-symmetric one.\\

The introduction of a non-symmetric distance between the values of the functions has further implications. An easy way would be to consider the sign of the difference $f(x)-f(y)$, but this only works if we restrict our set of functions to real-valued ones instead of complex-valued ones. Moreover in this case, we would like that all functions $f$  give by their range an information on the causal structure of the spacetime. So the set of functions should be restricted to causal functions, which are functions that do not decrease along every causal future-directed curve. We will see later that the causal functions are an important element in the generalization process.\\

Let us now try a different way to generalize the distance function, by introducing a {\it co-Lipschitz} condition instead of a {\it Lipschitz} condition.

\begin{defn}
A function $f$ is {\bf\index{function!co-Lipschitz} co-Lipschitz} if there exists a constant $M>0$ such that\footnote{The constant can also be defined as $\frac 1M$ instead of $M$, in order to have a correspondence with bi-Lipschitz maps respecting \mbox{$\frac 1M \; d(x,y) \leq \abs{f(x) - f(y)} \leq M\; d(x,y)\ \forall x,y$}, but this is not relevant here.}
$$  \abs{f(x) - f(y)} \geq M\; d(x,y)  \qquad \forall x,y.$$
The greatest $M$ is called the best co-Lipschitz constant.
\end{defn}

The main idea is to replace the supremum with an upper bound as in the Riemannian distance by an infimum with an equivalent lower bound.\\

We suggest the following general form:
\begin{equation}\label{genform}
 d(p,q) = \inf\set{\abs{f(p) - f(q)} : \text{\it $f$ co-Lipschitz}} \quad \text{for }p \prec q
 \end{equation}
 where {\it $f$ co-Lipschitz} stands for a condition on $f$ implying that $f$ is co-Lipschitz with best co-Lipschitz constant $1$.\\
 
Of course if we want to enlarge to any points $p$ and $q$ which do not necessarily respect the condition $p \prec q$, then the expression $\abs{f(p) - f(q)}$ must be replaced by a non-symmetric one.\\

There are several advantages of using a co-Lipschitz condition instead of a Lipschitz one in the Lorentzian case:
\begin{itemize}
\item We have trivially that
\begin{equation}\label{colin}
\inf \set{\abs{f(p) - f(q)} : \text{\it $f$ co-Lipschitz}} \ \geq\ d(p,q)
\end{equation}
\item Such kind of constraint -- with a lower bound -- seems more natural in space with indefinite inner product
\item The usual distance function $f(z) = d_p(z)$ respects the co-Lipschitz condition on the chronological future thanks to the wrong way triangle inequality: 
$$\abs{f(x)-f(y)} = \abs{d(p,x)-d(p,y)} \geq 1\  d(x,y)\quad\text{for } p\prec x \prec y$$
\end{itemize}

However, the usual distance function is not sufficient in order to obtain the equality case of (\ref{colin}) since the co-Lipschitz condition is not respected globally but only inside the chronological future.\\

We propose here an easy conceptual way to obtain the equality under some conditions. Let $p$ and $q$ be two causally related points with $p\prec q$. First, we will suppose the space to be globally hyperbolic. In this case we know that there exists at least one smooth Cauchy surface $C$ containing $p$. Our second hypothesis will be that this Cauchy surface $C$ can be chosen in such way that the distance between $p$ and $q$ corresponds to the distance between $C$ and $q$, i.e.~such that  \mbox{$d(C,q) = d(p,q)$} where $d(C,z) = \sup_{t\in C} \ d(t,z)$ denotes the distance between the Cauchy surface $C$ and a point $z$. Actually, this condition is similar to require that the maximal geodesic from $C$ to $q$ has its starting point at $p$.\\

Under those hypotheses,  we suggest the following equality function:
$$ f(z) = d(z,C) - d(C,z).$$

We can check that this function gives the requested equality:
$$\abs{f(p)-f(q)} = |\underbrace{d(p,C)}_{0} - \underbrace{d(C,p)}_{0} - \underbrace{d(q,C)}_{0} + \underbrace{d(C,q)}_{d(p,q)}| = d(p,q).$$

We must check that this function is co-Lipschitz with best co-Lipschitz constant 1. We will separate the proof in 3 different cases (with two being similar), depending on the localization of the Cauchy surface.

\begin{itemize}
\item Let $x,y$ be two points such that $x \preceq C\preceq y$, then:
$$\abs{f(x)-f(y)} = |\underbrace{d(x,C)}_{\geq0} - \underbrace{d(C,x)}_{0} - \underbrace{d(y,C)}_{0} + \underbrace{d(C,y)}_{\geq0}| $$
$$ \geq d(x,t) + d(t,y) = d(x,y)$$
 with $t$ being the point at the intersection between $C$ and the maximal geodesic from $x$ and $y$ (the existence of this geodesic is guaranteed by global hyperbolicity).
\item Let $x,y$ be two points such that  $C\preceq x\preceq y$ (the remaining case  $x\preceq y \preceq C$ is similar):
$$\abs{f(x)-f(y)} = |\underbrace{d(x,C)}_{0} - \underbrace{d(C,x)}_{\geq0} - \underbrace{d(y,C)}_{0} + \underbrace{d(C,y)}_{\geq0}| $$
$$ = \underbrace{d(C,y)}_{\geq d(t,y)} - \underbrace{d(C,x)}_{=d(t,x)} \geq d(t,y) - d(t,x) \geq d(x,y)$$
 with $t\in C$ being a point such that $d(t,x) = d(C,x)$ (whose existence is guaranteed by global hyperbolicity) and where we use at the end the wrong way triangle inequality \mbox{$d(t,x) + d(x,y) \leq d(t,y)$}.
\end{itemize}

So we have shown here that a function in the general form (\ref{genform}) could be a good candidate for a Lorentzian distance function. In the forthcoming sections, we will construct a technical formulation of such function which could be used as a good starting point for a noncommutative generalization. We will retain the very useful condition of global hyperbolicity, but we will propose a more complicated construction of the equality function in order to guarantee its existence for every globally hyperbolic spacetime.\\

\subsection{Global timelike eikonal inequality condition}\label{secglobalfunction}

The construction of the noncommutative Riemannian distance function in the Section \ref{riemdist} was done in three steps:
\begin{enumerate}
\item The construction of a path independent formulation of the Riemannian distance for commutative spaces based on a Lipschitz condition
\item The creation of an operatorial formulation of the Lipschitz condition
\item The generalization to noncommutative spaces
\end{enumerate}

In this section, we present a complete adaptation of the step 1 to the Lorentzian case, so we show that it is possible to construct a path independent formulation of the Lorentzian distance based on a co-Lipschitz condition. The obtained formula could be considered as a good starting point for further generalization, as a Lorentzian operatorial formulation which could be extended to noncommutative spaces.\\

Our strategy will be to follow a similar construction to the one in the Section \ref{riemdist}, and we will see that some Lorentzian elements as the wrong way Cauchy--Schwartz inequality will be of great help.\\

Let us consider a strongly causal time-orientable Lorentzian manifold $(\M,g)$ and two points $p$ and $q$ on it such that $p \precc q$. Let us choose an arbitrary future directed timelike piecewise smooth curve $\gamma : [0,1] \rightarrow \M$ with $\gamma(0) = p$ and $\gamma(1) = q$ (there must exist at least one such curve because $p \precc q$). Then, for each function $f \in C^\infty(\M,\setR)$:
\begin{eqnarray}
f(q) - f(p) &=& f(\gamma(1)) - f(\gamma(0)) \nonumber\\
&=& \int_0^1 \dv{}{t} f(\gamma(t)) \, dt \label{lordist0} \\
&=& \int_0^1 df_{\gamma(t)}(\dot\gamma(t)) \, dt \nonumber\\
&=& \int_0^1  g_{\gamma(t)}( \nabla f_{\gamma(t)},\dot\gamma(t)) \, dt. \nonumber
\end{eqnarray}

Since $\gamma$ is future directed, $\dot\gamma(t)$ is everywhere a future directed timelike vector. If we suppose that $\nabla f$ is everywhere timelike with constant orientation (we will take past directed), then $g_{\gamma(t)}( \nabla f_{\gamma(t)},\dot\gamma(t))$ is of constant sign (hence positive), and we have:
\begin{equation}\label{lordist1}
f(q) - f(p) = \int_0^1  g_{\gamma(t)}( \nabla f_{\gamma(t)},\dot\gamma(t)) \, dt = \int_0^1 \abs{ g_{\gamma(t)}( \nabla f_{\gamma(t)},\dot\gamma(t)) }\, dt.
\end{equation}

\vspace{1em}                                                      

We can remark here that we use for the Lorentzian case the set of real-valued smooth functions $C^\infty(\M,\setR)$ instead of the complex-valued ones $C^\infty(\M)$. This choice is dictated by the time-orientability of the manifold. Indeed, since $f$ is real-valued, if we suppose $\nabla f$ to be timelike then the orientation of $\nabla f$ is relevant. In the complex case there is no way to separate the timelike vectors into the two classes of future and past oriented vectors since \mbox{$g( \nabla f ,T)$} for any timelike vector field $T$ would be complex-valued.\\

Since both $\nabla f$ and $\dot\gamma(t)$ are everywhere timelike, we can apply the wrong way Cauchy--Schwartz inequality (Proposition \ref{cs}):
\begin{eqnarray}\label{lordist2}
& &\hspace{-1cm} \int_0^1  \abs{g_{\gamma(t)}( \nabla f_{\gamma(t)},\dot\gamma(t))} \, dt   \nonumber\\
& &\geq \int_0^1  \sqrt{-g_{\gamma(t)}(\nabla f_{\gamma(t)},\nabla f_{\gamma(t)})}\; \sqrt{-g_{\gamma(t)}(\dot\gamma(t),\dot\gamma(t))} \, dt  \nonumber\\
& &\geq  \inf\set{{\sqrt{-g(\nabla f,\nabla f)}}}\; l(\gamma).
\end{eqnarray}

Now we can see that this construction also holds if $\gamma$ is causal instead of just timelike. Indeed, $g( \nabla f,\dot\gamma(t))$ is non-negative for any future directed null vector $\dot\gamma(t)$ and the wrong way Cauchy--Schwartz inequality still holds by the Corollary \ref{cscor}. So we can extend our construction to the case $p \prec q$.

\begin{thm}\label{lorcontinuous}
If  $f \in C^\infty(\M,\setR)$ satisfies the following conditions:
\begin{itemize}
\item  $\sup\, g( \nabla f, \nabla f ) \leq -1$
\item $\nabla f$ is past directed
\end{itemize}
then for each $p,q\in\M$ such that $p \prec q$, we have 
$$f(q) - f(p) \geq d(p,q).$$
\end{thm}

\begin{proof}
$\sup\, g( \nabla f, \nabla f ) \leq -1$ implies that $\nabla f$ is a timelike vector field. From its past directed orientation, (\ref{lordist1}) holds. The inequality (\ref{lordist2}) becomes
$$f(q) - f(p)  = \int_0^1  \abs{g_{\gamma(t)}( \nabla f_{\gamma(t)},\dot\gamma(t))} \, dt \geq l(\gamma).$$
The result is obtained by taking the supremum on all future directed causal piecewise smooth curves from $p$ to $q$.
\end{proof}

This theorem shows that the counterpart of (\ref{rdist8}) for Lorentzian manifolds is:
\begin{equation}\label{ineqlor}
d(p,q) \leq \inf\set{ f(q)-f(p) :
\begin{array}{c}
 f \in C^\infty(\M,\setR),\\\
\sup g( \nabla f, \nabla f ) \leq -1,\\
 \ \nabla f \text{ is past directed}
\end{array}
}\cdot
\end{equation}

The condition $\sup\, g( \nabla f, \nabla f ) \leq -1$ is the global eikonal timelike inequality condition. It corresponds to the co-Lipschitz condition expressed in the Section \ref{secgenform}. This condition is totally independent of any path consideration, so it could give rise to a translation into an algebraic framework in order to be extended to noncommutative spaces. However, we do not know at this time how to create an operatorial formulation.\\

The condition on the orientation of $\nabla f$ for smooth functions $f$ is actually not a restrictive condition but a simple choice of time orientation, which can be easily translated into a noncommutative framework. Indeed, smooth functions obeying the condition $\sup\, g( \nabla f, \nabla f ) \leq -1$ can easily be separated in two distinct sets. In order to do that, we can just take two fixed points $p_0$ and $q_0$ (or their corresponding states in a noncommutative algebra) such that $\inf_f \abs{ f(q_0) -f(p_0) } > 0$ (which corresponds to the fact that $p_0$ and $q_0$ are causally related). Our two sets are
$$
\set{ f : \sup g( \nabla f, \nabla f ) \leq -1,\; f(q_0) - f(p_0) > 0}  
$$
and
$$
\set{ f : \sup g( \nabla f, \nabla f ) \leq -1,\; f(q_0) - f(p_0) < 0}.
$$
The choice of one of those sets corresponds to the choice of a particular time orientation.\\

In a similar way to the Riemannian case, we would like to obtain the equality case of (\ref{ineqlor}) by using in some way the usual Lorentzian distance function $d_p$. So once more we want to extend the set of functions $f \in C^\infty(\M,\setR)$ to a larger one which could include the functions $d_p$. In the Riemannian case, this was done by considering the set of bounded Lipschitz continuous functions. However, the Lorentzian distance $d_p$ does not belong to the set of Lipschitz continuous functions because of the wrong way triangle inequality, and also is not a bounded function any more since the manifold is not compact. Actually, it is clear that there exists no bounded smooth function respecting the condition $\sup g( \nabla f, \nabla f ) \leq -1$.\\

The problem is the following: since the construction (\ref{lordist0}) relies on the second fundamental theorem of calculus, we must enlarge the set of functions to a.e.~differentiable functions for which this theorem is still valid. However, this theorem is not valid for every a.e.~differentiable functions. The best known counterexample is the Cantor function, which is differentiable with null derivative except on a fractal set of measure zero. The set of a.e.~differentiable functions which still respect the second fundamental theorem of calculus is the set of absolute continuous functions, which are functions respecting the following condition, in the case of functions of one variable:
\begin{defn}
A function $f: I \rightarrow \setR$ is {\bf\index{function!absolutely continuous} absolutely continuous} on the closed bounded interval $I\subset\setR$ if for every $\epsilon>0$, there exists $\delta>0$ such that, for every finite collection of disjoint open intervals $(a_1b_1)\dots (a_k,b_k)$ of $I$ such that
$$\sum_{i=1}^k (b_i-a_i)  < \delta,$$
we have
$$\sum_{i=1}^k (f(b_i)-f(a_i))  < \epsilon.$$
\end{defn}

This definition of absolute continuous functions is actually equivalent to the fact that the second fundamental theorem of calculus holds. Informations about absolute continuous functions can be found in many books on real analysis, Lebesgue integration or measure theory, as e.g.~\cite{Gordon,Nielsen,Royden,Rudin}. See also \cite{Barcenas} for an alternative proof.\\

However we have to deal here with functions of several variables, even if the application of the second fundamental theorem is done on a single variable restriction (actually, the minimal condition would be to impose that the functions are absolutely continuous on causal paths). We can wonder  how to translate such kind of condition to an algebraic framework. One lead could be the use of Sobolev spaces, since Sobolev norms can be used to characterize absolute continuity \cite{Ziemer}. We can cite here the works of J. Maly \cite{Mal} and S. Hencl \cite{Hen} for a characterization of absolute continuity for functions of several variables  with the use of Sobolov spaces.\\

We can also suggest another way to solve this problem. Indeed, the second fundamental theorem does not need to hold to guarantee (\ref{ineqlor}) since the final formula results in an inequality. A weaker formulation would be sufficient if it guarantees the conservation of this inequality. This could be done by using the Lebesgue differentiation theorem:
\begin{thm}[Lebesgue]
Every function $f : [a,b] \rightarrow \setR$ of bounded variation\footnote{A fonction $f : [a,b] \rightarrow \setR$ is said to be of bounded variation if its total variation \mbox{$\sup_{P\in\P} \sum_{i=1}^{n_P-1} \abs{f(x_{i+1}) -f(x_i)}$} is finite, with $\P$ being the set of all partitions \mbox{$P=\set{x_1,\dots,x_{n_P}}$} of the interval $[a,b]$. In particular, every finite monotone function is of bounded variation.} is a.e.~continuous and a.e.~differentiable. Moreover, if the function $f$ is non-decreasing, then
$$\int_a^b f'(s)ds \;\leq\; f(b) - f(a).$$
\end{thm}

Let us suppose that we have a continuous function $f$ which is non-decreasing along every causal future-directed curve. If $\gamma : [0,1] \rightarrow \M$ is a future directed causal piecewise smooth curve, then \mbox{$f\circ \gamma : [0,1]  \rightarrow \setR$} is a one parameter non-decreasing continuous function of bounded variation on the interval $[0,1]$, so we have:
\begin{equation}\label{lordist5}
f(\gamma(1)) - f(\gamma(0)) \geq \int_0^1 \dv{}{t} f(\gamma(t)) \, dt  = \int_0^1  g_{\gamma(t)}( \nabla f_{\gamma(t)},\dot\gamma(t)) \, dt.
\end{equation}

\begin{defn}
The set of {\bf\index{function!causal} causal functions} $\C(\M)$ is the set of coutinuous functions $f \in C(\M,\setR)$ which are non-decreasing along every causal future-directed curve.
\end{defn}

One can easily check that the set of causal functions forms a convex cone.

\begin{thm}\label{lorinecausal}
We have the following inequality for every strongly causal  time-orientable Lorentzian manifold $\M$:
$$
d(p,q) \leq \inf\set{ f(q)-f(p) :  f \in \C(\M),\ \text{\rm ess } \sup g( \nabla f, \nabla f ) \leq -1}\cdot
$$
\end{thm}

\begin{proof}
We just have to use (\ref{lordist5}) instead of (\ref{lordist0}) in the proof of the Theorem \ref{lorcontinuous}. The condition on $\nabla f$ to be past directed is redundant since it is automatically implied by the two others.
\end{proof}

Now we can wonder whether the equality case of the Theorem \ref{lorinecausal} could be reached. We will prove in the Section \ref{equalitycase} that this is always the case if the space is globally hyperbolic.\\

We have to remark that, in the general case, the existence of functions in $\C(\M)$ respecting \mbox{$\text{\rm ess } \sup g( \nabla f, \nabla f ) \leq -1$} is not necessarily guaranteed, and so the equality case has no sense (except where the distance is infinite). However, when this set is non empty (as in globally hyperbolic spacetimes), defining \mbox{$d(p,q) = \inf\set{ f(q)-f(p)} $} for $p \prec q$ makes sense. Since we need a non-symmetric function, we can just define $d(p,q)$ to be zero when there exists a function $f\in\C(\M)$ respecting \mbox{$\text{\rm ess } \sup g( \nabla f, \nabla f ) \leq -1$} such that $f(q)-f(p)<0$. So we propose the following Lorentzian distance function:

\begin{thm}\label{maindefdist}
Let $(\M,g)$ be a globally hyperbolic spacetime and $d$ the Lorentzian distance function on $\M$, then:
\begin{equation}\label{maineqdist}
d(p,q) = \inf\set{  \langle f(q)-f(p) \rangle \ :\   f \in \C(\M),\ \text{\rm ess } \sup g( \nabla f, \nabla f ) \leq -1}
\end{equation}
where $\langle \alpha \rangle = \max\set{0,\alpha}$.
\end{thm}

The proof of this theorem will be done in the Section \ref{equalitycase} by showing that the equality case of the Theorem \ref{lorinecausal} could be reached (or at least indefinitely approached). This theorem is also valid if one replaces the set of causal functions $\C(\M)$ by any set of absolute continuous functions containing the Lorentzian distance $d_p$ and re-enters the condition on $\nabla f$ to be past directed. However we have not at this time the proof that the Lorentzian distance function is absolute continuous, so we will favor the formulation in terms of causal functions.\\

We can check that the formula (\ref{maineqdist}) respects all the properties of a Lorentzian distance function:

\begin{itemize}
\item $d(p,p) = 0$ is trivial
\item $d(p,q) \geq 0$ for all $p,q\in\M$ is also trivial
\item $d(p,q) > 0 \implies d(q,p) = 0$ comes from the use of the same functions $f$ and from the non-symmetricity of  \mbox{$ \langle f(q)-f(p) \rangle $}
\item The wrong way triangle inequality is valid since for all $p,q,r$ such that $d(p,q)>0$ and $d(q,r)>0$ we have
\begin{eqnarray*}
d(p,r) &=& \inf_{f\in\C(\M)}\set{   \langle f(r)-f(p) \rangle  \ : \text{\rm ess } \sup g( \nabla f, \nabla f ) \leq -1 }\\
&=& \inf_{f\in\C(\M)}\{   \langle f(r)-f(q) \rangle + \langle f(q)-f(p) \rangle  \ :\\
&&\qquad\qquad\qquad\qquad\qquad  \text{\rm ess } \sup g( \nabla f, \nabla f ) \leq -1 \}\\
&&\quad(\text{since }f(r)-f(q) > 0 \text{ and } f(q)-f(p) > 0)\\
&\geq& \inf_{f\in\C(\M)}\set{   \langle f(r)-f(q) \rangle  \ : \text{\rm ess } \sup g( \nabla f, \nabla f ) \leq -1 }\\
&&+\ \inf_{f\in\C(\M)}\set{   \langle f(q)-f(p) \rangle  \ : \text{\rm ess } \sup g( \nabla f, \nabla f ) \leq -1 }\\
&=& d(p,q) + d(q,r)
\end{eqnarray*}
\end{itemize}

Our formula (\ref{maineqdist}) is  on some points similar to the Lorentzian distance function proposed by V. Moretti in \cite{Mor}. Indeed, V. Moretti has proved the following distance formula:
\begin{equation}\label{morfunction}
d(p,q) = \inf\set{  \langle f(q)-f(p) \rangle \ :\   f \in \C(\bar I) ,\ p,q\in\bar I,\  \norm{[f,[f,\frac {\Delta}{2}]]^{-1}}_I \leq 1}
\end{equation}
where $\Delta = \nabla_\mu \nabla^\mu$ is the Laplace--Beltrami--d'Alembert operator, where $I$ belongs to a family of open, causally convex regions of $\M$ and where the set of causal functions $\C(\bar I)$ and the operator norm $\norm{\,\cdot\,}_I$ are restricted to the region $I$. This approach is one step further since there is an operatorial formulation for the condition \mbox{$\text{\rm ess } \sup\set{ g_x( \nabla f_x, \nabla f_x ) : x\in\bar I} \leq -1$}.\\

However, the formula (\ref{morfunction}) is entirely based on local conditions on the causal regions $I$. From our point of view, these local considerations should be avoided if we want to generalize such formula to noncommutative spaces, since the concept of locality is a priori absent in algebraic theories. Our formula (\ref{maineqdist}) has the advantage to be  based only on global constraints instead of local ones, so it should present a better starting point for further generalization.\\

The next step of generalization of the function (\ref{maineqdist}) should be to find an operatorial formulation of the condition \mbox{$\text{\rm ess } \sup g( \nabla f, \nabla f ) \leq -1$}. The operatorial formulation from (\ref{morfunction}) cannot be directly transposed to (\ref{maineqdist})  since it relies on the existence of some local spaces \mbox{$L^2(\bar I)$} on which the functions in $\C(\bar I)$ act as multiplicative (bounded) operators. While dealing with the set of global causal functions $\C(\M)$, those functions do not belong to $L^2(\M)$ since there are unbounded for most of them and only locally integrable, so the usual norm is not available. The way to define a norm on the set of causal functions $\C(\M)$ will be approached in the Section \ref{seccaus}. \\

\subsection{Construction of the equality case}\label{equalitycase}

To conclude our section on the Lorentzian distance function, we must give a proof of the equality case of the Theorem \ref{maindefdist}. We will work with a globally hyperbolic spacetime $(\M,g)$ (necessarily time-orientable).\\

In the Section \ref{secgenform} we have suggested that, for $p \prec q$, the equality could be reached by a function of the form $ f(z) = d(z,C) - d(C,z)$ where  $C$ is a  Cauchy surface containing $p$ such that  \mbox{$d(C,q) = d(p,q)$} where $d(C,z) = \sup_{t\in C} \ d(t,z)$ denotes the distance between the Cauchy surface $C$ and a point $z$. For the case where $p$ and $q$ are not causally connected, $C$ should just be a Cauchy surface containing both $p$ and $q$.  \\

However, if such surfaces can easily be built for simple Lorentzian spaces (as e.g.~Minkowski space) we do not have the guarantee of their existence for every globally hyperbolic spacetimes. Moreover, the gradient of $ f(z) = d(z,C) - d(C,z)$ will be ill-defined on the cut locus of $C$ and we have no information about the measure of this cut locus (see \cite{EGK} for the study of the behaviour of the distance function related to a Cauchy surface and for the definition of the cut locus relative to such surface).\\

So the equality case we present here will be entirely based on the usual distance function $d_p$ which we already know some useful informations about its gradient (Theorem \ref{graddp1}, Theorem \ref{graddp2} and Corollary~\ref{graddp3}). We will only work with the set $\C(\M)$ but the whole construction will be valid for any set of absolute continuous functions containing the Lorentzian distance $d_p$, if one can prove the belonging of the function $d_p$ to such set. The belonging of the function $d_p$ to the set of causal functions $\C(\M)$ is obvious.\\

The proof will be divided in three cases, corresponding to the three following possible relations between the points $p$ and $q$:
\begin{itemize}\itemsep=0pt
\item $p \prec q$, where $p$ and $q$ are causally related, with $q$ in the future of $p$
\item $q \succ p$, where $p$ and $q$ are causally related, with $q$ in the past of $p$
\item $p \nprec q$ and $q \nprec p$, where $p$ and $q$ are not causally related
\end{itemize}

\vspace{1em}                                                      

In particular, with these three results, the proof of the Theorem \ref{maindefdist} will be completed: 
\begin{itemize}\itemsep=0pt
\item  (Proposition \ref{equality}) If $p \prec q$, then there exists a sequence of functions \mbox{$f_\epsilon \in\C(\M)$} ($\epsilon > 0$) respecting \mbox{$\text{\rm ess } \sup g( \nabla f, \nabla f ) \leq -1$} and such that \mbox{$d(p,q) \leq f_\epsilon(q) -f_\epsilon(p) < d(p,q) + \epsilon$}.
\item (Corollary \ref{equalityrev}) If $p \succ q$, then there exists a function $f \in\C(\M)$ respecting \mbox{$\text{\rm ess } \sup g( \nabla f, \nabla f ) \leq -1$} and such that \mbox{$f(q)-f(p) \leq 0$}.
\item (Proposition \ref{zero}) If $p \nprec q$ and $q \nprec p$, then there exists a sequence of functions $f_\epsilon \in\C(\M)$ ($\epsilon > 0$) respecting \mbox{$\text{\rm ess } \sup g( \nabla f, \nabla f ) \leq -1$} and such that $\abs{ f_\epsilon(q) -f_\epsilon(p) } < \epsilon$.
\end{itemize}

\vspace{1em}                                                      

The basic idea of our construction is to create some functions with a different behaviour in two regions:
\begin{itemize}\itemsep=0pt
\item The first region is a region containing the points $p$ and $q$ where these functions correspond to simple suitable distance functions.
\item The second region is the remaining of the manifold where these functions are locally finite sums of distance functions in order to have the eikonal condition respected in the whole space.
\end{itemize}

We begin the proof by an important lemma, which will give the way to make the construction of the second region.\\


\begin{lemma}\label{lemmaprincipal}
Let $S$ be a smooth spacelike Cauchy surface and two points $q \in J^+(S)$ $($resp.\ $q \in J^-(S))$ and $q' \in I^+(q)$ $($resp.\ $q' \in I^-(q))$. There exists a~function $f \in \C(\M)$ $($resp.\ $-f \in \C(\M))$ such that:
\begin{itemize}
\item $f \geq 0 $
\item $g( \nabla f, \nabla f ) \leq -1$ and $\nabla f$ is past directed (resp.\ future directed) where $f > 0$, except on a set of measure zero 
\item $f > 0$ on $J^+(S) \setminus I^-(q')$ $($resp.\ on $J^-(S) \setminus I^+(q'))$
\item $f = 0$ on $J^-(q)$ $($resp.\ on $J^+(q))$\\
\end{itemize}
\end{lemma}

\begin{proof}
To make the understanding of the proof easier, we will make some illustrations of the form of a $2$-dimensional slice of a flat manifold, as in the Figure \ref{L1}. We will leave the general $n$-dimensional curved case to the imagination of the reader.
\begin{figure}[!ht]
\vspace{1em}                                                      
\begin{center}
\includegraphics[width=7cm]{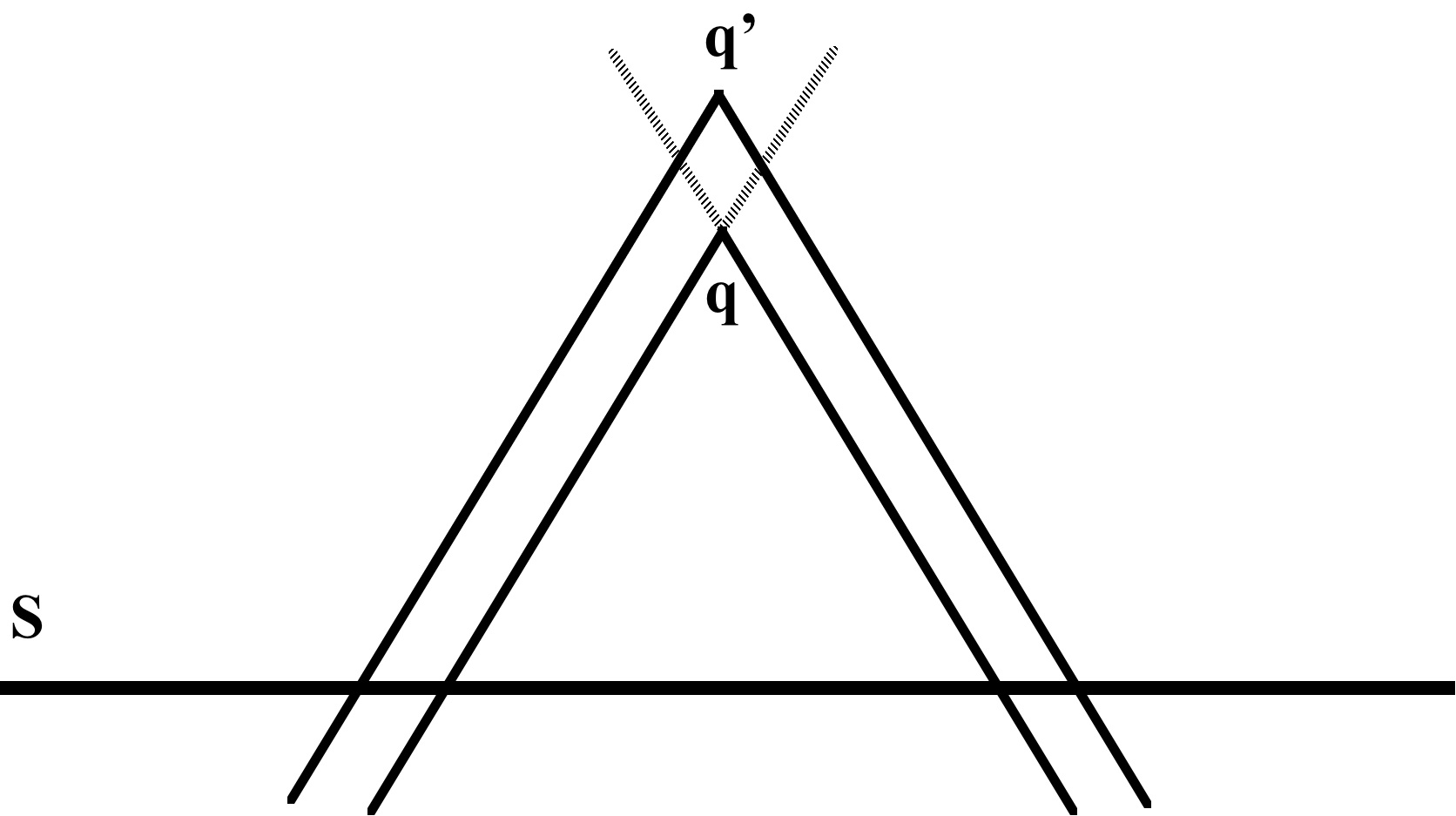}
\end{center}
\caption{}\label{L1}
\vspace{1em}                                                      
\end{figure}

At first, let us notice that $J^-(q) \subset  I^-(q')$. Indeed, $q \precc q'$, so $I^-(q) \subset  I^-(q')$. We have that $J^-(q) = \overline{I^-(q) }\subset  I^-(q')$ by considering that $I^-(q') = \set{ p\in \M : d(p,q') > 0}$ (Corollary \ref{Idpq})  and that for every point $ z\in J^-(q)$, $d(z,q') \geq d(z,q) + d(q,q') > 0$.

\vspace{1em}                                                      

On the smooth spacelike surface $S$ we can consider the Riemannian metric $g_R$ which is the restriction of the global metric $g$ to $S$, with a Riemannian distance $d_R$ and a topology associated. If we take a point $p \in I^-(S)$, the intersection $I^+(p) \cap S$ is an open subset of $S$ (in the topology of $S$) with a finite diameter $d_R(I^+(p) \cap S) < \infty $  because \mbox{$J^+(p) \cap J^-(S)$} is compact (Proposition \ref{propJJ}). We will work with points $p$ closed to the surface $S$ such that $d_R(I^+(p) \cap S)$ is small. Let us define (Figure \ref{L2}):
$$
P = \set{ p \in I^-(S) \setminus J^-(q) : d_R(I^+(p) \cap S) < 1}.
$$
\begin{figure}[!ht]
\vspace{1em}                                                      
\begin{center}
\includegraphics[width=7cm]{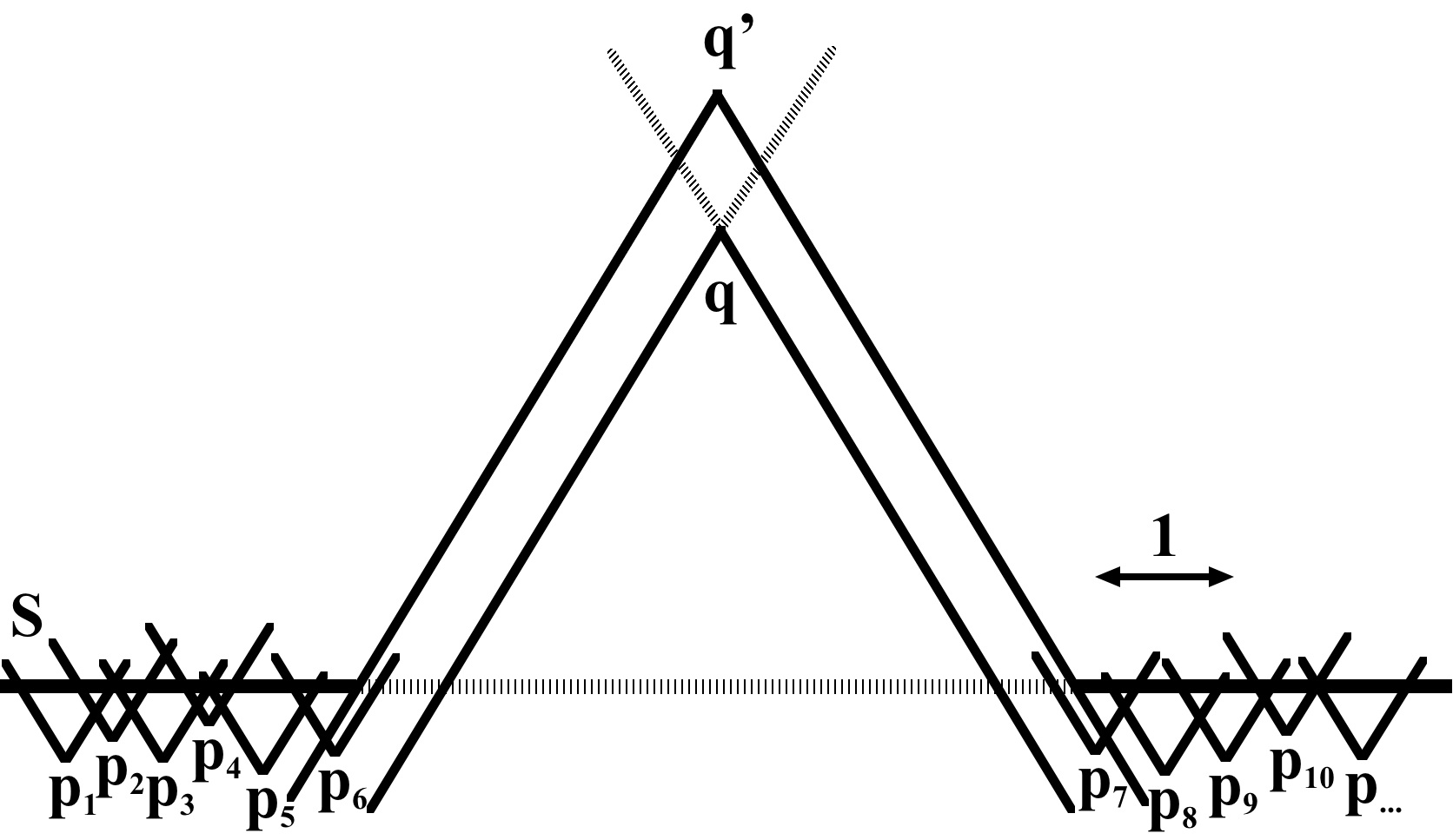}
\end{center}
\caption{}\label{L2}
\vspace{1em}                                                      
\end{figure}

Then the collection
$$
W = \set{ I^+(p) \cap S : p \in P}
$$
is an open covering of the closed set $S \setminus I^-(q')$ (because we have \mbox{$I^-(q') \supset J^-(q)$}).

We will show that there exists a locally finite subcovering of $W$ by a method similar to one used in \cite{BS03}. Let us fix $s \in S$ and let us consider the open and closed balls $B_s(r)$ and $\bar B_s(r)$ in~$S$ of center $s$ and radius $r$ for the distance~$d_R$. The following subsets are compact in $S$:
$$
S_n = \bar B_s(n) \setminus \parenthese{ B_s(n-1) \cup I^-(q')} \quad n \in \setN.
$$
Their union is a compact covering of $S$, except for the intersection of $S$ and $I^-(q')$ (Figure \ref{L3}):
$$
\bigcup_{n\in \setN} S_n = S \setminus I^-(q').
$$
\begin{figure}[!ht]
\vspace{1em}                                                      
\begin{center}
\includegraphics[width=5cm]{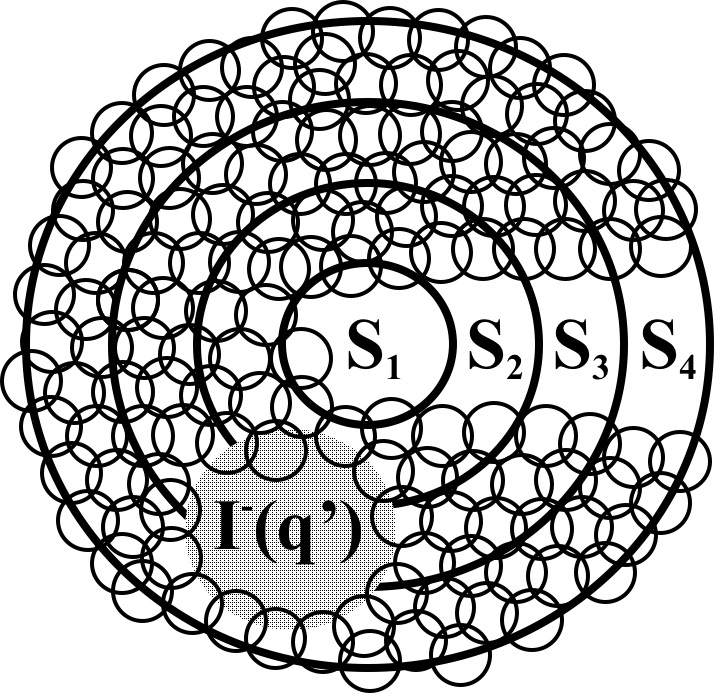}
\end{center}
\caption{}\label{L3}
\vspace{1em}                                                      
\end{figure}

For each $S_n$ we can find a finite subset $\set{W_{1n}, \dots, W_{k_n n}} \subset W$ which covers $S_n$ since the subsets $S_n$ are compact. Then
$$
W' = \set{W_{kn} : n\in \setN , k = 1, \dots, k_n}
$$
 is a locally finite subcovering of $S \setminus I^-(q')$ because every $W_{kn}$ has diameter smaller than $1$ (i.e.~every $W_{kn}$ intersects at most two subsets $S_n$). This shows that we can extract a subset of points $P' \subset P$ such that $W' = \set{ I^+(p) \cap S : p \in P'}$ is locally finite.

From that, we can show that $\set{I^+(p) : p \in P'}$ is a locally finite covering of  $J^+(S) \setminus I^-(q')$. Indeed, let us take a point $z$ in $J^+(S) \setminus I^-(q')$. As a consequence of the Proposition \ref{propJJ}, the set $I^-(z) \cap S$ is an open set of finite diameter in $S$, and so it intersects only a finite number of $W_{kn}$ (and must intersect at least one). Hence $I^-(z)$ contains a non empty but finite subset of $P'$. The same reasoning can be done for any small neighbourhood of $z$.

\vspace{2em}                                                      

Now we can construct the non-negative following function:\vspace{1em}                                                      
$$
f(z) = \sum_{p \in P'} d(p,z) = \sum_{p \in P'} d_p(z).\vspace{1em}                                                      
$$

This function is well defined because the sum is pointwise finite and it is continuous by continuity of the distance function. $f$ is null on $J^-(q)$ because no $p \in P'$ belongs to $J^-(q)$ and it is positive on \mbox{$J^+(S) \setminus I^-(q')$} because every point in $J^+(S) \setminus I^-(q')$ is  inside the chronological future of at least one $p \in P'$. For every $z \in \M$ we can find a neighbourhood where $f$ is  a finite sum $\sum d_p$ of distance functions, so $f\in\C(\M)$.

The following is a direct consequence of the Corollary~\ref{graddp3}. Because $\nabla d_p$ is well defined except on a set of measure zero, by countability of the measure, the locally finite sum $\nabla f = \sum \nabla d_p$ is well defined except on a~set of measure zero. Then where $\nabla f$ is well defined and $f$ positive, we have that $\nabla f$ is timelike past directed (because it is the sum of null or timelike past directed vectors), and\vspace{1em}                                                      
$$
g( \nabla f, \nabla f ) =  \sum_{p\in P'} g( \nabla d_p, \nabla d_p ) +  \sum_{p,p'\in P' \atop p \neq p'} g( \nabla d_p, \nabla d_{p'} ) \leq -1,\vspace{1em}                                                      
$$
where the first sum contains terms equal to $-1$ or $0$, with at least one term equal to $-1$, and the second sum contains terms negative or null because all $\nabla d_p$ are null or timelike past directed (so they have the same orientation).

In the reverse case where $q \in J^-(S)$ and $q' \in I^-(q)$, we can do an identical proof by reversing future and past sets and by taking \mbox{$f(z) = \sum_{p \in P'} d(z,p)$} as a function with null or timelike future oriented gradient which is non-increasing along every causal future-directed curve.\\
\end{proof}

\vspace{\fill}

\newpage\vspace{1em}                                                      

\begin{cor}\label{corollaryprincipal}
Let $S$ be a smooth spacelike Cauchy surface and four points $q_1,q_2 \in J^+(S)$ and $q_1' \in I^+(q_1)$, $q_2' \in I^+(q_2)$ $($resp.\ $q_1,q_2 \in J^-(S)$, $q_1' \in I^-(q_1)$, $q_2' \in I^-(q_2))$. There exists a~function $f \in \C(\M)$ $($resp. $-f \in \C(\M))$ such that:
\begin{itemize}
\item $f \geq 0 $
\item $g( \nabla f, \nabla f ) \leq -1$ and $\nabla f$ is past directed $($resp.\ future directed$)$ where $f > 0$, except on a~set of measure zero
\item $f > 0$ on $J^+(S) \setminus \parenthese{I^-(q_1') \cup I^-(q_2')}$ \vspace{-0.2cm}\begin{flushright}$($resp.\ on \mbox{$J^-(S) \setminus \parenthese{I^+(q_1') \cup I^+(q_2')})$}\end{flushright} \vspace{-0.2cm}
\item $f = 0$ on $J^-(q_1) \cup J^-(q_2)$ $($resp.\ on $J^+(q_1) \cup J^+(q_2))$\\
\end{itemize}
\end{cor}

\begin{proof}
(Figure \ref{L4}) The proof is identical to the Lemma \ref{lemmaprincipal} except that we start with the set 
$$
P = \set{ p \in I^-(S) \setminus \parenthese{J^-(q_1) \cup J^-(q_2)} : d_R(I^+(p) \cap S) < 1}
$$
to create a locally finite covering of $J^+(S) \setminus \parenthese{ I^-(q_1') \cup I^-(q_2') }$.\\
\begin{figure}[!ht]
\begin{center}
\includegraphics[width=9cm]{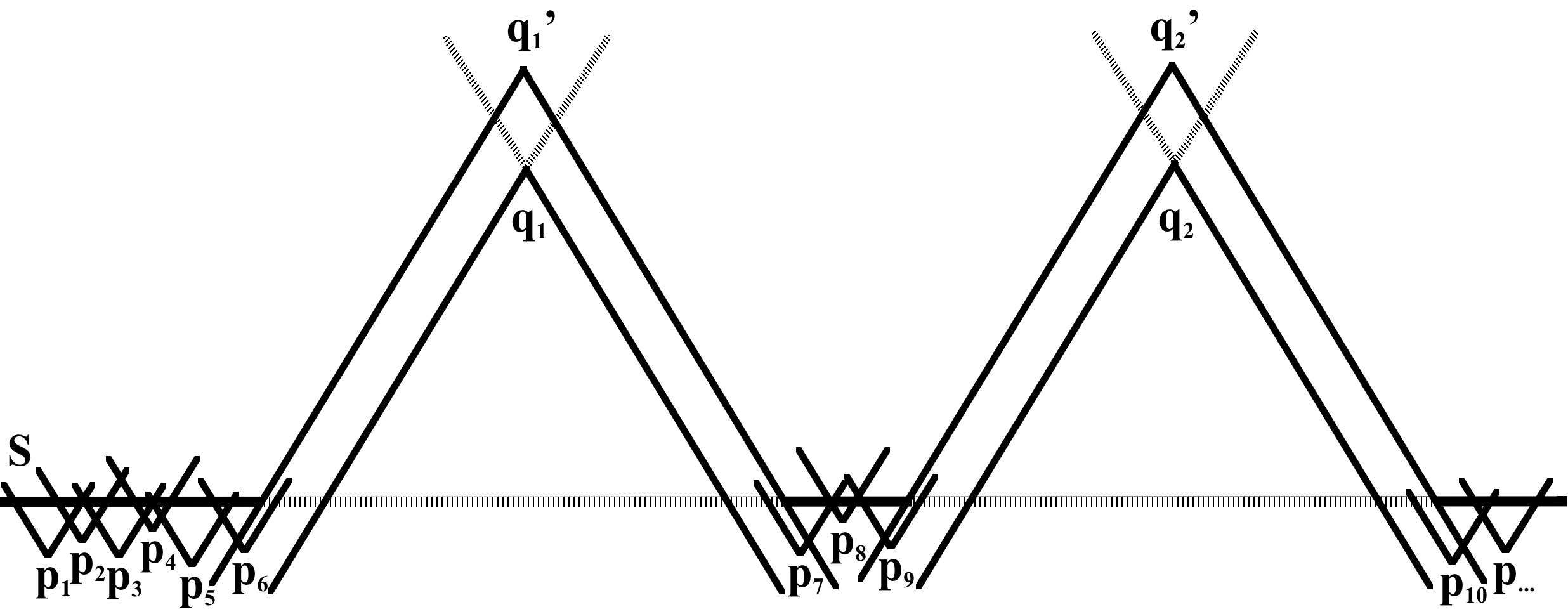}
\end{center}
\caption{}\label{L4}
\end{figure}
\end{proof}

\vspace{1em}                                                      

Now by the use of the Lemma \ref{lemmaprincipal} we can construct our different functions respecting the eikonal inequality condition on the whole manifold $\M$.\\


\begin{prop}\label{equality}
If $p \prec q$, then $\forall\, \epsilon > 0$ there exists a function \mbox{$f\in \C(\M)$} such that:
\begin{itemize}
\item $\text{\rm ess} \sup g( \nabla f, \nabla f ) \leq -1$
\item $\nabla f$ is past directed
\item $f(q)-f(p) \geq 0$
\item $\abs{(f(q)-f(p)) - d(p,q)} < \epsilon$
\end{itemize}
\end{prop}

\begin{proof}
Let us choose a smooth spacelike Cauchy surface $S$ containing $q$, whose  existence is guaranteed by the Theorem \ref{thgeroch}.  Then let us choose two free points~$p'$ and~$q'$ such that $p' \in I^-(p)$ and $q' \in I^+(q)$.  Because $q \in I^+(p')$ and $q\in S$ we can choose $q'$ close to $q$ such that $I^-(q') \cap J^+(S) \subset I^+(p')$ (Figure \ref{L5}).
\begin{figure}[!ht]
\begin{center}
\includegraphics[width=7cm]{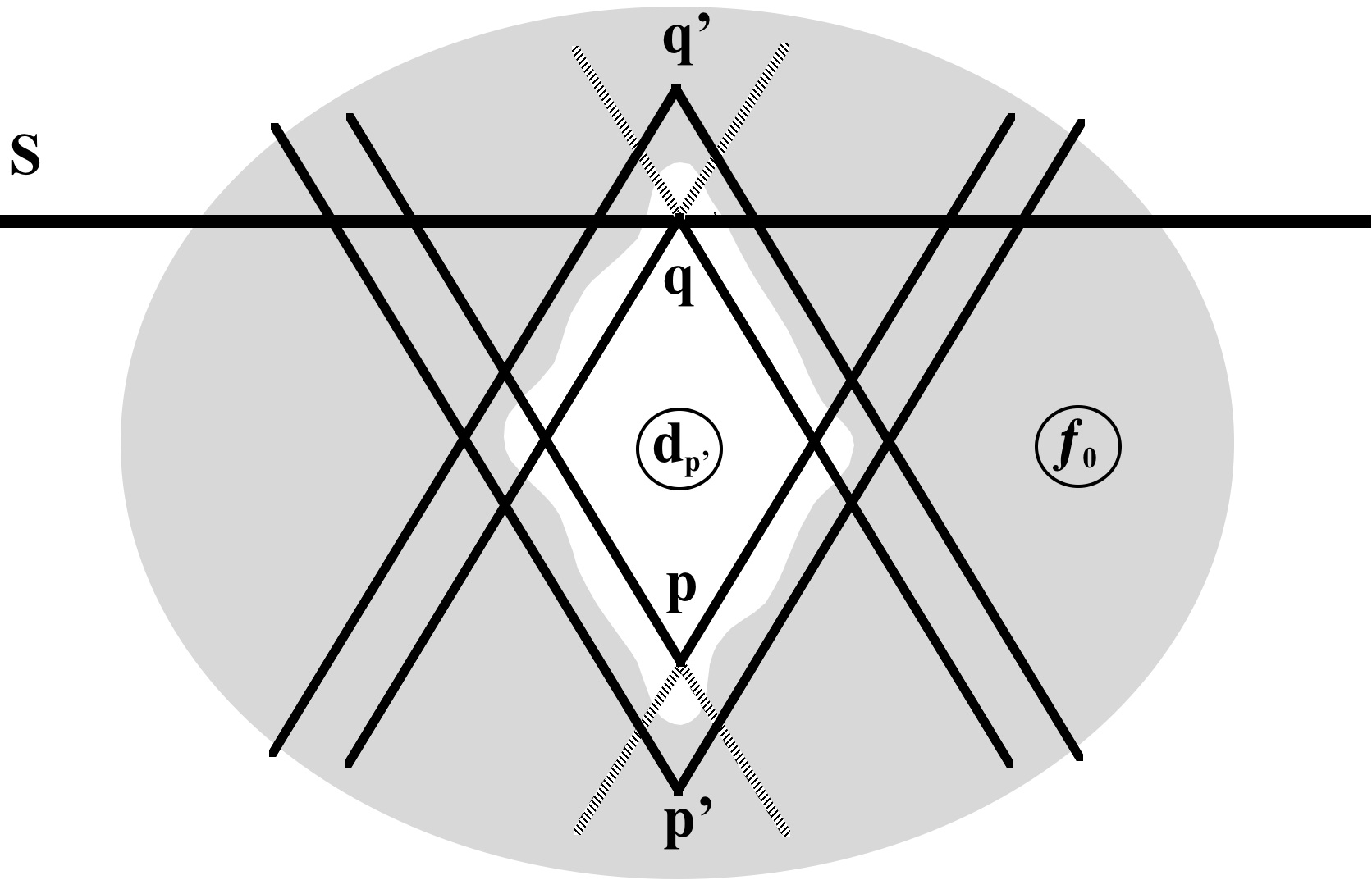}
\end{center}
\caption{}\label{L5}
\end{figure}

We can apply the Lemma \ref{lemmaprincipal} to $S$, $q$ and $q'$ to get a function $f_1$ respecting the properties:
\begin{itemize}\itemsep=0pt
\item $f_1 \in \C(\M)$
\item $f_1 \geq 0 $
\item $g( \nabla f_1, \nabla f_1 ) \leq -1$ and $\nabla f_1$ is past directed where $f_1 > 0$, except on a set of measure zero
\item $f _1> 0$ on $J^+(S) \setminus I^-(q')$
\item $f_1 = 0$ on $J^-(q)$
\end{itemize}
and then we can apply it to $S$, $p$ and $p'$ to get a function $f_2$ respecting the properties:
\vspace{1em}                                                      
\begin{itemize}\itemsep=2pt
\item $-f_2 \in \C(\M)$
\item $f_2 \geq 0 $
\item $g( \nabla f_2, \nabla f_2 ) \leq -1$ and $\nabla f_2$ is future directed where $f_2> 0$, except on a set of measure zero
\item $f_2 > 0$ on $J^-(S) \setminus I^+(p')$
\item $f_2 = 0$ on $J^+(p)$
\end{itemize}

\vspace{1.5em}                                                      

Then the function:
$$
f_0 = f_1 - f_2
$$
has the following properties:
\vspace{1em}                                                      
\begin{itemize}\itemsep=2pt
\item $f_0\in\C(\M)$
\item $f_0 = 0$ on the compact $J^-(q) \cap J^+(p)$
\item The support of $\nabla f_0$ includes the set 
$$\parenthese{J^+(S) \setminus I^-(q')} \cup  \parenthese{J^-(S) \setminus I^+(p')}$$
\item $g( \nabla f_0, \nabla f_0) \leq -1$ and $\nabla f_0$ is past directed on its support, except on a set of measure zero
\end{itemize}

\vspace{1em}                                                      

The last assertion comes from
$$
g( \nabla f_0, \nabla f_0 ) =  g( \nabla f_1, \nabla f_1 ) + g( \nabla f_2, \nabla f_2 ) -  2 g( \nabla f_1, \nabla f_2 ) \leq -1,
$$
where the first term is equal to $-1$ on the support of $f_1$ and is non-positive elsewhere, the second term is equal to~$-1$ on the support of $f_2$ and is non-positive elsewhere, and the last term is non-positive because $\nabla f_1$ and $\nabla f_2$ have not the same orientation.

\vspace{1em}                                                      

We can now define the function:
$$
f = f_0 + d_{p'}  \quad \Lequi \quad f(z) = f_0(z) + d(p',z)
$$
which has the following properties:
\vspace{1em}                                                      
\begin{itemize}\itemsep=2pt
\item $f\in\C(\M)$
\item $f(p)= d(p',p)$
\item $f(q)= d(p',q)$
\item The support of $\nabla f$ is $\M$
\item $g( \nabla f, \nabla f) \leq -1$ and $\nabla f$ is past directed on $\M$, except on a set of measure zero
\end{itemize}

\vspace{1em}                                                      

To verify that the support of $\nabla f$ is $\M$, we can see that
$$
\M \setminus J^+(p') \subset \parenthese{J^+(S) \setminus I^-(q')} \cup  \parenthese{J^-(S) \setminus I^+(p')}
$$
because we have the fact that $I^-(q') \cap J^+(S) \subset I^+(p')$ and that the support of $d_{p'}$ is $J^+(p')$.

\vspace{1em}                                                      

So now we have a function $f$ such that
$$f(q)-f(p) = d(p',q) - d(p',p) \geq d(p,q) \geq 0$$
 by the inverse triangle inequality.

\vspace{1em}                                                      

Let us set the function
$$
\alpha(p') = \parenthese{f(q) - f(p)} - d(p,q) = \parenthese{d(p',q) - d(p',p)} - d(p,q).
$$

We can remember that the points~$p'$ and~$q'$ were chosen freely under the conditions $p' \in I^-(p)$, $q' \in I^+(q)$ and \mbox{$I^-(q') \cap J^+(S) \subset I^+(p')$}. So the point $p'$ can be chosen arbitrarily closed to $p$ (and $q'$  arbitrarily closed to $q$).   $\alpha$~is a~continuous function because the distance function is continuous, and $\alpha(p) = 0$. Hence it is always possible to choose the initial point $p'$ such that $\abs{\alpha(p')} < \epsilon$, which implies $\abs{(f(q)-f(p)) - d(p,q)} < \epsilon$.\\
 \end{proof}

\vspace{\fill}


\newpage\begin{cor}\label{equalityrev}
If $p \succ q$, then there exists a function $f\in \C(\M)$ such that:
\begin{itemize}
\item $\text{\rm ess} \sup g( \nabla f, \nabla f ) \leq -1$
\item $\nabla f$ is past directed
\item $f(q)-f(p) \leq 0$
\end{itemize}
In particular, $\langle f(q)-f(p) \rangle = 0$.
\end{cor}

\vspace{0.5em}                                                      

\begin{proof}
This is trivial by switching $p$ and $q$ in the Proposition \ref{equality}.\qedhere\\
\end{proof}

\vspace{1em}                                                      


\begin{prop}\label{zero}
If $p \nprec q$ and $q \nprec p$, then $\forall \, \epsilon > 0$ there exists a function $f\in \C(\M)$ such that:
\begin{itemize}
\item $\text{\rm ess} \sup g( \nabla f, \nabla f ) \leq -1$
\item $\nabla f$ is past directed
\item $\abs{f(q)-f(p)} < \epsilon$\\
\end{itemize}
\end{prop}

\begin{proof}
The case $p=q$ is trivial, so we will suppose $q\notin J^{\pm}(p)$.

Let us choose a smooth spacelike Cauchy surface $S$ containing $q$ and let us assume that $p \in J^-(S)$ (otherwise we exchange the role of $p$ and~$q$). Then we choose the following free points (Figure \ref{L6}):
\begin{itemize}
\item $p_+ \in S$ such that $p_+ \in J^+(p)$ (if $p\in S$, we just have $p_+ = p$)
\item $p' \in I^-(p)$ and $q' \in I^-(q)$ such that the sets $J^+(p') \cap S$ and $J^+(q') \cap S$ are disjoint (this is always possible because $q\in S$ and $q\notin J^{+}(p)$)
\item $p_+' \in I^+(p_+)$ and $q_+' \in I^+(q)$ such that $I^-(p_+') \cap J^+(S) \subset I^+(p')$ and $I^-(q_+') \cap J^+(S) \subset I^+(q')$
\end{itemize}
\begin{figure}[!ht]
\begin{center}
\includegraphics[width=9cm]{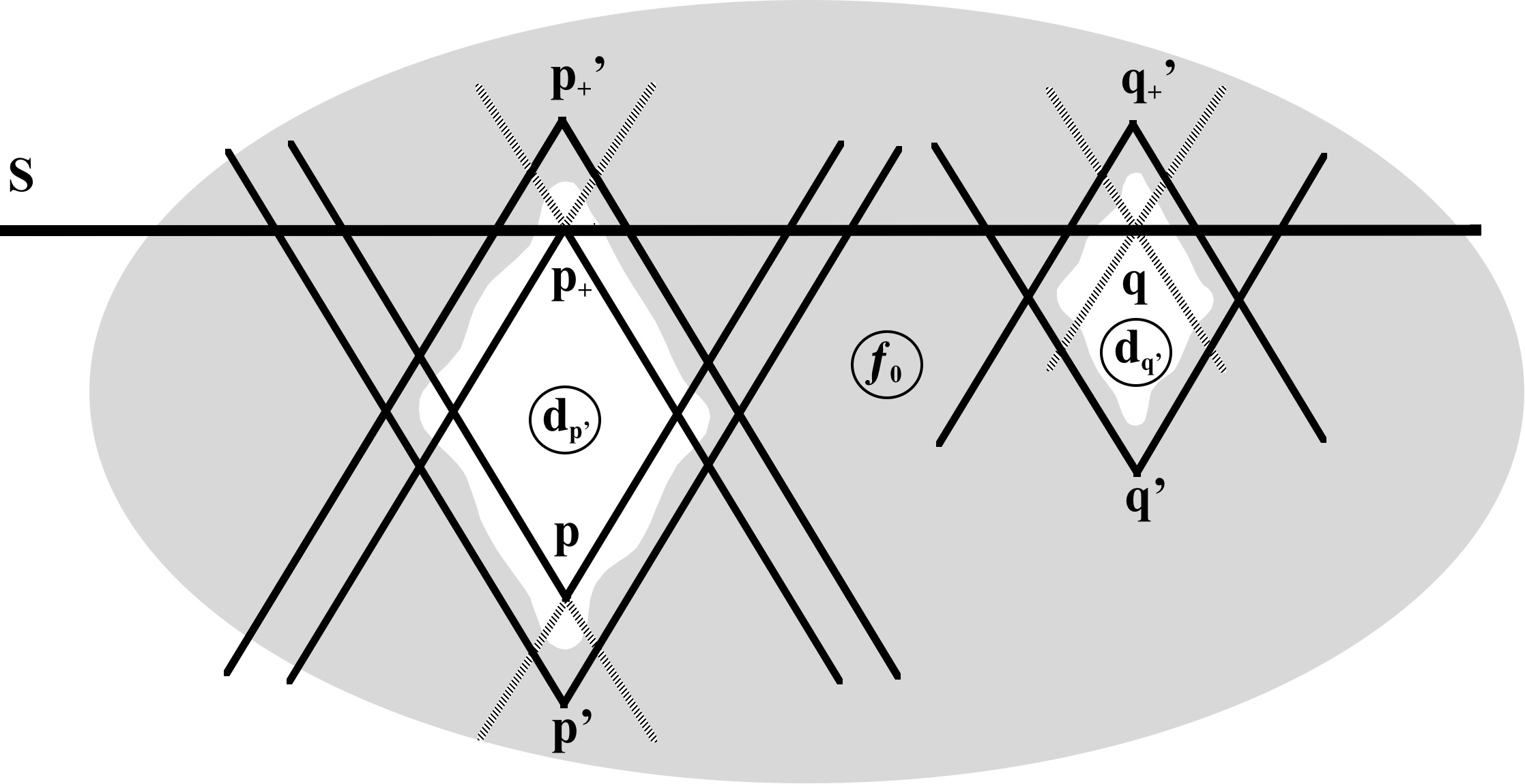}
\end{center}
\caption{}\label{L6}
\end{figure}
Then we have a similar situation to the Proposition \ref{equality} but with two disjoint sets
$$
\parenthese{J^+(S) \cap I^-(p_+')} \cup  \parenthese{J^-(S) \cap I^+(p')}
$$
and
$$
\parenthese{J^+(S) \cap I^-(q_+')} \cup  \parenthese{J^-(S) \cap I^+(q')}.
$$
The first one contains the compact $J^-(p_+) \cap J^+(p)$ and the second contains the point $q$.

We can  apply the Corollary \ref{corollaryprincipal} a first time to $S$, $p_+$, $q$, $p_+'$ and $q_+'$ to get a function $f_1$ with null or past directed gradient, and a second time to $S$, $p$, $q$, $p'$ and $q'$ to get a function $f_2$ with null or future directed gradient. In the same way as in the Proposition \ref{equality} we find a function:
$$
f_0 = f_1 - f_2
$$
with the following properties:
\begin{itemize}\itemsep=0pt
\item $f_0\in\C(\M)$
\item $f_0 = 0$ on the compact $J^-(p_+) \cap J^+(p)$ and on $q$
\item The support of $\nabla f_0$ includes the set
$$\parenthese{J^+(S) \!\setminus \! \parenthese{I^-(p_+') \cup I^-(q_+')}} \cup    \parenthese{J^-(S)\! \setminus \!\parenthese{I^+(p') \cup I^+(q') }}$$
\item $g( \nabla f_0, \nabla f_0) \leq -1$ and $\nabla f_0$ is past directed on its support, except on a set of measure zero
\end{itemize}

Finally we define the function $f$:
$$
f = f_0 + d_{p'} + d_{q'}  \quad \Lequi \quad   f(z) = f_0(z) + d(p',z) + d(q',z)
$$
which has the following properties:
\begin{itemize}\itemsep=0pt
\item $f\in\C(\M)$
\item $f(p)= d(p',p)$
\item $f(q)= d(q',q)$
\item The support of $\nabla f$ is $\M$
\item $g( \nabla f, \nabla f) \leq -1$ and $\nabla f$ is past directed on $\M$, except on a set of measure zero
\end{itemize}

Once more we can choose the points $p' \in I^-(p)$ and $q' \in I^-(q)$ such that $d(p',p) < \frac\epsilon2$ and $d(q',q) < \frac\epsilon2$, which implies  $\abs{f(q)-f(p)} < \epsilon$.\\
\end{proof}

The Proposition \ref{equality}, the Corollary \ref{equalityrev} and the Proposition \ref{zero} conclude the proof of the Theorem~\ref{maindefdist}.\\

\newpage

\section[Causality in noncommutative geometry]{Causality in noncommutative geometry}\label{seccaus}

In this last section, we turn our attention to the concepts of causality and time, and more precisely to some technical adaptations which should be introduced in order to take causality into account in noncommutative geometry.\\

We will begin by a  review of the works of F. Besnard \cite{Bes} on the noncommutative generalization of partially ordered spaces. We will see that once more the set of causal functions will be of a great importance in order to translate causal informations in noncommutative geometry. Unfortunately, there is no natural algebraic structure with finite norm which can accept those functions in the case of non-compact manifolds. In order to solve this problem, we will propose in the last part a new structure for which the causal functions are well defined, and which leads to the creation of Lorentzian spectral triples including some causal informations, as a particular time function.\\

\subsection{Noncommutative partially ordered spaces}\label{NCcausal}

As a first step in the translation of the concept of causality in noncommutative geometry, one can wonder what could be the translation of partial order. We make here a quick presentation of some elements introduced in \cite{Bes}, to which we refer for the proofs. We must remark that a first approach to generalize the causal order relation was presented by V.~Moretti in \cite{Mor}, with some similarities, but in the context of the family of local causal spaces $L^2(\bar I)$ (cf. the formula (\ref{morfunction}) and the discussion behind). We begin by a review of the basic notions of poset theory. A complete introduction on this subject can be found in the book of L.~Nachbin \cite{Nachbin}.

\begin{defn}
A {\bf\index{partially ordered set (poset)} partially ordered set (poset)} is a set $M$ admitting a partial order, i.e.~a binary relation $\leq$ defined on a subspace of $M\times M$ such that $\forall a,b,c\in M$:\vspace{-0.5em}                                                      
\begin{itemize}\itemsep=0pt
\item $a \leq a$
\item If $a \leq b$ and $b \leq a$, then $a = b$
\item If $a \leq b$ and $b \leq c$, then $a \leq c$
\end{itemize}

A {\bf\index{partially ordered set (poset)!space} partially ordered space} is a poset $(M,\leq)$ together with a structure of topological space.
\end{defn}

From this definition, it is obvious that any globally hyperbolic spacetime $\M$ is a  partially ordered space $(\M,\preceq)$ with $\preceq$ being the causal relation. Of course, partially ordered spaces cover a  larger class of spaces.

\begin{defn}
A function $f : M \rightarrow N$ between two posets $M$ and $N$ is called {\bf\index{function!isotone} isotone} if
$$\forall x,y\in M,\quad x\leq y \implies f(x) \leq f(y).$$
If $M$ is a partially ordered space, we can consider the set of continuous real isotone functions on $M$:
$$I(M) = \set{ f \in C(M,\setR) : f \text{ is isotone}}.$$
\end{defn}

\begin{defn}
A {\bf\index{partially ordered set (poset)!space!completely separated} completely separated ordered space} is a partially ordered space $M$  such that the order is completely determined by the continuous real isotone functions, i.e.
$$\forall x,y\in M,\quad x \leq y \lequi \forall f\in I(M),\ f(x) \leq f(y).$$
\end{defn}

\begin{prop}
Let $\M$ be a globally hyperbolic spacetime. Then, $\M$ is a completely separated ordered space, with the set of continuous real isotone functions being the set of causal functions $I(\M)=\C(\M)$.
\end{prop}
\begin{proof}
By definition of the causal functions and the order on $\M$, the equivalence $\C(\M)=I(\M)$ is obvious.

To show that the causal functions completely determine the order, we just need to prove that, if $x\npreceq y$, then there exists a function $f\in\C(\M)$ such that \mbox{$f(x)>f(y)$}. If $x\succ y$, then \mbox{$\T(x)>\T(y)$} where $\T\in\C(\M)$ is a time function, whose existence is guaranteed by global hyperbolocity (Corollary \ref{corgeroch}). If $x$ and $y$ are not causaly related (i.e.~$x\npreceq y$ and $x\nsucceq y$), then we can use an extension of this result given in \cite{BS06} which guarantees the existence of a  time function $\T$ such that the Cauchy surface $\T^{-1}(0)$ contains a chosen compact spacelike subset with boundary. The time function is then constructed such that $\T(x)=0$ and $\T(y)<\T(y^+)=0$, where $y^+$ is a point $y\prec y^+$  such that $x\npreceq y^+$ and $x\nsucceq y^+$ still hold.
\end{proof}

By its definition, the set $I(M)$ is a convex cone. Moreover, if $M$ is compact, it is stable by two other operations called {\it meet} and {\it join}.

\begin{defn}
The {\bf\index{meet} meet} operation $a \wedge b$ is the unique greatest lower bound of $a$ and $b$, i.e.~$a \wedge b = c$ if and only if $c \leq a$, $c \leq b$, and for all $d $ such that $d \leq a$ and  $d \leq b$ we have $d \leq c$. In a similar way, the {\bf\index{join} join} operation $a \vee b$ is the unique least upper bound of $a$ and $b$. A poset which is stable by the meet and join operations is called a {\bf\index{lattice} lattice}.
\end{defn}

So if $M$ is compact, $I(M)$ is a sublattice. Actually, any sublattice cone which contains the constants and determines the order is a dense subspace of $I(M)$, as proven by L. Nachbin \cite{Nachbin}:

\begin{thm}
Let $M$ be a compact poset. If $J \subset C(M,\setR)$ is stable by addition, by the meet and join operations, and by multiplication by non-negative reals, contains the constant functions, separates the points and is closed, then $M$ is a completely separated ordered space with $I(M)=J$.
\end{thm}

This theorem clearly shows that, in the compact case, the choice of a convex cone $I\subset \A = C(M)$ of Hermitian elements respecting some suitable conditions completely determines the order. In \cite{Bes}, the following noncommutative generalization is proposed:

\begin{defn}
Let $\A$ be a unital $C^*$-algebra. An {\bf\index{isocone} isocone} is a subset $I\subset \A$ respecting the following conditions:
\begin{itemize}
\item $I$ is composed of Hermitian elements, i.e.~\mbox{$\forall a\in I, a=a^*$}
\item $I$ is closed
\item $\forall \alpha, \beta \geq 0$ and $\forall a,b \in I$, $\alpha\,a + \beta\,b \in I$
\item $\forall a,b \in I$, $a\wedge b = \frac{a+b}{2} - \frac{\abs{a-b}}{2} \in I$ and $a\vee b = \frac{a+b}{2} + \frac{\abs{a-b}}{2} \in I$
\item $\forall \alpha \in\setR$, $\alpha 1 \in I$
\item $\overline{\text{span}(I)} = \A$
\end{itemize}
A couple $(I,\A)$ where  $\A$ is a unital $C^*$-algebra and $I$ an isocone is called an {\bf\index{algebra!$I^*$-algebra} $I^*$-algebra}.\\
\end{defn}

With this definition, the Gel'fand--Naimark theorem can be extended to compact completely separated ordered spaces:

\begin{thm}
Let $(I,\A)$ be an $I^*$-algebra such that $\A$ is commutative. Then $\Delta(\A)$ is a compact  completely separated ordered space under the partial order defined by $I$, and the Gel'fand transform \mbox{$\bigvee : (I,\A) \rightarrow (I(\Delta(\A)),C(\Delta(\A)))$} is an isometric ${}^*$-isomorphism.
\end{thm}

By defining suitable morphisms, the class of compact completely separated ordered spaces and the class of commutative $I^*$-algebras form two categories. The definition of the set of continuous real isotone functions gives a contravariant functor \mbox{$(M,\leq) \leadsto ( I(M) ,C(M) )$} and the Gel'fand transform provides a reverse contravariant functor. Then we have the following result:

\begin{thm}
The category of compact completely separated ordered spaces and the category of commutative $I^*$-algebras are dually equivalent.
\end{thm}

We can see that the concept of isocone allows us to translate a notion of partial order in noncommutative geometry, at least in the compact case. However, to obtain a complete translation of the notion of causal order, two problems remain:
\begin{itemize}
\item $I^*$-algebras are only defined for compact spaces, and we know that we must find a translation of causality for non-compact ones. One way to solve this problem is to consider the compactification of the Lorentzian manifolds (and in particular the Stone--\v Cech compactification that we will define afterwards). Some elements about such compactification are given in \cite{Bes}. The compactification seems not to be  very problematic here since it is used only to translate the information about the order, so once more this looks only like  a technical difficulty.
\item The non-commutative generalization is done for general partial orders, but we are only interested in partial orders resulting from a Lorentzian geometry. So one must find some extra conditions to guarantee that the order corresponds to a Lorentzian causal order. At this time, such conditions are still unknown.\\
\end{itemize}

\subsection{Temporal Lorentzian Spectral Triple}\label{secTLST}

We present here some unpublished research about an extension of the notion of pseudo-Riemannian spectral triple in order to take causality into account. This can be considered as a work in progress since there are some remaining technical points and since the set of axioms is far from being fixed. A definition of temporal Lorentzian spectral triples will be presented as a working basis. This new definition opens the door to a lot of new possibilities of development of spectral triples for Lorentzian manifolds.\\

The need for such research comes from the following remark: until now, every existing attempt to deal with causality in noncommutative geometry needs the use in some way of causal functions:
\begin{itemize}
\item The first approach by G.N. Partfionov and R. R. Zapatrin \cite{PZ} introduces the class of  time functions satisfying \mbox{$(\nabla f)^2 \geq 1$} (but without the knowledge whether this set is empty or not).
\item The first Lorentzian distance function proposed by V. Moretti \cite{Mor} is entirely based on the set of local causal functions $\C(\bar I)$.
\item The works by F. Besnard \cite{Bes} (cf. the Section \ref{NCcausal}) use the set of causal functions to completely determine the order in partially ordered spaces.
\item Our approach to create a global Lorentzian distance function \mbox{\cite{F2,F3}} uses the set of causal functions, at first as a way to break the symmetry of the distance function (Section \ref{secgenform}), and then as a way to solve the problem of the potential lack of absolute continuity (Section \ref{secglobalfunction}). Anyway, the conditions \mbox{$\text{\rm ess } \sup g( \nabla f, \nabla f ) \leq -1$} $(\nabla f \text{ is past directed})$ used to define the distance function are clearly only valid on a subset of $\C(\M)$.
\item Any time function or temporal function (Definition \ref{timefunction}) which could be defined for globally hyperbolic spacetimes is clearly a particular case of causal functions.
\end{itemize}

So it appears that the set of causal functions should play an important role in the generalization of noncommutative geometry to causal Lorentzian spaces. However, on non-compact spaces, most of those functions are unbounded, which means that the supremum norm (and a fortiori $L^p$ norms) cannot be used to define a Banach algebra.\\

Actually, those functions can be categorized in imbricated convex cones. We have the convex cone of causal functions which is the largest one, and contains the convex cone of time functions. If restricted to (a.e.)~differentiable functions, those cones contain the convex cone of temporal functions, which can be restricted to the convex cone of functions respecting \mbox{$\text{\rm ess } \sup g( \nabla f, \nabla f ) \leq -1$}. If we want to create a noncommutative generalization of either one of those cones, we need to find a structure of $C^*$-algebra, or at least a normed algebra, in which those cones are embedded.\\

We begin by a review of the common algebras of functions on a non-compact manifold $\M$. We will consider involutive algebras, since the real-valued functions can always be recovered by taking the Hermitian elements. We can extract four different ones, given by inclusion order:
\begin{itemize}
\item The algebra $C_c(\M)$ of continuous functions with compact support, which is a $C^*$-algebra with the supremum norm
\item The algebra $C_0(\M)$ of continuous functions vanishing at infinity, which is a $C^*$-algebra with the supremum norm
\item The algebra $C_b(\M)$ of continuous bounded functions, which is a unital $C^*$-algebra with the supremum norm
\item The algebra $C(\M)$ of continuous functions, which is unital and contains unbounded elements
\end{itemize}

Each of these algebras can  be restricted to its dense subalgebra of smooth functions (or even to Lipschitz continuous functions). In the compact case, all these algebras are equivalent, and the Gel'fand--Naimark theorem can be applied. In the non-compact case, the Gel'fand--Naimark theorem can only be applied on the $C^*$-algebra $C_0(\M)$ if we want to recover the non-compact space $\M$.\\

We can notice that the $C^*$-algebra $C_0(\M)$ is an ideal of the unital \mbox{$C^*$-algebra} $C_b(\M)$, but is not an ideal of the unital algebra $C(\M)$ (actually, $C_c(\M)$ is an ideal of $C(\M)$). The algebras $C_0(\M)$ and $C_b(\M)$ are in relationship with each other, with the notion of multiplier algebra introduced by S. Helgason \cite{Helgason}.

\begin{defn}
Let $\A$ be a $C^*$-algebra. An ideal $I \subset \A$ is said to be {\bf\index{ideal!essential} essential} if $I$ has non-zero intersection with every  non-zero  ideal of $\A$.
\end{defn}

\begin{defn}
Let $\A$ be a non-unital $C^*$-algebra. The {\bf\index{algebra!multiplier} multiplier algebra $M(\A)$} is the C*-algebra, unique up to isomorphism, which is the largest unital $C^*$-algebra that contains $\A$ as a two-sided essential ideal, largest in the sense that any other such algebra can be embedded in it.
\end{defn}

The norm on the multiplier algebra $M(\A)$ is usually the operator norm, since $M(\A)$ can be seen as an algebra of bounded multiplicative operators on $\A$ (hence its name). Since the $C^*$-algebra $\A$ is itself (by use of a representation) a normed algebra of bounded multiplicative operators on some Hilbert space $\H$, both norms are equivalent, and the multiplier algebra $M(\A)$ has a representation as bounded multiplicative operators on $\H$, with $\A$ as a two-sided essential ideal. The multiplier algebra is usually defined for complete algebras, but we can easily extend the concept to pre-$C^*$-algebra (in this case the multiplier is only a unital pre-$C^*$-algebra, whose closure corresponds to the multiplier of the closure of the algebra).

\begin{defn}
The {\bf\index{compactification!Stone--\v Cech} Stone--\v Cech compactification} $\beta X$ of a topological space $X$ is the largest compact Hausdorff space generated by $X$, with a continuous map \mbox{$i : X \rightarrow \beta X$}, largest in the sense that, for any compact Hausdorff space $Y$ and any continuous map $f :X \rightarrow Y$, there exists a unique continuous map \mbox{$g : \beta X \rightarrow Y$} such that \mbox{$f = g \circ i$}.
\end{defn}

\begin{prop}
If $\M$ is a locally compact Hausdorff space, then $M(C_0(\M)) = C_b(\M)$. Moreover, by the Gel'fand--Naimark theorem, $C_b(\M)$ is the unital $C^*$-algebra of continuous functions on some compact Hausdorff space, which corresponds to the Stone--\v Cech compactification of $\M$, i.e.~$C_b(\M)\cong C(\beta\M)$.
\end{prop}

The proof of this result can be found in \cite{Wegge}. In fact, the unitization $C_0(\M)^+$ corresponds to the Alexandroff compactification $C(\M^+)$, and the multiplier $M(C_0(\M))$ corresponds to the Stone--\v Cech compactification $C(\beta\M)$. If we consider the pre-$C^*$-algebra $C^\infty_0(\M)$, then $M(C^\infty_0(\M)) = C^\infty_b(\M)$ since the set of smooth functions vanishing at infinity forms an ideal for the set of bounded smooth functions but not for the continuous ones.\\

The multiplier algebra \mbox{$M(C_0(\M))$} is a first possibility to extend the algebra to some causal functions, since the intersection \mbox{$M(C_0(\M)) \cap \C(\M) = C_b(\M) \cap \C(\M)$} is non empty and non trivial, and a cone of bounded causal functions can be taken into consideration. This cone contains also some time or temporal functions, which are functions with a "decreasing rate of growth" such that they remain bounded.\\

However, the condition \mbox{$\text{\rm ess } \sup g( \nabla f, \nabla f ) \leq -1$} cannot be respected by any bounded temporal function, and in particular any distance function $d_p$ cannot live in this cone. So we must find a way to create an algebra of bounded operators which contains functions respecting \mbox{$\text{\rm ess } \sup g( \nabla f, \nabla f ) \leq -1$}.\\

In order to answer this problem, one could wonder which algebraic structure, with a finite norm, can be given to the set $C(\M)$. Since continuous functions are locally integrable, we can use a structure defined for the space of locally integrable functions $L^1_{\text{loc}}(\M)$. Such structure is provided by the theory of partial inner product spaces and partial *-algebras. We give here few elements of this theory. A complete introduction, including the topological aspects, can be found in the book by J.-P. Antoine and C.~Trapani \cite{Antoine}.

\begin{defn}
A {\bf\index{linear compatibility relation} linear compatibility relation} on a vector space $V$ is a symmetric binary relation $f\# g$ which preserves linearity.
\end{defn}

For every subset $S\subset V$, we can define the vector subspace \mbox{$S^\#=\set{g\in V : g\# f,\ \forall f\in S} \subset V$} which respects the inclusion property \mbox{$S \subset \prt{S^\#}^\#$}. Vector subspaces such that \mbox{$S = \prt{S^\#}^\#$} are called {\bf\index{assaying subspaces} assaying subspaces}.\\

The family of assaying subspaces forms a lattice  with the inclusion order, and the meet and join operations given by \mbox{$S_1 \wedge S_2 =S_1 \cap S_2$} and  \mbox{$S_1 \vee S_2 = \prt{S_1 + S_2}^{\#\#}$}. Usually, it is enough to consider only an indexed sublattice \mbox{$\I = \set{S_r : r\in\I}$}  which covers $V$, with an involution defined on the index $I$ by $\prt{S_r}^\# = S_{\overline r}$.

\begin{defn}
A {\bf\index{inner product!partial} partial inner product} on $(V,\#)$ is a (possibly indefinite) Hermitian form $\scal{\,\cdot\,,\,\cdot\,}$ defined exactly on compatible pairs of vectors $f\# g$. A {\bf\index{partial inner product space} partial inner product space (PIP-space)} is a vector space $V$ equipped with a linear compatibility relation and a partial inner product. An {\bf\index{partial inner product space!indexed} indexed PIP-space} is a PIP-space with a fixed generating involutive indexed sublattice of assaying subspaces.
\end{defn}

The space \mbox{$L^1_{\text{loc}}(\M)$} is a PIP-space with the compatibility relation given by
$$f\# g\quad\lequi\quad \int_\M \abs{fg} d\mu_g \ <\  \infty$$
and the partial inner product
$$\scal{f,g} = \int_\M f^*g \;d\mu_g .$$

We can define an involutive indexed sublattice of assaying subspaces by using weight functions. If, for $r,r^{-1} \in L^2_{\text{loc}}(\M)$ with $r$ a.e.~positive Hermitian, we define the space $L^2(r)$ of measurable functions $f$ such that $fr^{-1}$ is square integrable, i.e.
$$L^2(r) = \set{f\in L^1_{\text{loc}}(\M) : \int_\M \abs{f}^2 r^{-2} \;d\mu_g \;<\; \infty},$$
then the family \mbox{$\I = \set{L^2(r)}_r$} respects $L^1_{\text{loc}} = \bigcup_{r} L^2(r)$ and is a generating sublattice of assaying subspaces, with an involution defined by $\overline{r} = r^{-1}$ since \mbox{$L^2(r)^\# = L^2(r^{-1})$}.\\

Actually, we have a realization of this indexed PIP-space as a lattice of Hilbert spaces \mbox{$\set{\H_r}_r$} where each \mbox{$\H_r = L^2(r)$} is endowed with the positive definite Hermitian inner product
$$\scal{f,g}_r = \int_\M f^*g \; r^{-2} \;d\mu_g .$$
In the following, we will denote such indexed PIP-space by $\H = \bigcup_{r} \H_r$, with the center space $\H_0$ being the space of square integrable functions.\\

The space $C_b(\M) = M\prt{C_0(\M)}$ is a unital $C^*$-algebra which acts as bounded multiplicative operators on $\H^0$, with the operator norm corresponding to the supremum norm \mbox{$\norm{f}_\infty = \sup_{x\in\M} \abs{f(x)}$}. We can wonder what happens if we make a similar modification to the supremum norm, by introducing a continuous weight function $r$. We denote the algebra of bounded continuous functions by $\A^M_0 = M\prt{C_0(\M)}$, where  $M$ here denotes the multiplier algebra. We define the following lattice of spaces:
$$\A^M_r = \set{ f\in C(\M) : \sup_{x\in\M} \abs{f(x)\;r^{-1}(x)} \;<\;\infty}\cdot$$

Except for the center space $\A^M_0$, each $\A^M_r $ is a vector space which has not a structure of algebra. Instead, those spaces have a structure of {\bf\index{algebra!*-algebra!partial} partial *-algebra}, which means that the product $fg\in \A^M_r$ is well defined only for a bilinear subset of $\A^M_r \times \A^M_r $. Actually, we have that $fg\in \A^M_r$ if $f\in \A^M_s$ and $g\in \A^M_t$, with $r=st$.\\

Each $\A^M_r $ is endowed with a norm $\norm{\,\cdot\,}_r = \norm{\,\cdot\ r^{-1}}_0$ , where  $\norm{\,\cdot\,}_0$ is the operator norm on $M\prt{C_0(\M)}$.
In the same way as the lattice of Hilbert spaces, we will define the partial\footnote{In fact, this is an abuse of language in order to say that each element of the lattice is a partial algebra, even if their union is an algebra.} *-algebra \mbox{$\A^M = \bigcup_{r} \A^M_r$} which obviously corresponds to the algebra of continuous functions $C(\M)$ .\\

It is trivial that \mbox{$\A^M = \bigcup_{r} \A^M_r$} acts on the PIP-space $\H = \bigcup_{r} \H_r$ as multiplicative operators. Moreover, we have that each $a\in A^M_s$ acts as a bounded operator \mbox{$a : \H_r \rightarrow \H_{rs}$}.\\

So we see that it is possible to consider $C(\M)$ as a partial *-algebra of bounded operators on some PIP-space,\footnote{We must remark that, under few hypotheses, it is possible to construct a generalization of the \mbox{GNS-representation} of partial *-algebras into PIP-spaces \cite{Antoine}.} with a set of norms defined on a generating sublattice. However, we do not need the whole space $C(\M)$. Indeed, causal functions are only unbounded "in the direction of the time", which means that we do not really need functions which are growing indefinitely on spacelike surfaces. So we propose the following idea: the set of weight functions can be restricted in order that the algebra \mbox{$\A^M = \bigcup_{r} \A^M_r$} contains some functions which are growing indefinitely only along causal curves. In that way, the chosen set of weight functions will determine a constraint on causality. A typical way is to restrict $C(\M)$ to the functions which satisfy a growth condition along causal curves.\\

\vspace{0.5em}                                                      

Once more, we will suppose that $\M$ is a globally hyperbolic spacetime admitting a spin structure. We know by the Corollary \ref{corgeroch} that $\M$ admits a Cauchy temporal function $\T$. Then we construct the following partial *-algebra:
\begin{itemize}
\item $\A^M_0 = M(C_0^\infty(\M))$ is the unital pre-$C^*$-algebra of smooth bounded functions on $\M$
\item $\A^M_\alpha = \set{ f\in C^\infty(\M) : \sup_{x\in\M} \abs{f(x)\; \prt{1+\T(x)^2}^{\frac \alpha2}} \;<\;\infty}$
\item \mbox{$\A^M = \bigcup_{\alpha \in\setR} \A^M_\alpha$} is the unital partial *-algebra of smooth  functions on $\M$ of polynomial growth which are bounded on each Cauchy surface $\T^{-1}(x)$, $x\in\M$.
\end{itemize}

\vspace{0.5em}                                                      

This partial *-algebra has a large number of interesting properties, easily derived from the definition:
\begin{itemize}
\item Each set $\A^M_\alpha$ is simply generated from $\A^M_0$ by $f\in\A^M_\alpha$ if and only if $f\prt{1+\T(x)^2}^{\frac \alpha2} \in\A^M_0$
\item Each set $\A^M_\alpha$ is endowed with a norm $\norm{\,\cdot\,}_\alpha = \norm{\,\cdot\ \prt{1+\T(x)^2}^{\frac \alpha2}}_0$ where $\norm{\,\cdot\,}_0$ is the norm on $\A^M_0$
\item We have the partial rule product 
$$ab\in\A^M_{\alpha} \ \lequi\ \exists \,\beta,\gamma\in\setR\ :\alpha = \beta + \gamma \ ,\ a\in\A^M_\beta, \ b\in\A^M_\gamma$$
\item $\T\in\A^M$, and more precisely $\T\in\A^M_{-1}$
\item All smooth causal functions with polynomial growth are included in the algebra $\A^M$
\item All a.e.~differentiable causal functions with polynomial growth are in the closure algebra \mbox{$\overline{\A^M} = \bigcup_{\alpha\in\setR} \overline{\A^M_\alpha}$}, where the closure is taken in each subset with respect to each norm $\norm{\,\cdot\,}_\alpha$. In particular, the Lorentzian distance function $d_p$ belongs to $\overline{\A^M}$, as well as the functions with polynomial growth respecting \mbox{$\text{\rm ess } \sup g( \nabla f, \nabla f ) \leq -1$}.
\end{itemize}

Of course, the polynomial growth is just a particular choice. For example, one can define a similar  partial *-algebra with causal functions of exponential growth by taking a weight function like \mbox{$e^{\alpha \abs{\T}}$}, $\alpha\in\setR$.\\

Now we can remember that the pre-$C^*$-algebra $\A_0 = C_0^\infty(\M)$, as well as its unitization $\A^M_0 = M(C_0^\infty(\M))$, act as multiplicative bounded operators on the Hilbert space $\H_0 = L^2(\M,S)$ of square integrable spinor sections over $\M$. From that, we can construct the indexed PIP-space $\H = \bigcup_{\alpha\in\setR} \H_\alpha$ of spinor sections over $\M$ which are square integrable under the weighted normed, and on which $\A^M = \bigcup_{\alpha \in\setR} \A^M_\alpha$ acts as a family of bounded operators.\\

Such space corresponds to a scale of Hilbert spaces generated by the self-adjoint operator \mbox{$\prt{1+\T^2}^{\frac 12} \geq 1$}. Indeed, \mbox{$\T : \text{Dom}\prt{\T}  \rightarrow \H_0$} is an unbounded self-adjoint operator acting on $\H_0= L^2(\M,S)$ with domain \mbox{$\text{Dom}\prt{\T}\subset \H_0$}, and \mbox{$\prt{1+\T^2}^{\frac 12}$} is an unbounded positive self-adjoint operator with similar domain \mbox{$\text{Dom}\prt{\T} = \text{Dom}\prt{\prt{1+\T^2}^{\frac 12}}$}. We can define the following Hilbert spaces:
$$\H_n = \bigcap_{k=0}^n \text{Dom}\prt{\prt{1+\T^2}^{\frac k2}} = \bigcap_{k=0}^n \text{Dom}(\T^k)\quad\forall n\in\setN.$$
The conjugate spaces \mbox{$\H_{\overline n} = \H_{-n}$}, $n\in\setN$ are just the topological duals of the spaces $\H_n$ (properly speaking, there are the identifications of the topological duals under the Hermitian inner product on $\H_0$). Then we have the discrete scale of Hilbert spaces:
$$ \bigcap_{n\in\setZ} \H_n \subset \dots \subset \H_2 \subset \H_1 \subset \H_0 \subset \H_{-1} \subset \H_{- 2} \subset \dots \subset  \bigcup_{n\in\setZ} \H_n = \H.$$

By interpolation theory, this discrete scale can be extended to a continuous scale $\set{\H_\alpha}_{\alpha\in\setR}$ (but the discrete scale is sufficient). Each $\H_\alpha$ is endowed with the positive definite Hermitian inner product:
$$(\psi,\phi)_\alpha = \prt{\psi,\prt{1+\T^2}^{ \alpha} \phi}_0 =  \int_\M \psi^*\phi \, \prt{1+\T^2}^{ \alpha}\,d\mu_g.$$

Then $\A^M = \bigcup_{\alpha \in\setR} \A^M_\alpha$ acts as a family of bounded multiplicative operators on $\H = \bigcup_{\alpha\in\setR} \H_\alpha$ with  \mbox{$a : \H_\alpha \rightarrow \H_{\alpha + \beta}$} if \mbox{$a\in\A^M_\beta$} and for each $\alpha \in\setR$. The norm \mbox{$\norm{\,\cdot\,}_\beta$} on \mbox{$a\in\A^M_\beta$} corresponds to the operator norm on \mbox{$\H_\alpha \rightarrow \H_{\alpha + \beta}$} and it is independent of the chosen $\H_\alpha$.\footnote{We will give the proof of this property later for the general case, including noncommutative algebras.}\\

So in our construction, we have three elements $\A_0,\H_0,\T$ which generate a unital indexed partial *-algebra $\A^M_0 = \bigcup_{\alpha \in\setR} \A^M_\alpha$ acting on an indexed PIP-space $\H = \bigcup_{\alpha\in\setR} \H_\alpha$, where the index $\alpha$ can also be restricted to $\setZ$. The only thing which is missing in order to have a generalized construction of a commutative spectral triple is a Dirac operator $D$.\\

The Dirac operator \mbox{$D = -i(\hat c \circ \nabla^S)$} acts as an unbounded operator on $\H_0$ with dense domain, and we can consider its extension to the PIP-space $\H$. We know that the Dirac operator respects the following properties:
\begin{itemize}
\item For every $a\in\A^M$, $[D,a] = -i\,c(da)$, since the construction of the Theorem \ref{commudirac} in the Section \ref{spinsec} still holds for non-compact pseudo-Riemannian 
manifolds and every smooth function.
\item By the Proposition \ref{cpt1} in the Section \ref{Krein}, if there exists a spacelike reflection $r$ on $\M$ such that the Riemannian metric $g^r$ associated is complete, then $D$ is essentially Krein-self-adjoint in the Krein space defined by the associated fundamental symmetry $\J_r$.
\end{itemize}

Those properties lead to two important technical remarks. First, one could think -- as it is sometimes wrongly admitted in the literature -- that the condition that $[D,a]$ is bounded for each $a\in C^\infty_0(\M)$ is still valid. However, this is not the case for non-compact manifolds since there exist some functions vanishing at infinity whose derivatives are unbounded (as example, we have the functions decreasing at infinity with infinite oscillations of increasing frequency). If we want to conserve this property, we need to restrict the pre-$C^*$-algebra $\A_0$ and its unitization $\A^M_0$ to functions respecting the condition of boundedness of the commutator $[D,a]$. Typically, this can be done by restricting $\A_0$ and $\A^M_0$ to functions of bounded derivatives (as e.g. Schwartz space, cf.~\cite{Gayral} for a construction of such algebra on a particular noncommutative case). Those algebras are not necessarily dense in the previous ones, but their still fit our goal since causal functions with polynomial growth have obviously bounded derivatives.\\

Second, the Krein-self-adjointness condition of Dirac operator is only proven when the manifold is complete under the metric obtained from the spacelike reflection. This condition is trivial for compact manifolds, but not for non-compact ones. Nevertheless, this is not a so strong condition to work with Lorentzian manifolds which are complete for an associated Riemannian metric, which implies that every Cauchy surface must be complete and that every inextendible geodesic is complete. When it is not the case, the result still holds for a restriction on a complete subspace. Some characterizations of completeness for globally hyperbolic spacetimes can be found in \cite{Beem}.\\

Now we arrive at the last ingredient of our construction. We can suppose that we have a Dirac operator $D$ which is Krein-self-adjoint for a Krein space obtained by a modification of the PIP-space $\H$ under a fundamental symmetry (which is not necessarily unique, but every fundamental symmetry is sufficient to characterize the Krein space).\\

Since $\M$ is globally hyperbolic, from the Theorem \ref{thgeroch} and the Corollary \ref{corgeroch} we know that $\T$ is a smooth time function with past-directed timelike gradient everywhere whose level sets are smooth Cauchy surfaces. So the Lorentzian metric admits a globally defined  orthogonal\footnote{More precisely, the splitting \mbox{$g = -d\T^2 + g_\T$} is obtained from a conformal transformation of the metric, in order to get an orthonormal splitting. So the result is valid for the equivalence classes under conformal invariance.} splitting \mbox{$g = -d\T^2 + g_\T$}, where $g_\T$ is a Riemannian metric on each level set $S_\T$. This splitting on the metric induces a splitting on the tangent (and the cotangent) bundle \mbox{$T\M = T\M_- \oplus T\M_+$}, where the subbundle $T\M_-$ has a one dimensional fiber generated by the gradient $\nabla\T$. So the temporal function $\T$ defines a spacelike reflection, with the associated Riemannian metric being \mbox{$g^r = d\T^2 + g_\T$}. Since $\nabla\T$ is a generating element of $T\M_-$, $d\T$ is a generating element of $T\M_-^*$.\\

From the Proposition \ref{propslr}, we know that each spacelike reflection generates a fundamental symmetry $\J$ defined from the Clifford action $c$ on a local oriented orthonormal basis \mbox{$\set{e}$} of $T\M_-^{*}$ by $\J = ic(e)$. So $\J = ic(d\T) = i\gamma^0$ is a fundamental symmetry, and if the manifold $\M$ is complete under the metric \mbox{$g^r = d\T^2 + g_\T$}, then $D$ is Krein-self-adjoint for the Krein space defined by $\J$, which is equivalent to the fact that $D\J$ (or $\J D$) is a self-adjoint operator in $\H$.\\

Then, since we have  $[D,a] = -i\,c(da)$ $\forall a\in\A^M$, we can conclude that $\J = - [D,\T]$, or equivalently $\J = [D,\T]$, is a fundamental symmetry of the Krein space. The condition on the Krein-self-adjointness of $D$ becomes a condition on the self-adjointness of $D[D,\T]$ or equivalently $[D,\T]D$ on the PIP-space $\H$.\\

The particular construction of this fundamental symmetry forces the Dirac operator to correspond to a metric with Lorentzian signature, since its self-adjointness is recovered by the multiplication by a single Dirac matrix. So we have now enough elements to propose a generalization of the notion of pseudo-Riemannian spectral triple, with a guaranty on the Lorentzian signature. Moreover, by introducing the extension to PIP-spaces, we can create an extended algebra, with a set of norms defined on it, on which causal or temporal cones can be defined. In particular, the temporal function $\T$ giving the fundamental symmetry is an element of this extended algebra.\\

The following definition is a proposition of generalization of this construction to noncommutative algebras. Of course we need to impose some conditions on the commutation between particular elements in order to conserve similar properties, and we will explain our choices in the following remarks. This definition is presented as a working basis, since we do not know at this time which exact conditions must be imposed in order to properly define Lorentzian spectral triples.

\vspace{2em}                                                      

\begin{defn}
A {\bf\index{spectral triple!temporal Lorentzian} Temporal Lorentzian Spectral Triple} \mbox{$(\A_0,\H_0,D,\T)$} is the data of:
\begin{itemize}
\item A Hilbert space $\mathcal{H}_0$ with positive definite inner product \mbox{$(\,\cdot\,,\,\cdot\,)_0$}
\item A non-unital pre-$C^*$-algebra $\mathcal{A}_0$ with a representation as bounded multiplicative operators on $\mathcal{H}_0$  with operator norm \mbox{$\norm{\,\cdot\,}_0$} 
\item An unbounded self-adjoint operator $\T$ in $\H_0$ with domain $\text{Dom}(\T)\subset \H_0$ such that \mbox{$\prt{1+ \T^2}^{-\frac{1}{2}}\in \A^M_0$}  and commutes with all elements in $\A^M_0$, where $\A^M_0 = M(\A_0)$ is the unital multiplier algebra of $\A_0$
\begin{itemize}
\item $\T$ generates an indexed partial  *-algebra \mbox{$\A^M = \bigcup_{n\in\setZ} \A^M_n$} by:
$$a\in \A^M_{n+1}\quad\text{ if and only if }\quad  \prt{1+ \T^2}^{\frac{1}{2}} a \in \A^M_{n},$$ with a norm \mbox{$\norm{\,\cdot\,}_n = \norm{(1+\T^2)^{\frac n2}\ \cdot\,}_0$} defined on each $ \A^M_n$.
\item $\A^M$ has a representation as a family of bounded multiplicative operators on the indexed PIP-space:
$$\mathcal{H} = \bigcup_{n\in\setZ} \H_n\quad\text{with}\quad\H_n = \bigcap_{k=0}^n \text{Dom}(\T^k)\quad\forall n\geq 0$$
$$\text{and }\  \H_{-n} =  \H_{\overline n} \ \text{ the conjugate dual of }\ \H_n$$
with \mbox{$a : \H_n \rightarrow \H_{n+m}$} if \mbox{$a\in\A^M_m$} and with a positive definite inner product  \mbox{$(\,\cdot\,,\,\cdot\,)_{n}=(\,\cdot\,,\prt{1+ \T^2}^{n}\,\cdot\,)_0$} defined on each $\H_n$.
\end{itemize}
\item An unbounded operator $D$ densely defined on $\mathcal{H}$ such that:
\begin{itemize}
\item \mbox{$[D,\prt{1+\T^2}^{\frac 12}]$} and \mbox{$[D,\T]$} commute with all elements in $\A^M$
\item all commutators $[D,a]$ are bounded for every $a\in\A^M$
\item \mbox{$D [D,\T]$} is a self-adjoint operator in $\H$
\item \mbox{$\J=[D,\T]$} defines a Krein space structure for $\H$, with \mbox{$\J^*=\J$} and \mbox{$\J^2=1$}
\end{itemize}
\end{itemize}
\end{defn}

\begin{defn}
A temporal Lorentzian spectral triple \mbox{$(\A_0,\H_0,D,\T)$} is {\bf\index{spectral triple!temporal Lorentzian!finitely summable} finitely summable} (or \mbox{$n^+$-summable}) if there exists a positive integer $n$ such that $a\,{\Delta_\T}\!\!\!^{-n} \in \L^{1+}$ for all $a\in\A_0$, where
$$\Delta_\T = \prt{ 1 + [D]_\T^2 }^{\frac 12}$$
with
$$[D]_\T^2 = \frac{1}{2}\prt{D [D,\T] D [D,\T]+ [D,\T] D [D,\T] D}.$$
\end{defn}

\vspace{1em}                                                      

Let us add some comments et precisions about the definition of temporal Lorentzian spectral triples:
\begin{itemize}
\item The element $\T\in\A^M$ is the {\bf\index{temporal element} temporal element} of the temporal Lorentzian spectral triple. It represents a notion of global time for the spectral triple. Of course this time is global and not local, and so it is not sufficient to determine the causality by itself. We can notice that the whole definition is invariant under the modification $\T \rightarrow -\T$ which simply corresponds to a choice of time orientation.
\item The condition that \mbox{$\J = [D,\T]$} determines a Krein space, with $D$ being a Krein-self-adjoint operator for this Krein space, is extremely important. Indeed, this is the only indication that the spectral triple corresponds to a Lorentzian geometry, with signature \mbox{$(n-1,1)$}. If we remove this condition, we simply have a pseudo-Riemannian spectral triple with an extension of the algebra to a  larger class of functions. This condition can also be written as \mbox{$\J = -[D,\T]$}, \mbox{$\J = i[D,\T]$} or \mbox{$\J = -i[D,\T]$} (with for the last two cases the condition that \mbox{$iD [D,\T]$} is a self-adjoint operator). As we have already said, the addition of a minus sign corresponds to a change of time orientation, which is similar to a change of sign in the Krein inner product. The addition of a $i$ factor corresponds to a change in the Hermicity conditions of the Clifford representation, so is similar to a switch from a signature \mbox{$(-,+,+,+,\dots)$} to a signature \mbox{$(+,-,-,-,\dots)$}.
\item To the four elements $\A_0$, $\H_0$, $D$ and $\T$, a fifth one is missing if we want to create a spectral triple with a notion of causality. The additional element should be a cone defined among the Hermitian elements of $\A^M$ and representing a set of causal, time or temporal functions which should be sufficient to determine the causality. We can expect that such cone would have a definition similar to that in the Section \ref{NCcausal}, but with a link with the temporal element $\T$ in order to guarantee that the induced causal order corresponds to a Lorentzian manifold.
\item One could add to the definition of temporal Lorentzian spectral triples a real condition and an even condition, with the introduction of a $\setZ_2$-grading $\gamma$ and an antilinear isometry $J$ respecting suitable commutative conditions with the other elements. Conditions can also be added concerning the smoothness of the Dirac operator $D$, which we have not included here. Moreover, further conditions must certainly be added in order to guarantee a unique correspondence between commutative temporal Lorentzian spectral triples and (equivalence classes of) complete globally hyperbolic Lorentzian manifolds.
\item We use here the multiplier algebra $\A^M_0 = M(\A_0)$ as unitization of $\A_0$, since its definition can easily be transposed to non-unital noncommutative algebras. However, this maximal unitization is very large, and if one would like to construct such unitization for particular cases, it could be  best to consider another preferred unitization $\tilde\A_0 \subset \A_0^M$. An example of such preferred unitization can be found in \cite{Gayral} for the Moyal plane.
\item In the commutative case, smooth causal functions with polynomial growth are included in the algebra $\A^M$, but a.e.~differentiable ones -- as the distance function $d_p$ -- are not, at least those which do not respect the boundedness condition of \mbox{$[D,f]$}. In order to include those functions, we must consider the closure algebra
$$\overline{\A^M} = \bigcup_{n\in\setZ} \overline{\A^M_n}$$
with respect to each norm $\norm{\,\cdot\,}_n$.
\item The partial rule product
$$a\in\A^M_m, \ b\in\A^M_n\ \implies \ ab\in\A^M_{m+n}\ \text{ and }\ ba\in\A^M_{m+n}$$
is still valid even in the noncommutative case, since we have that $\prt{1+ \T^2}^{\frac{n}{2}}$ commutes with all elements in $\A^M$ for each $n\in\setZ$. Indeed, $\A^M$ is generated by multiple applications of $\prt{1+ \T^2}^{\frac{1}{2}}$ or $\prt{1+ \T^2}^{\frac{-1}{2}}$ on $\A^M_0$, and if \mbox{$a\in\A^M_0$}, we have already the condition \mbox{$[\prt{1+ \T^2}^{-\frac{1}{2}},a] = 0$}, which implies \mbox{$[\prt{1+ \T^2}^{\frac{1}{2}},a] = 0$} by multiplication on both sides by $\prt{1+ \T^2}^{\frac{1}{2}}$.
\item The indexed families of partial algebras and Hilbert spaces $\set{\A^M_n}_{n\in\setZ}$, $\set{\H_n}_{n\in\setZ}$ can be extended to continuous scales $\set{\A^M_\alpha}_{\alpha\in\setR}$, $\set{\H_\alpha}_{\alpha\in\setR}$ by using interpolation theory with the operator \mbox{$\prt{1+ \T^2}^{\frac{\alpha}{2}}$}.
\item Among the family of partial algebras $\A^M_n$, only $\A^M_0$ is a pre-Banach algebra. Indeed, for $a\in\A^M_m$ and $b\in\A^M_n$, \mbox{$ab\in\A^M_{m+n}$}  respects the relation \mbox{$\norm{ab}_{m+n} \leq \norm{a}_{m} \norm{b}_{n}$}. In the same way, the $C^*$-algebra condition looks like  \mbox{$\norm{a^*a}_{2n} = \norm{a}^2_{n}$}.
\item The condition that \mbox{$[D,\T]$} commutes with all elements in $\A^M$ ensures that the involution in $\A^M$ corresponds to the Krein-adjoint operation in the Krein space representation.
\item The norm \mbox{$\norm{\,\cdot\,}_n = \norm{(1+T^2)^{\frac n2}\ \cdot\,}_0$} for $ \A^M_n$ is still the operator norm. Indeed, if we consider $a\in\A^M_m$ as an operator \mbox{$a : \H_n \rightarrow \H_{n+m}$} for some $n\in\setZ$, then:
$$\norm{a}_{\text{op}} = \!\!\!\!\sup_{\phi\in\H_n,\;\phi\neq0}\!\!\!\! \frac{(a\,\phi , a\,\phi )_{n+m}}{ (\phi , \phi)_n} =  \!\!\!\!\sup_{\phi\in\H_n,\;\phi\neq0}\!\!\!\! \frac{(a\,\phi\ ,\  \prt{1+ \T^2}\!^{n+m} a\,\phi )_0}{ (\phi\ ,\  \prt{1+ \T^2}^{n} \phi)_0}$$
$$ = \!\!\!\!\sup_{\phi\in\H_n,\;\phi\neq0}\!\!\!\! \frac{\prt{\prt{1+ \T^2}\!\!^{\frac m2}a\,\prt{1+ \T^2}\!\!^{\frac n2}\phi\ ,\  \prt{1+ \T^2}\!\!^{\frac m2} a\,\prt{1+ \T^2}\!\!^{\frac n2}\phi }_0}{ \prt{\prt{1+ \T^2}^{\frac n2} \phi \ ,\  \prt{1+ \T^2}^{\frac n2} \phi}_0}$$
$$ = \!\!\!\!\sup_{\phi\in\H_0,\;\phi\neq0}\!\!\!\! \frac{\prt{\prt{1+ \T^2}\!\!^{\frac m2}a\,\phi\ ,\  \prt{1+ \T^2}\!\!^{\frac m2} a\,\phi }_0}{ \prt{\phi \ ,\   \phi}_0}$$
$$ = \norm{\prt{1+ \T^2}\!\!^{\frac m2}a}_0 = \norm{a}_m$$
where we use the facts that $\prt{1+ \T^2}\!^{\frac 12}$ and $a$ commute with each other and that $\prt{1+ \T^2}\!^{\frac 12}$ is self-adjoint. The result is clearly independent of the chosen $\H_n$.
\item The condition that \mbox{$[D,\prt{1+\T^2}^{\frac 12}]$} commutes with all elements in $\A^M$ is equivalent to the fact that, for all $a\in\A^M$, \mbox{$[D,a]$} commutes with $\prt{1+\T^2}^{\frac 12}$. Indeed, by using the fact that $\prt{1+\T^2}^{\frac 12}$ commutes with $a$:
\begin{eqnarray*}
&& [D,\prt{1+\T^2}^{\frac 12}] \;a - a \;[D,\prt{1+\T^2}^{\frac 12}]\\
&=& D\prt{1+\T^2}^{\frac 12} a - \prt{1+\T^2}^{\frac 12} D a\\
&&\qquad\qquad\qquad\qquad - a D \prt{1+\T^2}^{\frac 12} + a \prt{1+\T^2}^{\frac 12}D\\
&=& Da \prt{1+\T^2}^{\frac 12}  - \prt{1+\T^2}^{\frac 12} D a\\
&&\qquad\qquad\qquad\qquad - a D \prt{1+\T^2}^{\frac 12} + \prt{1+\T^2}^{\frac 12} aD\\
&=& [D,a] \prt{1+\T^2}^{\frac 12} - \prt{1+\T^2}^{\frac 12} [D,a].
\end{eqnarray*}
Then, by a reasoning similar to above, the operator norm of $[D,a]$ is independent of the chosen $\H_n$.\\
\end{itemize}

We have presented here a way to define noncommutative spectral triples corresponding to Lorentzian geometry and containing a time element. The axioms proposed form a working basis, since several other conditions must be added in order to obtain some even, reality or smoothness conditions. Those temporal Lorentzian spectral triples could be used as a support to define causal cones describing the causal order relation.\\

Temporal spectral triples contain four elements $\A_0$, $\H_0$, $D$ and $\T$, which can be interpreted in the following way:
\begin{itemize}
\item $\A_0$ is the algebra which represents the possibly noncommutative space.
\item $\H_0$ is just a support space on which the three other elements act, and in fact is independent of the signature of the space.
\item $D$ and $\T$ are two unbounded operators closely related to each other, mainly by the condition that $D[D,\T]$ is self-adjoint. $[D,\T]$ generates an "algebraic Wick rotation" on the support space $\H_0$ in order to create a Lorentzian signature for the metric which is represented by $D$.
\item Moreover, the self-adjoint operator $\T$ representing the time can be used to generate larger spaces $\A^M$ and $\H$ containing causal elements.
\end{itemize}

The next step should be the construction of particular examples of noncommutative temporal Lorentzian spectral triples, but this work is not simple since the construction of spectral triples with noncommutative non-unital algebras is far from trivial. In particular, the construction of a non-unital algebra $\A_0$ such that $[D,a]$ is bounded for every $a\in\A_0$ as well as for its unitization is a technical difficulty to take into consideration. Existing technics to create noncommutative non-unital spectral triples are Moyal deformations \cite{Gayral} and isospectral deformations \cite{Suij}.


\chapter*{Conclusion\markboth{Conclusion}{Conclusion}}
\addcontentsline{toc}{chapter}{Conclusion}

At the end of this walk in the noncommutative world, and especially in its Lorentzian aspects, it is a good time to think about the diverse challenges that this generalization is facing. We do not really want to talk here about a number of {\it perspectives}, since the establishment of a hyperbolic version of Connes' noncommutative geometry is at a such preliminary stage that any progression on this subject is an interesting perspective. Every newly open door leads to the emergence of dozens of new problems or questions. Our contributions have brought us two new elements -- a global path independent formulation of the Lorentzian distance and a set of axioms defining temporal Lorentzian spectral triples -- which are the first steps to possible further developments, and these two elements have their own remaining questions.\\

A formulation of the Lorentzian distance (\ref{maineqdist}) in term of the Dirac operator is still missing, and this is a need in order to generalize such distance to noncommutative spectral triples. One problem of course is the Lorentzian character of the condition \mbox{$\text{\rm ess } \sup g( \nabla f, \nabla f ) \leq -1$}. The musical isomorphism between $\nabla f$ and $df$ cannot be used to transpose the condition to the differential operators since this isomorphism is dependent of the signature of the metric, and the information relative to causality is lost in the process. An idea for this purpose could be the construction of an operatorial formulation not based on the usual operator norm on the Hilbert space $\H_0 = L^2(\M,S)$ but on a non necessarily positive definite inner product derived from a fundamental symmetry $\J$ (and possibly from a time element $\T$). Another problem is the fact that the functions respecting such condition cannot be considered as bounded multiplicative operators on $\H_0$, but the definition of temporal Lorentzian spectral triples brings us a new way to define a similar representation. The last question is the choice of the set of functions. We have seen that the cone of causal functions fits our needs, but in such a case some axioms defining the causal cone must be added to the definition of temporal Lorentzian spectral triples. Another possibility is to create a set of functions respecting the condition of absolute continuity, probably by using a Sobolev norm instead of the usual one, but in this case an additional condition on the constant orientation of the gradient is needed. Such condition could be related to the time element $\T$ which corresponds to a smooth timelike vector field.\\

The construction of temporal Lorentzian spectral triples provides also a lot of questions. This construction is actually a first attempt to combine the causal approach from V.~Moretti, F.~Besnard and N.~Franco to the Krein space approach from A.~Strohmaier and M.~Paschke. The purpose is to construct at the end a Lorentzian distance on spectral triples, with the introduction of a notion of causality in noncommutative geometry. The condition of existence of temporal Lorentzian spectral triple is a first concern. We know that commutative temporal Lorentzian spectral triples exist.  We just need a globally hyperbolic spacetime $\M$ with a spin structure and a metric splitting \mbox{$g = -d\T^2 + g_\T$}, possibly obtained after a conformal transformation on the metric, such that the manifold is complete under the Riemannian metric \mbox{$g^r = d\T^2 + g_\T$}. $D$ is the Dirac operator which is acting on $\H_0 = L^2(\M,S)$,  and the algebra \mbox{$\A_0\subset C(\M)^\infty$} is chosen such that \mbox{$[D,a]$} is bounded for all $a\in\A_0$ (so with functions of bounded derivatives). The existence of noncommutatives temporal Lorentzian spectral triples is less trivial, so the next challenge is the construction of particular noncommutative cases. Noncommutative Moyal planes \cite{Gayral} should provided a good basis for such construction. Of course the axioms of temporal Lorentzian spectral triples are subject to further development. In particular, the question of the smoothness of the Dirac operator (of more precisely of its elliptic modification $\Delta_\J$) should be taken into account, in order to have a similar condition to the admissible fundamental symmetries proposed in \cite{Stro}. Finally, if a Lorentzian distance can be defined on such spectral triples, the question of the reconstruction of Lorentzian manifolds must be taken into consideration.\\

If one day Lorentzian spectral triples and Lorentzian distance in noncommutative geometry can be set, then the number of challenges will just be increased. The introduction of discrete Lorentzian spaces will be possible, and the coupling between general relativity and the standard model will be conceivable. A great question will be the establishment of a Lorentzian spectral action with an effective formulation, in order to deal with real physical problems.


\backmatter


\cleardoublepage
\addcontentsline{toc}{chapter}{Index}
\printindex


\ifodd\thepage
\newpage \hbox{} \vspace*{\fill} \thispagestyle{empty} 
\newpage \hbox{} \vspace*{\fill} \thispagestyle{empty} 
\else 
\newpage \hbox{} \vspace*{\fill} \thispagestyle{empty} 
\newpage \hbox{} \vspace*{\fill} \thispagestyle{empty} 
\newpage \hbox{} \vspace*{\fill} \thispagestyle{empty} 
\fi

\end{document}